\documentclass{article}
\pdfoutput=1

\usepackage{graphicx}
\usepackage{amsmath,amssymb,amsthm,mathtools}
\usepackage{dsfont}
\usepackage{enumitem}
\usepackage[round]{natbib}
\usepackage{tikz}
\usetikzlibrary{arrows.meta}
\usepackage{caption}
\usepackage{booktabs}
\usepackage{longtable}
\usepackage{subcaption}

\usepackage{pgfplots}
\pgfplotsset{compat=1.18}
\graphicspath{{figures/}}

\input{glyphtounicode}
\usepackage[T1]{fontenc}
\usepackage{sourceserifpro}
\usepackage{sourcecodepro}
\usepackage[hyphens]{url}

\usepackage{geometry}
\usepackage[onehalfspacing]{setspace}
\usepackage{microtype}
\AtBeginDocument{\frenchspacing}

\usepackage[x11names]{xcolor}
\usepackage{fancyhdr}
\fancypagestyle{plain}{\fancyhf{}}

\usepackage{titling}
\pretitle{\centering \bfseries\Huge}
\posttitle{\par \vskip 0.5em}

\date{August 2026}

\usepackage{etoolbox}
\AtBeginEnvironment{titlepage}{%
  \pagestyle{empty}%
  \setlength{\parindent}{0pt}%
}

\usepackage{titlesec}
\titleformat*{\section}{\centering\large\bfseries}
\titleformat*{\subsection}{\bfseries}
\titleformat{\paragraph}[runin]{\itshape}{}{}{}[.]
\titlelabel{\thetitle.\quad}

\usepackage[
  pdftex,
  hidelinks,
  hypertexnames=false,
  pdfpagemode=UseNone,
  pdfdisplaydoctitle=true
]{hyperref}

\theoremstyle{plain}
\newtheorem{proposition}{Proposition}
\newtheorem{lemma}{Lemma}
\newtheorem{definition}{Definition}

\newtheorem{assumption}{Assumption}

\newcommand{\E}{\mathbb{E}}

\newcommand{\Var}{\mathrm{Var}}
\newcommand{\Cov}{\mathrm{Cov}}

\newcommand{\ind}{\mathds{1}}

\hypersetup{pdftitle={Shared Bidding Algorithms and Competition: Evidence from Electricity Markets}}

\title{Shared Bidding Algorithms and Competition: Evidence from Electricity Markets}
\author{Nicolas Eschenbaum\thanks{Swiss Economics, Zürich, Switzerland. \href{mailto:nicolas.eschenbaum@swiss-economics.ch}{nicolas.eschenbaum@swiss-economics.ch}.}}

\begin{document}
\maketitle

\begin{abstract}
Competing firms increasingly delegate market decisions to
algorithms supplied by the same third-party providers. We study whether
a shared algorithm leads competitors to internalise one another's
profits, using data from the Australian National Electricity Market,
where batteries' bids are observed at 5-minute frequency and can
be linked to an autobidding provider. Bids constructed by
the same provider co-move, and do so more strongly after a disclosure
reform made the scarcity state easier to observe: the same
information that steers batteries towards efficient arbitrage also
synchronises the bids of competitors who share a provider. To separate
co-movement from joint profit maximisation,
we perform a conduct test by estimating each battery's dynamic value of stored energy and reclearing
the market under counterfactual bids. We find that batteries forgo 
profitable dispatch in the evening peak when
it would lower the profit of same-provider batteries owned by
rival firms. The estimated conduct parameter is close to one. 
But this effect arises only where a provider's share
of near-margin battery capacity exceeds roughly $30\%$. The identified conduct costs
consumers an annualised \$$5.5$~million given the battery fleet from 2025, 
and concentration analysis based on ownership would treat 
these batteries as independent competitors and miss the impact of shared 
autobidding software.
\end{abstract}

\medskip
\noindent{\small
}

\section{Introduction}

Firms increasingly rely on third-party algorithms to choose their prices or
bids. In many markets, firms' actions are now intermediated by 
an upstream layer of algorithm providers -- forecast vendors, data analytics providers, or bidding platforms that serve multiple downstream competitors. 
This has led to antitrust authorities raising concerns that the use of algorithms 
may lead to algorithmic collusion \citep[e.g.,][]{FTC2017, AutoriteBundeskartellamt2019, CMA2021Algorithms} 
and a growing literature studies whether algorithms learn anti-competitive
strategies \citep[e.g.,][]{CalvanoEtAl2020, AskerFershtmanPakes2024}.

In this paper, we provide empirical evidence on the effects that shared 
algorithms and algorithm provider concentration have on downstream market
conduct. We develop a hand-collected mapping of autobidding providers to 
utility-scale batteries in the Australian National Electricity Market (NEM). 
Autobidding software allows battery owners to specify high-level goals and 
constraints (such as their risk appetite) and the autobidding software 
then automatically calculates and submits the optimised bid stacks.\footnote{Very similar autobidding software is also commonly used in online ad auctions, see \cite{Aggarwal2024autobidding}.} 
We combine this with 5-minute bidding, price, and dispatch data of four
bidding zones in the NEM to construct a detailed panel of market behaviour
and upstream provider concentration in 2025/2026.

Electricity markets -- and utility-scale batteries in particular -- have
recently faced regulatory scrutiny stemming from concerns about algorithmic
collusion associated with the rapid expansion of storage
\citep{AEMC2024AlgorithmicCollusion, AFM2024AlgorithmicCollusion}.
Batteries interact across multiple concurrent markets at high
frequency -- features typically associated with elevated collusion
risk. At the same time, the complexity of the physical and market
constraints, together with the speed of bidding, mean that effective
operation requires algorithms, which are typically provided by
third-party firms.

We document a substantial degree of software provider concentration
in the NEM, with two systems -- Tesla Autobidder and Fluence
Mosaic -- operating two-thirds of known-provider battery capacity. The use of common
software produces substantial synchronisation in the level of bids
and in the timing of rebids among same-provider batteries. A July
2025 disclosure reform (ERI) made the aggregate regional state of
charge in the last 5-minute interval, and the per-battery levels
from the previous trading day, public. Combined with the observation
of same-provider batteries operating across multiple bidding regions,
this lets us measure how a stronger public scarcity signal changes
same-provider co-movement. The use of the same architecture, parameters,
and bidding logic may be expected to lead to 
more similar bidding responses as the public signal of the underlying 
scarcity state becomes stronger. We find that improved scarcity
observability indeed increases same-provider synchronisation. Both the 
level of bids and the timing of rebids are strongly correlated among 
same-provider batteries and this correlation increases post-ERI. The same
disclosure also improves the steering of batteries towards dynamic
arbitrage: bids rise in upper bid bands during scarce intervals and
fall during abundant intervals, and battery bids respond more
strongly to their own state of charge. 

We then perform a conduct test of bidding behaviour. We first estimate a 
dynamic Bellman function to determine the opportunity cost of the state of 
charge of a battery. We proceed to construct local deviations to individual 
battery bid stacks and reclear the market to determine the counterfactual 
market outcome.\footnote{The counterfactual reclearing co-optimises
energy and Frequency Control Ancillary Services (FCAS) markets simultaneously,
and we calculate the change in profit across markets.} We focus on deviations that move withheld capacity -- 
i.e., capacity in bid bands above the clearing price -- down to the marginal 
bid band at the given interval. If batteries are engaging in strategic withholding, 
this deviation should be profitable. Each deviation defines a revealed-preference 
moment inequality that tests whether behaviour can be rationalised by 
owner-level profit maximisation or instead requires internalisation of 
profits of same-provider batteries owned by other firms. 
A conduct parameter value of $\lambda=0$ corresponds to owner-level
incentives, while $\lambda=1$ corresponds to full internalisation of
same-provider, different-owner battery profits. Because the test
requires same-provider batteries operated by different owners with
relevant near-margin MW in the clearing interval, we can only run
the test for batteries served by the two large commercial providers
in our sample: Tesla Autobidder and Fluence Mosaic.

We estimate the conduct parameter as a continuous function of the
focal provider's near-margin share over evening-peak hours, when
battery market power is highest. The test rejects owner-level
incentives in energy-market deviations whenever the same-provider
share lies between roughly $30\%$ and $70\%$. Below about $30\%$ the test retains power but finds no
internalisation: the improvement in fit stays near zero and rises
steeply at $30\%$, pointing to a genuine concentration threshold.
Above about $70\%$ the provider's MW is again dominated by a single
owner and the multi-owner sample the test relies on becomes too
thin.\footnote{In finite samples
$\widehat\lambda$ in a low-power band can sit at the $[0,1]$
boundary because the objective is essentially flat. We treat
$\widehat\lambda=1$ as evidence of same-provider conduct only when
the moment objective falls materially below $Q_n(0)$ and the
subsampling confidence set excludes zero.} FCAS deviations show no
conduct signal in either provider.

We then provide a series of robustness exercises that rule out 
alternative interpretations of the finding. First, unobserved
owner-level costs cannot account for the results, because the required costs
would have to exceed the maximum movement we see in our Bellman
opportunity-cost estimates across numerical sensitivity checks, and
in two Tesla bands no finite per-MWh charge can rationalise the data.
Second, the results do not reflect hindsight. We date each operative
quantity bid to its submission timestamp and find that the conduct
signal is concentrated in rebids submitted within fifteen minutes of
dispatch, when five-minute predispatch forecasts are accurate, and
vanishes for bids lodged more than an hour ahead. Third, the results
are specific to the provider coalition. For each focal deviation we
construct placebo coalitions of different-provider batteries near the
margin in the same interval, matched to the same-provider coalition
on battery count and installed capacity so that they carry the same
mechanical price exposure. None of $500$ matched placebo coalitions
improves the fit of the conduct objective as much as the actual
same-provider coalition ($p\le0.010$). Finally, a positive control
shows that the estimator recovers full internalisation of profits
across batteries of the same owner, where joint-profit maximisation should
be recovered by construction.

The remaining explanation is thus that the conduct test identifies 
algorithmic profit internalisation. Two features of the setting 
establish that the identified conduct is economically meaningful. 
First, these batteries are non-negligible price setters. The 
reclearing analysis shows that a single battery's most profitable 
owner-only deviation moves the regional clearing price by $\$1.28$/MWh 
at the median and by more than $\$6$/MWh in the upper decile. These are 
substantial price effects of a single battery owner, given that the 
market also includes other dispatchable assets such as gas or hydro plants.
Second, the internalisation arises at the software provider rater than 
ownership level. 
In the identifying intervals the focal provider's near-margin capacity 
is spread across $2.2$ separately owned firms on average, and $83\%$ 
of the internalised profit spillover accrues to batteries owned by a 
different firm than the one deviating -- i.e., firms that ownership-based 
concentration analysis treats as independent competitors. To the best 
of our knowledge, we are the first to provide empirical evidence on 
the effect of algorithm-provider concentration on downstream market conduct

To assess the dollar size of the identified conduct, we use the
reclearing engine to compute, for each focal battery-interval in
the identifying range, the regional clearing-price difference the
owner-only counterfactual deviation would have generated. The
implied price difference is about \$$1.28$/MWh at the median and
\$$2.70$/MWh on average, corresponding to an annualised consumer
cost of approximately \$$5.5$~million ($\$3.1$m attributable to
Tesla Autobidder and \$$2.4$m to Fluence Mosaic). This is a local
partial-equilibrium estimate that holds all other bid stacks fixed
and does not solve for a new market equilibrium under alternative
conduct.\footnote{The utility-scale battery fleet has since grown
substantially in the NEM and continues to grow, so the same
identifying signal applied to a larger fleet would imply a
correspondingly larger cost.}

\paragraph{Related Literature} There is a large and rapidly growing literature on algorithmic collusion.
Seminal works have shown how reinforcement-learning algorithms may converge to
supra-competitive pricing and policies that show patterns of collusive behaviour
\citep{CalvanoEtAl2020,AskerFershtmanPakes2024,BanchioMantegazza2024}.
Related work has documented the emergence of algorithmic collusion in simulation
studies across various scenarios and settings.\footnote{See, for example,
\citet{Eschenbaumetal2022,Klein2021,WaltmanKaymak2008,KimbroughMurphy2009,
CooperHomemDeMelloKleywegt2015,HansenMisraPai2021a,KastiusSchlosser2022,
SanchezCartasKatsamakas2022,MeylahnDenBoer2022,AbadaLambinTchakarov2024}.}
\citet{AbadaLambin2023} in particular study the operation of storage assets in electricity
markets and consider the Tesla Autobidder software -- one of the two main players
in our setting -- as their motivating example.
Recent papers further emphasize that the relevant issue is not only
autonomous learning, but also the delegation of decision-making, information
processing, and implementation to a common third-party algorithm.\footnote{See \citet{Harrington2021,Harrington2025hub,HarringtonOrtner2025,
SugayaWolitzky2025disclosure,AleksenkoMiklosThal2026}.}

Our work is most closely related to \citet{CalderWangKim2026}, 
who similarly use hand-collected data on
algorithmic pricing adoption in multifamily rental markets and 
combine reduced-form evidence with a structural conduct
test to distinguish own-profit from same-provider joint-profit incentives.
 In contrast to their work, we do not observe algorithm adoption 
over time, because in our setting autobidding is already widespread. 
Instead, we observe bid stacks themselves and exploit the fact that we 
can evaluate counterfactual market outcomes using a very accurate clearing engine. 
Rather than estimating the effect of algorithm adoption, our conduct 
test estimates whether observed battery bids are rationalised by owner-level 
incentives, and whether violations are systematically reduced when 
same-provider profit effects are internalised. We then relate the 
estimated effects to the provider's share of near-margin battery 
capacity in the bidding interval. This allows us to ask whether 
internalisation of same-provider profits becomes systematically 
more informative as provider concentration rises, and whether there 
exists a threshold level of concentration at which this becomes the case.

The only other empirical study of
the effect of algorithm adoption on market conduct that we are aware of is
\citet{AssadClarkErshovXu2024}, who show that adoption of pricing
algorithms by competing gasoline stations leads to increased market
prices. A related empirical literature studies repricing
software in online retail markets, where the algorithm reprices an
existing posted price in response to rivals rather than constructing
the price from primitives, and shows that the use of repricers can
increase market prices \citep{BrownMacKay2023, Musolff2025algorithmic}.
We contribute evidence from a market in which the algorithmic decision
itself -- a submitted bid stack across ten price bands -- is directly
observed, and the first empirical evidence on \emph{bidding}
algorithms rather than pricing or repricing software.

Moreover, two effects from the recent theoretical literature on
common-algorithm providers arise in our setting. 
We show this with a simplified, stylised model in which common
autobidding providers map scarcity signals into battery bid stacks. 
As in \citet{HarringtonOrtner2025},
provider concentration amplifies the aggregate price response to
common scarcity information. And similarly to \citet{AleksenkoMiklosThal2026},
consumer harm from common-information pricing can be decomposed into a
cross-seller correlation channel and an across-state dispersion
channel. 
The intuition is as follows. Because batteries face a dynamic opportunity cost of discharge, a
higher scarcity estimate raises the opportunity cost and thus the
bids, leading to correlated bids of same-provider batteries. As a
result, the market outcome becomes a function of provider
concentration which amplifies the aggregate price impact. The more 
concentrated the software provider layer is, the stronger the  
market price moves with the scarcity signal. Consumers then face the 
highest prices during those time intervals in which their consumption is 
highest, leading to a reduction in consumer welfare. Our simple model
shows how these effects can arise in a setting with bids and
dynamic storage.\footnote{Our model is intentionally kept simple, and 
we abstain from a complete model of the repeated auction setting with capacity 
constraints and dynamic storage.} 

Our paper speaks to an active competition-policy debate. Competition
authorities have warned that pricing and bidding algorithms may facilitate
coordination when competitors rely on a common software intermediary, share
competitively sensitive data through that intermediary, or allow software
defaults and implementation architecture to substitute for independent
decision-making \citep{FTC2017,Monopolkommission2018,
AutoriteBundeskartellamt2019,CMA2018PricingAlgorithms,CMA2021Algorithms,
AFM2024AlgorithmicCollusion,HovenkampSchrepel2026}. 

This issue has become
a concern in electricity markets. The \citet{AEMC2024AlgorithmicCollusion} argues
that algorithmic bidding is already pervasive in the NEM, that the NEM is
increasingly vulnerable to algorithmic collusion, and that policy needs to
balance the efficiency gains from AI against the risk that more information and
more capable algorithms make coordination easier to sustain. We directly contribute 
to this debate by measuring the provider layer, and testing whether 
same-provider profit effects become informative precisely when a provider controls a large
share of the near-margin battery capacity. 
In our setting, battery owners are the nominal decision-makers. 
However, the complexity of calculating the optimal bids under the complex market constraints 
imposed by the market regulator, as well as the physical constraints of the 
battery and high frequency of bidding makes it likely that battery owners rarely 
override the bids that the autobidding software calculates. 

Our conduct test builds on structural conduct testing and
moment-inequality methods. In differentiated-products markets, conduct is often
tested by estimating demand and recovering marginal costs under alternative conduct models, 
and testing whether the resulting cost shocks satisfy exclusion restrictions
\citep{BerryHaile2014,BackusConlonSinkinson2021}. This is the route followed
by \citet{CalderWangKim2026}. In our setting, short-run load is approximately perfectly inelastic, 
so there is no differentiated-products demand system whose price elasticities can be used to infer
markups in the usual way. The central clearing mechanism instead
offers a different possibility: exact reclearing of feasible bid-stack deviations. 
The resulting revealed-preference inequalities builds on the moment-inequality approaches based on optimality or
no-profitable-deviation restrictions \citep{BajariBenkardLevin2007,PakesPorterHoIshii2015,HoPakes2014}.

The object our conduct test estimates -- a weight on rival profits in a firm's effective objective -- 
is also the central object of the common-ownership literature, which asks whether diversified shareholders 
lead firms to internalise their competitors' profits \citep{AzarSchmalzTecu2018,BackusConlonSinkinson2021}. 
In our setting this spillover can be measured directly, because reclearing returns the dollar profit effect of a deviation on 
every rival battery. More importantly, the conduct parameter (or weight) that we estimate is associated 
with a different term: not who owns the asset, but who runs the bidding algorithm.

A related industrial-organisation literature studies strategic bidding and
storage market power. Supply-function competition is the canonical framework
for strategic bidding in uniform-price electricity auctions
\citep{KlempererMeyer1989,Vives2011,Holmberg2008}, and repeated uniform-price
auctions can facilitate tacit collusion by weakening deviation incentives
\citep{Fabra2003}. Empirical work has used detailed bidding data to study
strategic behaviour and dynamic constraints in electricity markets
\citep{Reguant2014}. A newer literature studies how storage changes market
outcomes, investment incentives, and the integration of renewable generation
\citep{Butters2021,AndresCerezoFabra2023}. Our contribution is to study
utility-scale batteries as strategic bid-stack participants whose actions are
increasingly intermediated by common autobidding providers.

Lastly, our analysis of bidding behaviour is also related to the operational literature that models batteries as dynamic optimisers that choose
charge, discharge, and ancillary-service participation subject to
state-of-charge, efficiency, capacity, degradation, and market constraints.
This is directly relevant in the NEM, where batteries co-optimise energy and
FCAS participation at high frequency. Closest to our work is \citet{KarimiArpanahiPourmousaviMahdavi2024}, who estimate battery energy storage system (BESS) operating
characteristics from public NEM data and solve a receding-horizon scheduling
problem across energy and FCAS markets, including the FCAS trapezium, while
\citet{PrakashBruceMacGill2025} study the role of future price information for
storage scheduling in the NEM. We build on their formulation of batteries as constrained
dynamic optimisers for our analysis and the Bellman estimation, but study the submitted bid stacks themselves and the
provider layer that helps construct them.

The paper proceeds as follows. Section~\ref{sec:battery-bidding} describes
battery bid construction in the NEM and discusses a stylised model that guides
the empirical analysis. Section~\ref{sec:empirics} constructs the empirical
panel, documents synchronisation and state-responsive bidding around the ERI
reform, estimates the dynamic value of stored energy, and implements the
counterfactual-reclearing conduct test. The final section concludes.

\section{Battery Bidding and Autobidder Providers}\label{sec:battery-bidding}

This section defines how battery optimisation maps into observed bids and why
provider identity enters that mapping. We proceed in two steps. First, we
describe the dynamic battery problem and the bid stack submitted to the NEM, as 
well as the role of autobidders.
Second, we introduce a stylised model that isolates the main mechanisms studied in
the empirical analysis.

Throughout, we denote batteries by $i,j$, owners by $o$, providers by $k$,
regions by $r$, intervals by $t$, and bid bands by $b$. Upper-case Latin
letters denote stocks, prices, and counts ($S$ stored energy, $P$ price,
$V$ the continuation value, $K$ number of providers); bold
lower-case letters denote vectors of the corresponding scalars ($\mathbf b$
a bid stack). Greek letters denote parameters, unobservables, and profit
objects ($\theta$ scarcity, $\lambda$ the conduct parameter, $\pi$ a
single-battery per-period profit flow, $\Pi$ an owner-portfolio payoff
inclusive of continuation value); calligraphic capitals denote sets and
information structures ($\mathcal B$ the feasible stack, $\mathcal I$ an
information set); and hats denote estimators ($\widehat\lambda$,
$\widehat m_{it}$). 

\subsection{Utility-Scale Batteries and Bid Construction}\label{subsec:battery-optimization}

Batteries' ability to shift generated energy across the trading day, i.e., dynamic
arbitrage, becomes very influential in electricity markets with substantial
photovoltaic (PV) generation. As a result, storage is now one of the fastest-expanding 
power-sector technologies.\footnote{The IEA
reports 108~GW of new battery storage capacity deployed in 2025, around 40\%
more than in 2024, with installed capacity eleven times higher than in 2021
\citep{IEA2026GlobalEnergyReview}. In the United States, developers plan
24~GW of utility-scale battery additions in 2026, with Texas alone accounting
for 12.9~GW of planned additions \citep{EIA2026CapacityAdditions}.} For example, in ERCOT 
(the energy market in Texas), installed battery capacity reached 14.1~GW by July~2025, almost triple its
early-2023 level, and batteries discharged more than 7~GW during the evening
ramp for the first time \citep{ERCOTMonthlyJuly2025}. The NEM is smaller but
follows the same trajectory: AEMO reports a record 64~GW connection pipeline in
the December 2025 quarter, with battery storage accounting for 46\% of projects
and hybrid solar-battery projects for another 19.7\%
\citep{AEMO2026DecemberQuarterPipeline}. As storage becomes a larger share of flexible capacity, 
the way batteries 
construct bid stacks can increasingly affect both arbitrage efficiency and price formation 
in the evening hours when PV production subsides \citep{Butters2021, AndresCerezoFabra2023}.

Formally, a utility-scale battery is a constrained dynamic asset. 
Time is discrete: in each 5-minute NEM dispatch interval the
battery chooses charge, discharge, and FCAS enablement subject to energy,
power, and market-validation constraints.\footnote{FCAS stands for
Frequency Control Ancillary Services. In the NEM, AEMO procures these services
to keep system frequency close to 50~Hz when realised generation and load differ
from expected levels. Raise services increase net injection or reduce load when
frequency is too low; lower services reduce net injection or increase load when
frequency is too high. These services include both continuous regulation and
contingency response over short horizons. Batteries are well suited to FCAS
because inverter controls can change charge or discharge quickly, but FCAS
enablement reserves headroom or footroom and is therefore jointly constrained
with energy dispatch.} Let $S_{it}$ denote stored energy of battery
$i$ at interval $t$.\footnote{$S_{it}\in[0,\bar S_i]$, where $\bar S_i$
is the battery's energy capacity.} The shadow value of $S_{it}$ is the
intertemporal opportunity cost of discharge: using stored energy in
the current interval lowers the option value of using the same energy
later in the peak.

Figure~\ref{fig:battery-role} shows the two features of battery bidding that
matter for the analysis. Panel~A plots the daily arbitrage cycle. Solar generation
compresses residual demand during the day, when batteries tend to charge, while
the evening ramp creates the high-value discharge window. This makes battery
bidding a finite-horizon inventory problem within the trading day. Panel~B
measures the corresponding market-power state. Batteries matter most for price
formation when battery capacity is close to the regional clearing margin, which
occurs primarily during the peak evening intervals.

Figure~\ref{fig:battery-role} focuses on the energy market because
the energy clearing margin is the simplest way to show the dynamic
arbitrage and price-setting logic. Our
empirical analysis, however, evaluates operation and profits across both energy
and FCAS markets. This distinction matters because FCAS enablement is jointly
constrained with energy dispatch: capacity allocated to one market changes the
headroom or footroom available in the other.

\begin{figure}[t]
\centering
\caption{Battery operation and the market-power state}
\label{fig:battery-role}
\begin{subfigure}{0.49\linewidth}
    \centering
    \includegraphics[width=\linewidth]{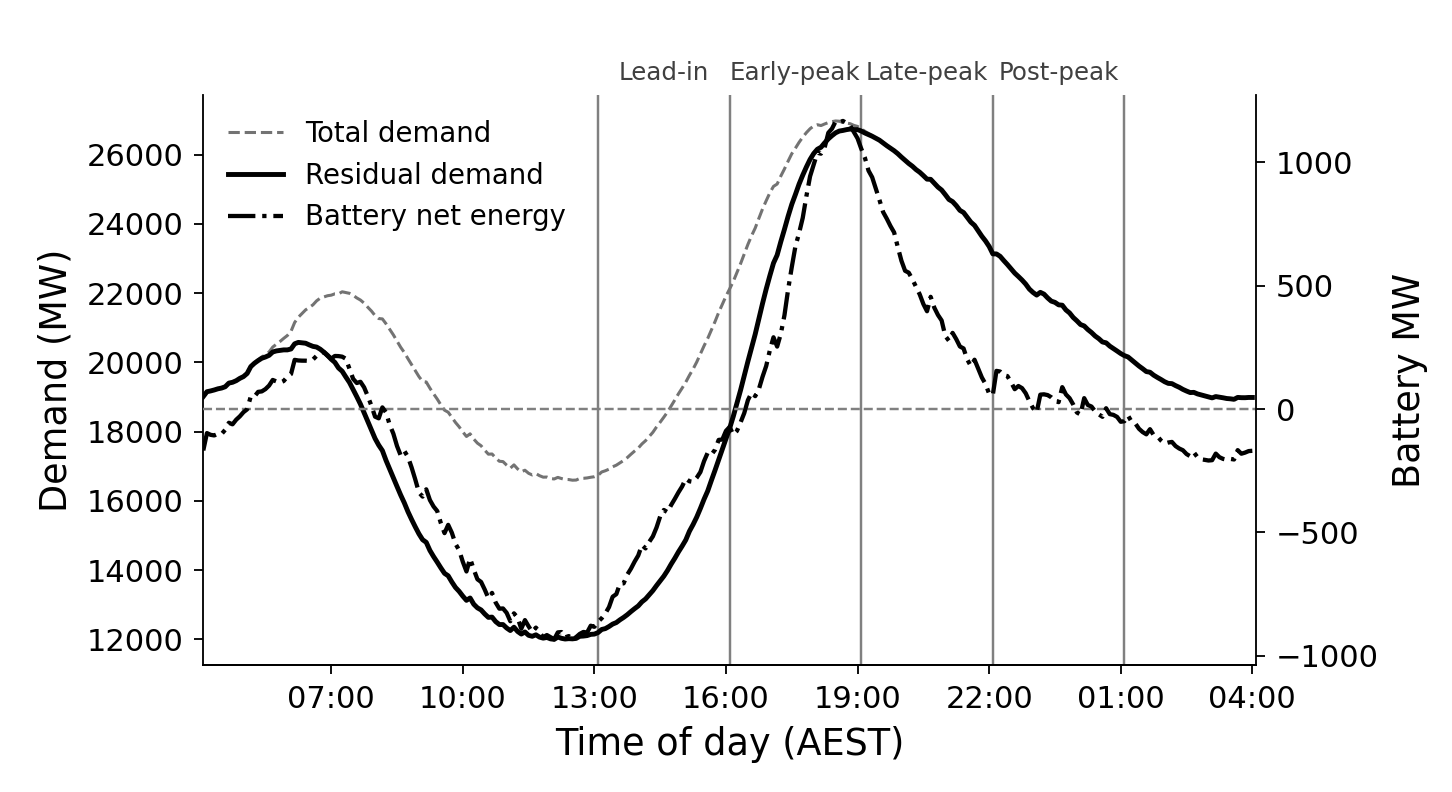}
    \caption{Daily arbitrage}
\end{subfigure}
\hfill
\begin{subfigure}{0.49\linewidth}
    \centering
    \includegraphics[width=\linewidth]{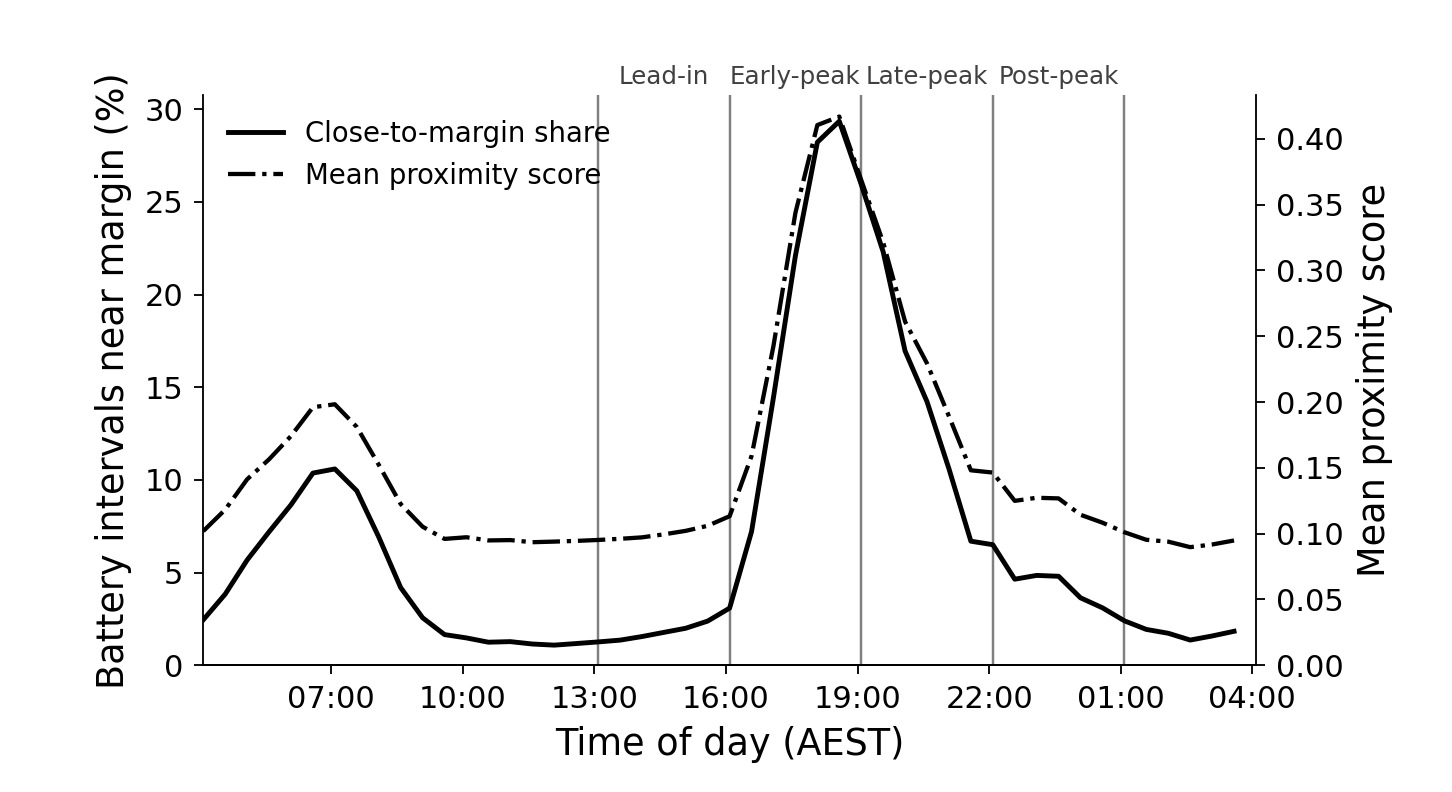}
    \caption{Near-margin battery capacity}
\end{subfigure}
\vspace{2pt}
{\footnotesize \textit{Notes:} Panel~A plots average demand, residual
demand, and net battery output over the day. Panel~B measures when
battery capacity is close to the regional clearing margin.}
\end{figure}

The NEM implements this decision through regional uniform-price auctions that
clear every 5 minutes. Within each region, AEMO dispatches the least-cost
feasible stack and sets the regional reference price (RRP) at the marginal accepted
offer. A NEM energy bid is lodged for a trading day running from 4{:}00~a.m. on
day~$t$ to 4{:}00~a.m. on day~$t+1$ and contains 288 5-minute dispatch
intervals. For each unit, participants set up to ten monotonically
increasing price bands and, for each interval, the MW available in each band
together with technical parameters such as maximum availability and ramp rates.
For scheduled bidirectional batteries, bids are submitted separately for the
load and generation sides, so a battery effectively has ten charge bands and ten
discharge bands. 

The timing of bids separates price choices from quantity choices. Price bands
for trading day~$t$ must be lodged by 12{:}30~p.m.~AEST on day~$t-1$ and are
then fixed for the trading day. Quantities and related technical parameters may
be rebid after that deadline through the relevant dispatch interval, subject to
the NEM rebidding rules: each rebid must provide a brief, specific, and
verifiable reason and the time of the event giving rise to the rebid. For a
given interval, the operative bid is the latest bid captured for dispatch for
that interval.\footnote{For technical details, see AEMO's bidding
specification, spot market operations timetable, and bidirectional-unit
validation material.} The economically relevant object in our application is
the generation-side stack, since it governs discharge offers into the regional
energy auction. Figure~\ref{fig:bid-stack-schematic} illustrates a battery's
bid stack. For a given clearing price $P^*$, the accepted
energy discharge of battery $i$ is the cumulative accepted band quantity.
Moving capacity from high bands towards the clearing band raises the
probability of dispatch, but also uses scarce stored energy and can lower the
market price.

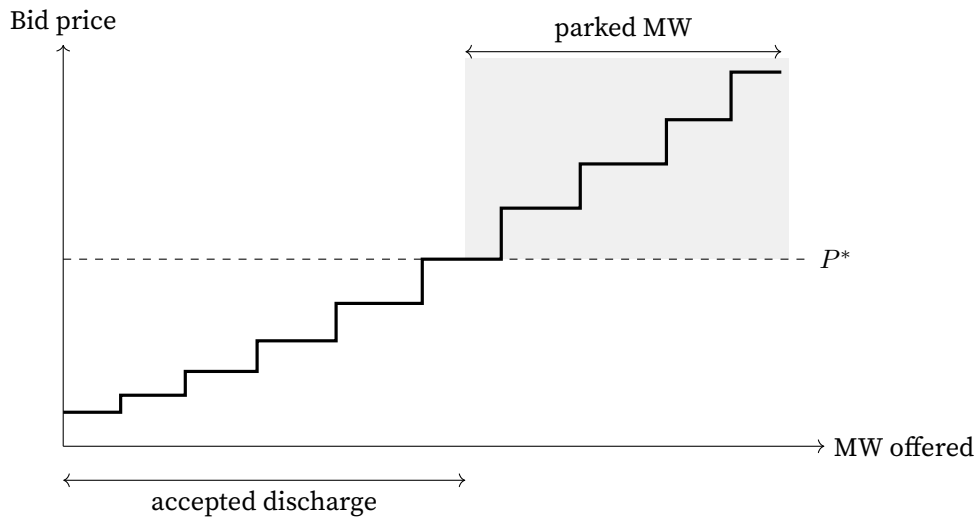
\begin{figure}[t]
\centering
\caption{Schematic generation-side bid stack with ten price bands}
\label{fig:bid-stack-schematic}
\begin{tikzpicture}[x=0.095cm,y=0.045cm]
    \fill[gray!12] (56,55) rectangle (101,114);
    \draw[->] (0,0) -- (106,0) node[right]{MW offered};
    \draw[->] (0,0) -- (0,118) node[above]{Bid price};
    \draw[very thick] (0,10) -- (8,10) -- (8,15) -- (17,15)
        -- (17,22) -- (27,22) -- (27,31) -- (38,31)
        -- (38,42) -- (50,42) -- (50,55) -- (61,55)
        -- (61,70) -- (72,70) -- (72,83) -- (84,83)
        -- (84,96) -- (93,96) -- (93,110) -- (100,110);
    \draw[dashed] (0,55) -- (104,55) node[right]{$P^*$};
    \draw[<->] (0,-10) -- (56,-10) node[midway,below]{accepted discharge};
    \draw[<->] (56,116) -- (100,116) node[midway,above]{parked MW};
\end{tikzpicture}\\
\vspace{2pt}
{\footnotesize \textit{Notes:} The clearing price $P^*$ is set by the
marginal accepted offer. Capacity above $P^*$ is withheld from energy
dispatch in that interval unless rebid or repriced.}
\end{figure}

The feasible stack is joint across markets. The difference between energy and FCAS bids is that the quantity offered in FCAS markets is 
enabled response capacity rather than energy dispatch. For batteries, FCAS
enablement is feasible only conditional on the energy operating point, because
raise and lower services require headroom and footroom. AEMO validates these relationships through
service-specific trapezium constraints, which we describe formally in Appendix~\ref{app:fcas-trapezium}.
The submitted stack is therefore the
output of a complex, high-frequency constrained dynamic optimisation problem.

Finally, an autobidder maps forecasts, telemetry, market data, owner settings,
and risk preferences into this feasible stack. Public descriptions make clear
that ``autobidder'' is a class of bid-construction and dispatch-control tools,
not a single standardised product. AEMO describes auto-bidding systems as
software that automatically carries out bids and rebids under pre-set
parameters, typically using 5-minute pre-dispatch and dispatch data, so that
auto-rebidding often occurs less than one hour ahead of dispatch.\footnote{See
AEMO, \emph{2025 General Power System Risk Review Report -- Draft},
Section~5.9. Public document:
\url{https://www.aemo.com.au/-/media/files/stakeholder_consultation/consultations/nem-consultations/2024/2025-general-power-system-risk-review/2025-draft-gpsrr-report.pdf}.}
Commercial offerings used or marketed for Australian batteries
include Tesla's \emph{Autobidder}, Fluence's \emph{Mosaic}, W\"artsil\"a's
\emph{GEMS}/\emph{IntelliBidder}, HARD Software's \emph{Infolite Auto
Bidding}, Habitat Energy's \emph{EVOLVE} optimisation service, and newer
NEM-specific products such as OptiGrid.

The relevant decision-maker remains the registered market participant, which may
bid internally or contract optimisation services to a software or trading
provider. Providers are therefore heterogeneous: some are stand-alone software
vendors, some are integrated with the plant-control stack, some combine software
with a trading desk, and some are in-house systems run by large asset owners.
What they share is the same basic role in bid construction. The provider takes
telemetry and market data as inputs, combines them with owner-supplied settings
such as SoC limits, degradation or warranty constraints, FCAS participation
choices, contract positions, and risk tolerances, and then forecasts prices and
co-optimises charge, discharge, and offer construction for an asset or customer
portfolio. 

However, the complexity of calculating the optimal bid stacks over a trading
day and across the sequence of 5-minute auctions -- together with the
technical difficulty of submitting them to the market operator -- makes it
difficult for an owner to override what the software has determined and
submitted. An owner's agency over individual bids is therefore limited, and
the bids the software constructs are unlikely to be changed.

This delegation creates the provider-layer object we study. In the empirical
mapping, 42 usable batteries can be linked to 11 identified providers across
21 provider-region combinations; the two largest commercial providers are Tesla
Autobidder and Fluence Mosaic.\footnote{Appendix~\ref{app:provider-id}
describes the provider mapping and the confidence classification.} A common
provider can create correlated bid construction whenever clients share
forecasts, update schedules, or implementation architecture, and can also be the
level at which same-provider profit effects are internalised.

\subsection{A Stylised Model}\label{subsec:mechanisms}

This section outlines a stylised model to explain how common providers affect batteries' bids. The aim is not to reproduce the NEM auction setting, but to isolate the key mechanisms of the setting. We start from the primitive objects of the empirical analysis and discuss three implications. \textit{First}, autobidding makes bids load on provider-driven information, because the provider's algorithm turns the available information (private or public) into the bid stacks. \textit{Second}, this common provider information channel generates bid co-movement, and it has a larger price effect when the provider's share is higher. Thus, even if we observe co-movement of same-provider batteries or stronger price effects in intervals where capacity is more concentrated among providers, we cannot conclude this implies coordinated conduct. \textit{Third}, however, if the provider is implementing owner-level optimisation, then observed bid stacks must satisfy owner-level no-profitable-deviation inequalities; if they do not, same-provider profit internalisation is an alternative conduct hypothesis. This will be the basis for our conduct test.

Formally, consider a single bidding region $r$ -- one of the NEM's regional markets -- and a 5-minute dispatch interval $t$. Batteries in the region are indexed by $i\in I_r$.
Battery $i$ is owned by $o_i$ and uses autobidding provider $k_i$. Let
$N_k:=\{i:k_i=k\}$ denote the set of batteries using provider $k$. Battery
$i$ enters the interval with stored energy $S_{it}$ and submits a feasible
multimarket bid stack \(\mathbf b_{it}\in\mathcal B_i(S_{it},X_{rt})\), where
$X_{rt}$ collects public market states and technical constraints. The set
$\mathcal B_i(\cdot)$ includes the energy bid bands, FCAS bid bands, power
limits, state-of-charge restrictions, and the FCAS trapezium constraints
described in Appendix~\ref{app:fcas-trapezium}. Given the vector of submitted
stacks $\mathbf b_t=(\mathbf b_{it})_{i\in I_r}$, the market clearing engine
maps bids and states into prices, dispatch, FCAS enablement, and current-period
profits.

Let $\{\theta_{r\tau}\}_{\tau\ge t}$ denote the regional scarcity process. The
current realisation $\theta_{rt}$ summarises contemporaneous tightness of the
regional stack, while the future path
$(\theta_{r,t+1},\theta_{r,t+2},\ldots)$ determines the continuation value of
stored energy. Provider $k$ observes an information set $\mathcal I_{krt}$,
including public information and provider-specific forecasts. We write
\(m_{krt}:=\E[\theta_{rt}\mid \mathcal I_{krt}]\) as shorthand for the
provider's current-scarcity assessment. The same information set also induces
beliefs over future scarcity, which enter continuation values below. 

We now discuss how this conditioning on the provider-determined information set leads to co-movement due to a ``responsive pricing'' (or ``responsive bidding'') effect  \citep{CalderWangKim2026}. That is, same-provider batteries will show co-movement even if the provider's algorithm maximises battery-owner profits. 

\paragraph{Owner-level responsive autobidding}
Let $\Pi_{ot}$ denote owner $o$'s current plus continuation payoff,
including current energy and FCAS profits and the continuation value of stored
energy after dispatch. Let $\mathcal I_{ot}$ denote the information available to
the owner's portfolio, including the provider information sets for the owner's
batteries, and let $\E_{ot}$ denote expectation under the owner's beliefs over
scarcity and rival bid strategies.\footnote{Equivalently, one could write an
owner strategy $\sigma_o:\mathcal I_{ot}\to\mathcal B_o(S_t,X_{rt})$ and an
expectation $\E_{\mu_{ot},\sigma_{-o}}[\cdot]$ under owner $o$'s beliefs
$\mu_{ot}$ and rival strategies $\sigma_{-o}$. We keep this notation implicit
because the empirical test only uses the resulting no-profitable-deviation
restrictions.} Then owner-level responsive bidding solves
\begin{equation}
\mathbf b^R_{ot}
\in
\arg\max_{\mathbf b_o\in\mathcal B_o(S_t,X_{rt})}
\E_{ot}\!\left[
\Pi_{ot}(\mathbf b_o,\mathbf b_{-o,t};S_t,X_{rt},\{\theta_{r\tau}\}_{\tau\ge t})
\mid \mathcal I_{ot}
\right].
\label{eq:model-responsive-objective}
\end{equation}
Here $\mathcal B_o$ is the product of the feasible bid-stack sets for the
owner's batteries, and $\mathbf b_{-o,t}$ is shorthand for the rival stacks
induced by the owner's beliefs about rival bidding rules. Note that equation~\eqref{eq:model-responsive-objective}
is not an equilibrium characterisation, but simply states that common software may change information, forecasts, and implementation, but it should not change which profits are maximised.

The responsive channel then follows from the dynamic opportunity cost of discharge.
A higher provider scarcity assessment raises
the expected continuation value of stored energy. Since raising the offer price
reduces the probability or extent of current dispatch, the owner-level bid rule
increases in the provider's scarcity assessment:
\begin{equation}
\frac{\partial b_i^R}{\partial m_{k_i,rt}}>0.
\label{eq:model-positive-response}
\end{equation}
The result is formally derived in Lemma~\ref{lem:app-shadow-cost} in Appendix~\ref{app:formal-model}. 
Assuming monotone storage value and scarcity posterior (Assumptions~\ref{ass:app-monotone-storage}--\ref{ass:app-monotone-posterior}), 
the provider's posterior shadow cost of discharge is weakly increasing in the scarcity assessment $m_{krt}$, 
and the local first-order condition then yields $\partial b_i^R/\partial m_{krt}>0$.

However, this gives rise to co-movement among same-provider batteries. Linearizing the responsive bid around the observed state gives
\begin{equation}
b_{it}^R
\approx
\bar b_{it}
+
\psi_i(S_{it},X_{rt})
+
\phi_{it}m_{k_i,rt},
\qquad
\phi_{it}:=\frac{\partial b_i^R}{\partial m_{k_i,rt}}>0.
\label{eq:model-linearized-policy}
\end{equation}
For two batteries $i$ and $j$ using the same provider $k$, conditional on
public controls, both bids load on the same provider-level object $m_{krt}$.
Batteries using different providers load on different provider-level objects.
Same-provider co-movement is therefore stronger whenever provider forecasts,
preprocessing, or implementation rules create variation shared
within provider but not across providers. This mechanism generates synchronised
and state-contingent bids without any same-provider profit internalisation.

\paragraph{Concentration effects of third-party providers}
This same co-movement effect also leads to the market outcome becoming 
a function of the degree of concentration among third-party providers. 
We develop a simple static
equilibrium benchmark, in the spirit of \citet{HarringtonOrtner2025},
to formally prove this point.
Appendix~\ref{app:formal-model} solves an affine provider-posterior
equilibrium whose aggregate response to a common scarcity
signal is governed by a strategic concentration index $\Lambda(x)$
that depends on the vector of provider shares
$x=(x_1,\ldots,x_K)$, the correlation across provider scarcity
assessments, an own-action slope, and a strategic-complementarity
coefficient (equation~\eqref{eq:app-affine-solution} in the
appendix). Importantly, the aggregate response is increasing in $\Lambda(x)$ and
$\Lambda(x)$ is Schur-convex in provider shares, so concentrating
MW on fewer providers makes the common scarcity component
more pronounced. Formally, Proposition~\ref{prop:app-concentration} in the appendix 
shows for the affine static equilibrium that under any majorization-increasing (concentration-increasing) 
change in the provider-share vector $x$, (i) the unconditional mean action $\alpha_t$ is invariant, 
(ii) the common-response coefficient $\mathcal S_t$ is weakly increasing, 
and (iii) both the common component of aggregate variance $\rho_t\mathcal V_t\mathcal S_t^2$ 
and total variance $\Var(A_t)$ weakly increase, strictly along a dominant-provider family.

Figure~\ref{fig:model-concentration-mechanism} illustrates this 
logic using one scarcity path under three provider structures and the 
simple affine equilibrium structure developed in Appendix~\ref{app:formal-model}. 
It depicts the response to a mean-zero scarcity path with a monopoly provider  
and three or ten symmetric providers, respectively. Because same-provider batteries 
condition on the same posterior, their bids co-move; as more capacity
is attached to the same posterior, the aggregate response becomes steeper.
This leads to stronger conditional response to common scarcity information and 
thus (relatively) higher bids during scarce periods and (relatively) lower bids 
during periods with abundance.

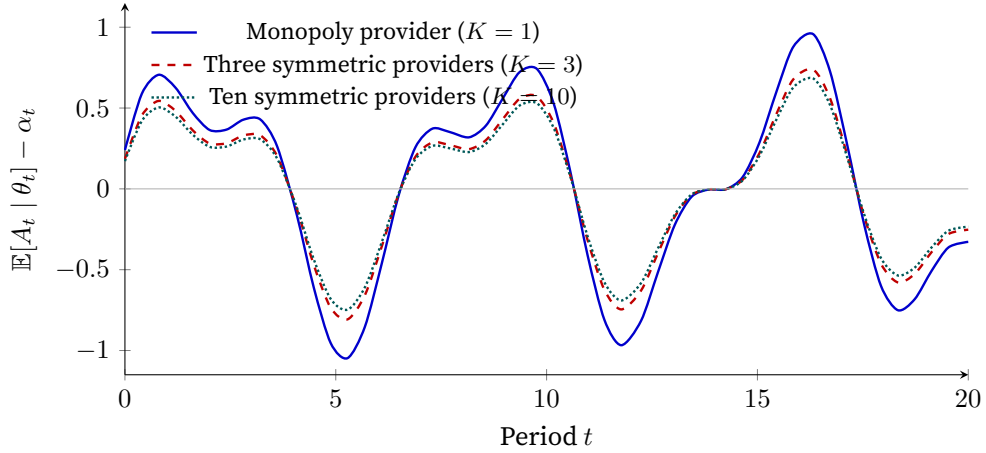
\begin{figure}[t]
\centering
\caption{Response to a common scarcity path under three provider structures}
\label{fig:model-concentration-mechanism}
\begin{tikzpicture}
\begin{axis}[
    width=0.85\linewidth,
    height=6.5cm,
    xmin=0, xmax=20,
    ymin=-1.15, ymax=1.15,
    xlabel={Period $t$},
    ylabel={$\E[A_t\mid \theta_t]-\alpha_t$},
    ytick={-1,-0.5,0,0.5,1},
    xtick={0,5,10,15,20},
    legend style={draw=none, fill=none, font=\small,
      at={(0.02,0.98)}, anchor=north west},
    axis lines=left,
    clip=true,
    every axis plot/.append style={line width=0.9pt, no markers, smooth}
]

\addplot[blue!80!black] coordinates {
    (0.00,0.2403) (0.40,0.5880) (0.80,0.7057) (1.20,0.6258)
    (1.61,0.4691) (2.01,0.3640) (2.41,0.3667) (2.81,0.4281)
    (3.21,0.4295) (3.61,0.2627) (4.08,-0.1644) (4.48,-0.6120)
    (4.88,-0.9542) (5.28,-1.0471) (5.69,-0.8536) (6.09,-0.4629)
    (6.49,-0.0413) (6.89,0.2579) (7.29,0.3723) (7.69,0.3527)
    (8.16,0.3189) (8.56,0.3910) (8.96,0.5589) (9.36,0.7224)
    (9.77,0.7419) (10.17,0.5267) (10.57,0.0985) (10.97,-0.4082)
    (11.37,-0.8097) (11.77,-0.9672) (12.24,-0.8162) (12.64,-0.5162)
    (13.04,-0.2226) (13.44,-0.0484) (13.85,-0.0038) (14.25,-0.0026)
    (14.65,0.0731) (15.05,0.2878) (15.45,0.6006) (15.85,0.8786)
    (16.32,0.9573) (16.72,0.7285) (17.12,0.2830) (17.53,-0.2197)
    (17.93,-0.6003) (18.33,-0.7506) (18.73,-0.6824) (19.13,-0.5074)
    (19.53,-0.3624) (20.00,-0.3278)
};
\addlegendentry{Monopoly provider ($K=1$)}

\addplot[red!75!black, dashed] coordinates {
    (0.00,0.1853) (0.40,0.4533) (0.80,0.5441) (1.20,0.4825)
    (1.61,0.3616) (2.01,0.2806) (2.41,0.2827) (2.81,0.3300)
    (3.21,0.3311) (3.61,0.2025) (4.08,-0.1268) (4.48,-0.4718)
    (4.88,-0.7356) (5.28,-0.8072) (5.69,-0.6580) (6.09,-0.3569)
    (6.49,-0.0318) (6.89,0.1988) (7.29,0.2871) (7.69,0.2719)
    (8.16,0.2458) (8.56,0.3014) (8.96,0.4309) (9.36,0.5569)
    (9.77,0.5719) (10.17,0.4061) (10.57,0.0759) (10.97,-0.3147)
    (11.37,-0.6243) (11.77,-0.7457) (12.24,-0.6292) (12.64,-0.3979)
    (13.04,-0.1716) (13.44,-0.0373) (13.85,-0.0030) (14.25,-0.0020)
    (14.65,0.0563) (15.05,0.2218) (15.45,0.4630) (15.85,0.6774)
    (16.32,0.7380) (16.72,0.5616) (17.12,0.2182) (17.53,-0.1693)
    (17.93,-0.4628) (18.33,-0.5787) (18.73,-0.5261) (19.13,-0.3912)
    (19.53,-0.2794) (20.00,-0.2527)
};
\addlegendentry{Three symmetric providers ($K=3$)}

\addplot[teal!70!black, densely dotted] coordinates {
    (0.00,0.1715) (0.40,0.4197) (0.80,0.5037) (1.20,0.4467)
    (1.61,0.3348) (2.01,0.2598) (2.41,0.2617) (2.81,0.3055)
    (3.21,0.3065) (3.61,0.1875) (4.08,-0.1174) (4.48,-0.4368)
    (4.88,-0.6810) (5.28,-0.7473) (5.69,-0.6092) (6.09,-0.3304)
    (6.49,-0.0295) (6.89,0.1841) (7.29,0.2657) (7.69,0.2517)
    (8.16,0.2276) (8.56,0.2790) (8.96,0.3989) (9.36,0.5156)
    (9.77,0.5295) (10.17,0.3759) (10.57,0.0703) (10.97,-0.2913)
    (11.37,-0.5779) (11.77,-0.6903) (12.24,-0.5825) (12.64,-0.3684)
    (13.04,-0.1589) (13.44,-0.0345) (13.85,-0.0027) (14.25,-0.0019)
    (14.65,0.0522) (15.05,0.2054) (15.45,0.4287) (15.85,0.6271)
    (16.32,0.6832) (16.72,0.5199) (17.12,0.2020) (17.53,-0.1568)
    (17.93,-0.4284) (18.33,-0.5357) (18.73,-0.4871) (19.13,-0.3622)
    (19.53,-0.2586) (20.00,-0.2340)
};
\addlegendentry{Ten symmetric providers ($K=10$)}

\addplot[black!30, thin, forget plot] coordinates {(0,0) (20,0)};
\end{axis}
\end{tikzpicture}\\
\vspace{2pt}
{\footnotesize \textit{Notes:} Simulated in the static affine
equilibrium of Appendix~\ref{app:formal-model}, for a monopoly
provider and three or ten symmetric providers facing an identical
scarcity path. The unconditional mean action \(\alpha_t\) is fixed
across cases.}
\end{figure}

This also implies that concentration can have consumer-welfare implications
even in the absence of coordinated conduct. Although the mean
clearing price is unchanged when concentration rises, the 
second moments of prices are affected. Proposition~\ref{prop:welfare-concentration} in Appendix~\ref{app:welfare} 
proves this formally by decomposing the expected consumer loss into two channels. 
The first is a state-amplification channel. This channel arises as long as consumer demand 
co-moves with scarcity so that higher scarcity states also imply higher consumer demand. 
Because concentration makes the market price track 
the common scarcity state more closely, consumers pay a higher price precisely in the 
scarcer states when they consume more.
The second is a price-variance channel. Concentration raises the variance of the clearing price, 
and in particular its common cross-provider component. This lowers welfare under a
risk-averse consumer welfare standard. Both channels are weakly increasing 
under concentration-increasing changes in provider shares. Moreover, the variance channel can 
be expressed as the direct counterpart of the cross-seller price-correlation channel in 
\citet{AleksenkoMiklosThal2026}: concentration shifts variance 
from diversifying cross-provider dispersion into a single common component.

\paragraph{Coordinated, same-provider conduct}
The preceding steps show that responsive autobidding can lead to bids conditioning on common scarcity, same-provider batteries co-moving, the price effect of common responses increasing when provider capacity is concentrated, and a reduction in consumer welfare as provider concentration rises. But these patterns arise even when providers maximise owners' profits rather than joint profits across the fleet of batteries that use their software.

The conduct distinction we now consider is an objective-function restriction,
not an additional implication of the stylised model. For a feasible unilateral
deviation $d$ from battery $i$'s observed stack, that is, a reallocation of
offered MW across the battery's price bands, let
\(\Delta\Pi_{o_i,t}(d)\) be the owner's deviation gain, a scalar dollar
amount, and let \(\Delta\Pi_{k_i,-o_i,t}(d)\) be the corresponding profit
effect on same-provider batteries owned by other firms. Owner-level responsive
optimisation requires the first inequality below; same-provider conduct
with conduct parameter \(\lambda\in[0,1]\) replaces it with the second:\footnote{This is the analogue of the conduct-matrix step in \citet{CalderWangKim2026}: conduct assumptions determine which units' profit effects enter the no-deviation condition. Here each recleared deviation produces the two relevant entries, own-profit and same-provider, different-owner profit effects, and \(\lambda\) selects the linear combination. In our setting those entries are generated by exact reclearing of bid-stack deviations rather than by a product-demand model.}
\begin{equation}
\E_{it}\!\left[\Delta\Pi_{o_i,t}(d)\mid Z_{it}\right]\le0,
\qquad
\E_{it}\!\left[
\Delta\Pi_{o_i,t}(d)
+\lambda\Delta\Pi_{k_i,-o_i,t}(d)
\mid Z_{it}\right]\le0.
\label{eq:model-conduct-ic}
\end{equation}
The case $\lambda=0$ is owner-level responsive optimisation. Positive
$\lambda$ means that the observed stack behaves as if the decision-maker
internalises part of the profit effect on same-provider, different-owner
batteries. The empirical conduct test therefore asks whether owner-level inequalities fail, and
whether those failures are systematically reduced when same-provider profit
effects are added.

Consider the focal owner $o$ and a same-provider battery owned by a different firm, 
and a marginal change that moves $\epsilon$ MW of the focal's withheld capacity into 
the clearing band. Let
$$g_o := \left.\frac{\partial \Delta\Pi_{o_i,t}}{\partial \epsilon}\right|_{\epsilon=0}, 
\qquad g_k := \left.\frac{\partial \Delta\Pi_{k_i,-o_i,t}}{\partial \epsilon}\right|_{\epsilon=0}$$
denote the marginal effect of the change on the owner's own profit and on the same-provider peer. 
The change raises the focal owner's dispatch at a price above its opportunity cost, 
so $g_o>0$. 
The same change lowers the regional price, so the peer loses inframarginal revenue, 
$g_k = -\zeta q_k < 0$, with $\zeta := -\partial P/\partial\epsilon > 0$ the price impact 
and $q_k$ the peer's dispatch.

Under owner-level conduct the observed stack is not optimal, because $g_o>0$ leaves a profitable 
deviation untaken and violates the first inequality in \eqref{eq:model-conduct-ic}. 
Under same-provider conduct it is optimal once the weighted peer loss offsets the own gain,
$$g_o + \lambda g_k \le 0, \qquad\text{that is,}\qquad \lambda \ge \lambda^\star
:= \frac{g_o}{\zeta q_k}.$$
Observed restraint in dispatch therefore reveals $\lambda \ge \lambda^\star$, and the conduct 
test then recovers the smallest such $\lambda$ across the sampled deviations. The threshold 
$\lambda^\star$ decreases in the same-provider peer's dispatch $q_k$, which 
rises with provider concentration, so the same restraint is rationalised by a smaller 
$\lambda$ where same-provider capacity is more concentrated.

The empirical analysis in Section~\ref{sec:empirics} proceeds in 
three steps. First, we use the ERI reform to study synchronisation and
state-responsive bidding. Second, we estimate a rolling Bellman value function for stored
energy that estimates the dynamic opportunity cost of bid-stack changes.
Third, we construct a panel of bid stack deviations and use counterfactual 
reclearing of the resulting profit
changes to test the owner-level and same-provider no-deviation
inequalities.

\section{Empirical Application}\label{sec:empirics}

\subsection{Data and Measurement}\label{subsec:empirical-panel}

We build the empirical analysis from a battery-by-5-minute panel. The unit
of observation is a battery \(i\) in dispatch interval \(t\), and the panel
covers April~1,~2025 through February~28,~2026.\footnote{We exclude six
trading days from all analyses: 2025-06-11, 2025-06-12, 2025-06-18,
2025-06-26, 2025-07-02, and 2026-01-26. These are the days in the panel
window with at least 300 interval-region observations whose regional
reference price exceeds \$1,\!000/MWh. The \$1,\!000/MWh threshold sits well
above the routine market regime (the 99th percentile of regional reference
prices in the sample is approximately \$311/MWh) and isolates sustained
administered-price scarcity events whose squared-violation contribution
would otherwise dominate the moment-inequality test. The exclusion removes
about 1.7~percent of the panel.}
It combines AEMO market data
with two objects that are not recorded directly in the market files: the
hand-built mapping from batteries to autobidder providers and a reconstructed
asset-level state-of-charge series. Section~\ref{sec:battery-bidding} describes
the institutional details of battery bidding in the NEM.

From AEMO, we observe the market primitives needed to reconstruct submitted and
realised behaviour: bid prices and quantities, SCADA dispatch and charging,
regional clearing prices, battery registration and capacity information, and
post-ERI battery energy-storage fields. We then attach provider identity at the
battery level using the mapping described in Appendix~\ref{app:provider-id}.
Analyses that require provider identity use the subset of batteries with a
reliable provider assignment; analyses that only require bids, dispatch, prices,
and state of charge use the full panel. Table~\ref{tab:analysis-panel-sources}
summarises the information in the panel and the source of each object.

The identification of software providers is based on public sources. For example, 
in some investor-relations documents of the software vendors, the filings list contracted deployments. 
In other cases, the owners have made public statements naming the provider. 
We further make use of combinations of known hardware-software-bundles. Of the 52 usable batteries in the panel, 42 can be
linked to an autobidder provider.\footnote{In the appendix, we distinguish between 
confirmed and probable provider assignments based on the source. Note also that we cannot observe 
changes in the provider used over time.} These known-provider batteries account for
6.64~GW of discharge capacity and are served by 11 providers across
21 provider-region combinations. The two largest commercial providers are Tesla
Autobidder, with 14 batteries and 2.32~GW, and Fluence Mosaic, with
10 batteries and 2.19~GW. Together they account for about two thirds of the
known-provider capacity.

Provider concentration also varies substantially across regions. Among
known-provider batteries, NSW1 has 11 batteries across 5 providers and a
provider HHI of $4{,}620$ points, with Fluence the largest provider
(63~percent). QLD1 has 10 batteries across 4 providers and an HHI of
$4{,}180$, with Tesla the largest (59~percent). VIC1 is similarly
concentrated, with 12 batteries across 5 providers and an HHI of $4{,}140$,
again with Tesla the largest (57~percent). SA1 is more fragmented: 9 batteries
across 7 providers and an HHI of $2{,}010$, with AGL in-house the largest
(29~percent). This regional variation
matters because the provider concentration relevant for bidding and prices is
local to the dispatch region and the clearing interval, not only a national
fleet share.

\begin{table}[t]
\centering
\caption{Analysis panel: information sources and sample}
\label{tab:analysis-panel-sources}
\small
\begin{tabular}{p{0.26\linewidth}p{0.28\linewidth}p{0.36\linewidth}}
\toprule
Information & Source & Construction \\
\midrule
Battery units and capacities & AEMO/NEMOSIS registration data & Battery mapping and technical characteristics. \\
Bid quantities & AEMO MMS via NEMOSIS & Latest five-minute energy bid quantities. \\
Bid prices & AEMO MMS via NEMOSIS & Daily price bands matched to bid quantities. \\
Clearing prices & AEMO MMS via NEMOSIS & Regional five-minute clearing prices. \\
Dispatch & AEMO MMS via NEMOSIS & Charging, discharging, and net dispatch. \\
State of Charge & AEMO MMS via NEMOSIS + calibrated model & Public SoC plus reconstructed SoC. \\
Autobidding providers & Hand-built mapping from public sources & Provider identity and mapping confidence. \\
\midrule
\multicolumn{3}{l}{\emph{Panel summary}} \\
\multicolumn{2}{l}{Battery-by-five-minute observations} & 4,121,035 \\
\multicolumn{2}{l}{Batteries} & 52 \\
\multicolumn{2}{l}{Trading days} & 328 \\
\multicolumn{2}{l}{Time frame} & April 1, 2025-February 28, 2026 \\
\bottomrule
\end{tabular}
\end{table}

Beginning on July~1,~2025, AEMO expanded public reserve information by
releasing additional battery and generator energy-state information, including
the previous trading day's battery energy availability, regional daily energy
limits, and regional battery state of charge. This disclosure reform is the
Enhanced Reserve Information reform (ERI). For our purposes, ERI is an
information-set change for batteries. Before ERI, market participants could
observe prices, dispatch, and bids, but rival battery energy positions had to be
inferred from market data. After ERI, aggregate battery energy-state information
became part of the public market data. We use this change as a shock to the
precision of the public signal about battery, and therefore market, scarcity.

Interpreting battery behaviour requires an estimate of each battery's current state of charge.
State of charge is central for the analysis because it is the dynamic state that
links current bidding to future opportunity cost. A battery with little stored
energy has a high cost of additional discharge; a battery with ample stored
energy can discharge at lower opportunity cost or preserve energy for later
high-price intervals. However, asset-level SoC is not observed
before ERI. We therefore reconstruct it from the physical stock-flow law that
governs a battery, using the observed SCADA charge and discharge quantities as
flows into and out of storage. Recovering battery operating characteristics
such as round-trip efficiency from public NEM SCADA data has been done before in 
\citet{KarimiArpanahiPourmousaviMahdavi2024}. Let \(S_{it}\) denote stored energy, \(r_{it}\)
charging MW, \(q_{it}\) discharging MW, and \(\Delta=1/12\) hours. For each
battery we propagate\footnote{The recursion allows battery-specific charge and
discharge efficiencies, self-discharge, and auxiliary load, and clips the
state to the physical interval \([0,\bar S_i]\). We calibrate these parameters
separately by battery using the post-ERI public energy-storage fields,
excluding the first week of public observations to reduce sensitivity to
initial conditions.}
\begin{equation}
S_{i,t+1}
=
\operatorname{clip}\!\left(
\rho_i S_{it}
+\eta_i^{\mathrm{ch}}r_{it}\Delta
-\frac{q_{it}\Delta}{\eta_i^{\mathrm{dis}}}
-\kappa_i\Delta,\;0,\;\bar S_i
\right),
\label{eq:main-soc-reconstruction}
\end{equation}
where \(\bar S_i\) is energy capacity, \(\eta_i^{\mathrm{ch}}\) and
\(\eta_i^{\mathrm{dis}}\) are charge and discharge efficiencies, \(\rho_i\) is
a self-discharge factor, \(\kappa_i\) captures auxiliary load, and the clipping
operator enforces the physical storage bounds.

The calibrated model fits the public post-ERI storage data closely. In the
unanchored validation exercise, which mimics a pre-ERI backcast, the median
battery has a Pearson correlation of 0.956 with the public SoC series, an
\(R^2\) of 0.910, and an RMSE of 7.2~percent of energy capacity. In a separate
multi-month holdout propagation, the median correlation is 0.979 and the median
RMSE is 6.9~percent of capacity. Appendix~\ref{app:soc-model} describes the
calibration and validation procedure, and Table~\ref{tab:soc-validation}
reports the full distribution of fit statistics.

All empirical SoC variables are built from this single reconstructed
series.\footnote{Before July~1,~2025, the series is the model propagation in
\eqref{eq:main-soc-reconstruction}, initialised at half of the asset's energy
capacity because no public anchor exists. After ERI, the same model is
re-anchored each day to the most recent public battery energy-storage field and
then propagated forward with observed charge and discharge. This prevents
mechanical drift while preserving the publication lag in public energy-state
information.} Own SoC is the battery's reconstructed stored-energy share.
Common battery scarcity is measured as one minus the regional battery SoC
share, using public SoC where it is available and reconstructed SoC otherwise.
We standardise this common-state measure within region-by-trading-period groups
using the pre-ERI window. The common-state regressions below therefore compare
how bids load on the same regional scarcity state before and after the
disclosure reform, while the own-SoC regressions use the battery's private
dynamic state.

\subsection{Synchronisation and State-Responsive Bidding}\label{subsec:aem-synchronization}

We use ERI to study how a common information shock changes bid construction
across batteries using the same autobidding provider. The stylised model in
Section~\ref{subsec:mechanisms} implies that common software routines can make
same-provider batteries move together, especially when those batteries operate
in the same local market and face the same scarcity state. This subsection
presents two pieces of evidence. The first is pairwise bid synchronisation. The
second is that after ERI same-provider bids condition more strongly on the
public common scarcity state, consistent with the disclosure improving the
precision of the signal that a provider's algorithm maps into bids. A companion check that bids
also load more on each battery's own dynamic state is reported in
Appendix~\ref{app:state-responsive-checks}.

For each battery pair \((i,j)\), trading day \(s\), and market family
\[
f\in\{\text{energy discharge},\ \text{energy charge},\ \text{raise FCAS},\ \text{lower FCAS}\},
\]
we compute, for each bid-stack outcome
\(Y\in\{\text{bid-band index},\,\text{rebid indicator}\}\)
(as defined in Table~\ref{tab:sync-energy-discharge}), the raw
within-day correlation
\[
\rho^{Y,f}_{ijs}
=
\operatorname{corr}_{t\in s}(Y_{it},Y_{jt}),
\]
using the 5-minute intervals in which both batteries have the outcome
observed. We then classify pairs into four groups: same provider
and same region (SPSR), same provider and different region (SPDR), different
provider and same region (DPSR), and different provider and different region
(DPDR). Let \(\bar\rho^{Y,f}_{gs}\) denote the average pair-day correlation for
pair type \(g\) on day \(s\). Separately for each market family and outcome,
we estimate group-by-day difference-in-differences regressions of the form
\begin{equation}
\bar\rho^{Y,f}_{gs}
=
\alpha_g+\gamma_s
+\sum_{c\in\{\mathrm{SPSR},\mathrm{SPDR},\mathrm{DPSR}\}}
\theta_c\,\ind\{g=c\}\mathrm{Post}_s
+\varepsilon_{gs},
\label{eq:sync-did}
\end{equation}
where \(\ind\{\cdot\}\) is the indicator function, \(\alpha_g\) are pair-type
fixed effects, and \(\gamma_s\) are trading-date fixed effects. DPDR is the
omitted group. Our coefficient of interest is
\(\theta_{\mathrm{SPSR}}-\theta_{\mathrm{SPDR}}\). Same-provider batteries in 
different regions (SPDR) should co-move only through provider-wide channels such as 
a shared algorithm, common forecasts, or synchronised update routines, 
because they face different local scarcity states which should cause them to 
set different bid levels and show different rebid timing. Same-provider batteries in 
the same region (SPSR) on the other hand share the same local state. The difference 
$\theta_{\mathrm{SPSR}}-\theta_{\mathrm{SPDR}}$ therefore nets out the 
provider-wide component and isolates the extra synchronisation that arises 
when same-provider batteries respond to the same local scarcity.

\begin{table}[!htbp]\centering
\caption{Energy-discharge synchronisation before and after ERI}
\label{tab:sync-energy-discharge}
\small
\setlength{\tabcolsep}{5pt}
\begin{tabular}{lcccc}
\toprule
\multicolumn{5}{l}{\textbf{Panel A: Energy-discharge bid-band index}} \\
 & \multicolumn{4}{c}{Mean daily raw correlation} \\
\cmidrule(lr){2-5}
Pair type & Pre-ERI & Post-ERI & Difference & Regression \\
\midrule
Same provider, same region (SPSR) & $0.369$ & $0.371$ & $+0.002$ & $+0.039^{***}$ \\
 & & & & $(0.014)$ \\
Same provider, different region (SPDR) & $0.221$ & $0.178$ & $-0.043$ & $-0.006$ \\
 & & & & $(0.009)$ \\
Different provider, same region (DPSR) & $0.271$ & $0.226$ & $-0.045$ & $-0.008$ \\
 & & & & $(0.007)$ \\
Different provider, different region (DPDR) & $0.217$ & $0.180$ & $-0.037$ & \multicolumn{1}{c}{omitted} \\
\midrule
$\theta_{\mathrm{SPSR}}-\theta_{\mathrm{SPDR}}$ (SPSR $-$ SPDR) & & & $+0.045$ & $+0.045^{***}$ \\
 & & & & $(0.015)$ \\
\addlinespace[4pt]
\multicolumn{5}{l}{\textbf{Panel B: Energy rebid timing}} \\
 & \multicolumn{4}{c}{Mean daily raw correlation} \\
\cmidrule(lr){2-5}
Pair type & Pre-ERI & Post-ERI & Difference & Regression \\
\midrule
Same provider, same region (SPSR) & $0.216$ & $0.354$ & $+0.138$ & $+0.153^{***}$ \\
 & & & & $(0.017)$ \\
Same provider, different region (SPDR) & $0.196$ & $0.252$ & $+0.057$ & $+0.072^{***}$ \\
 & & & & $(0.018)$ \\
Different provider, same region (DPSR) & $0.044$ & $0.035$ & $-0.010$ & $+0.006$ \\
 & & & & $(0.004)$ \\
Different provider, different region (DPDR) & $0.050$ & $0.035$ & $-0.015$ & \multicolumn{1}{c}{omitted} \\
\midrule
$\theta_{\mathrm{SPSR}}-\theta_{\mathrm{SPDR}}$ (SPSR $-$ SPDR) & & & $+0.081$ & $+0.081^{***}$ \\
 & & & & $(0.017)$ \\
\midrule
Pair-type FE & & & & Yes \\
Trading-date FE & & & & Yes \\
$N$ (group $\times$ day) & & & & 1,336 \\
Trading dates & & & & 334 \\
\bottomrule
\end{tabular}
\begin{minipage}{0.94\linewidth}
\vspace{0.15cm}
\footnotesize \textit{Notes:} Entries are raw Pearson correlations, not Fisher-transformed correlations. For each battery pair and trading date, we compute the within-day correlation of the listed outcome across 5-minute intervals. Pre/post columns are simple averages of daily pair-type means. The Regression column reports \(\hat\theta_g\) from \eqref{eq:sync-did}; standard errors, in parentheses, are clustered by trading date. The bottom row in each panel reports the Wald test for \(\theta_{\mathrm{SPSR}}-\theta_{\mathrm{SPDR}}\). $^{***}p<0.01$, $^{**}p<0.05$, $^{*}p<0.10$.
\end{minipage}
\end{table}

Table~\ref{tab:sync-energy-discharge} shows the main pattern for
energy-discharge offers. Panel~A uses the bid-band index, a quantity-weighted
measure of where the battery places its offered MW in the ten bid bands.
Same-provider batteries in the same region become more synchronised relative to
same-provider batteries in different regions: the SPSR--SPDR gap rises by
0.045 after ERI. Panel~B shows a related but distinct pattern for rebid
timing. Rebid timing is already correlated for same-provider batteries in
different regions, and rises from 0.196 to 0.252 for SPDR pairs after ERI. It
rises more for SPSR pairs, from 0.216 to 0.354, producing a relative increase
of 0.081. Thus, bid-update timing has a provider-wide component, while the
bid-stack level itself becomes most synchronised when same-provider batteries
operate in the same local market.

\begin{table}[!htbp]\centering
\caption{Synchronization effects across market families}
\label{tab:sync-market-effects}
\small
\setlength{\tabcolsep}{4.2pt}
\begin{tabular}{lcccc}
\toprule
Market family & Bid-band index & Upper-band price & Lower-band price & Rebid timing \\
\midrule
Energy discharge & +0.045$^{***}$ (0.015) & +0.147$^{***}$ (0.030) & +0.051 (0.045) & +0.081$^{***}$ (0.017) \\
Energy charge & +0.085$^{***}$ (0.016) & +0.087 (0.070) & +0.091$^{***}$ (0.025) & +0.081$^{***}$ (0.017) \\
Raise FCAS & +0.151$^{***}$ (0.019) & +0.081$^{***}$ (0.021) & +0.107$^{***}$ (0.022) & +0.087$^{***}$ (0.027) \\
Lower FCAS & +0.224$^{***}$ (0.017) & +0.295$^{***}$ (0.021) & +0.092$^{***}$ (0.020) & +0.072$^{***}$ (0.027) \\
\bottomrule
\end{tabular}
\begin{minipage}{0.94\linewidth}
\vspace{0.15cm}
\footnotesize \textit{Notes:} Entries report the SPSR--SPDR DiD coefficient from group-by-day regressions of raw within-day pairwise correlations. For each pair, trading date, market family, and outcome, we compute the Pearson correlation across 5-minute intervals in which both batteries have the outcome observed. Regressions include pair-type and trading-date fixed effects; standard errors, in parentheses, are clustered by trading date. Table~\ref{tab:sync-energy-discharge} reports the detailed pre/post correlations for the two main energy-discharge outcomes; Appendix Table~\ref{tab:sync-market-details-appendix} reports the corresponding detail for the other market families. The maximum number of trading dates is 334. $^{***}p<0.01$, $^{**}p<0.05$, $^{*}p<0.10$.
\end{minipage}
\end{table}

Table~\ref{tab:sync-market-effects} reports the same SPSR-minus-SPDR estimate
across market families and bid-stack outcomes. For energy discharge, the
effect is largest in upper bid prices, with an estimate of 0.147; lower-band
prices do not show a comparable effect. In energy charge, synchronisation rises
in the bid-band index and lower-band prices. In FCAS markets, the effects are
larger across the bid stack: the bid-band-index effect is 0.151 for raise FCAS
and 0.224 for lower FCAS. Rebid-timing synchronisation is positive in every
market family, with estimates between 0.072 and 0.087. This relative effect
combines two patterns. In energy, same-provider different-region timing rises
after ERI, consistent with provider-wide bid-update routines. In FCAS, the
SPDR timing correlation is already high before ERI and changes little after the
reform. Appendix Table~\ref{tab:sync-market-details-appendix} reports the
corresponding pre/post correlations outside energy discharge.

The same reform also increases how strongly bids condition on the common scarcity state.
We measure regional battery scarcity as one minus the regional state-of-charge
share, split it into a tight and a loose component, and estimate how strongly
quantity-weighted discharge bids load on the tight relative to the loose state
before and after ERI, separately by evening window; Appendix~\ref{app:state-responsive-checks}
gives the specification. Table~\ref{tab:common_state_window} reports the result.
In the upper discharge bands, which govern withholding, the post-ERI
tight-minus-loose loading rises by $1{,}136$ AUD/MWh in the early-peak window
and by a comparable amount in the lead-in window, while
the lower bands, which secure baseline dispatch, show no such shift. 
The effect is concentrated in the lead-in
and early-peak windows with the later windows showing no stable common-state response.
Same-provider bids thus track the common scarcity signal more closely once that
signal becomes more precise, rising in the upper bands during scarce intervals
and easing during abundant ones, exactly the responsive-bidding channel of
Section~\ref{subsec:mechanisms}. A similar specification shows that bids also
load more strongly on each battery's own reconstructed state of charge after
ERI (Appendix~\ref{app:state-responsive-checks}).

\begin{table}[!htbp]\centering
\caption{Common-state responsiveness by evening window}
\label{tab:common_state_window}
\begin{tabular}{lccc}
\toprule
Window & Upper bands & Lower bands & All bands \\
\midrule
Lead-in & 1,515$^{*}$ (819) & -33 (75) & 792 (826) \\
Early peak & 1,136$^{***}$ (368) & -11 (22) & 516 (479) \\
Late peak & -1,587 (1,077) & 121$^{*}$ (62) & -3 (698) \\
Post-peak & 324 (1,199) & 106 (78) & 1,230 (1,193) \\
\bottomrule
\end{tabular}
\begin{minipage}{0.82\linewidth}
\vspace{0.15cm}
\scriptsize \textit{Notes:} Entries are post-ERI average changes in the tight-minus-loose loading of bid prices on the public regional scarcity signal. Each row is estimated in a separate regression restricted to the indicated trading-period window. Outcomes are quantity-weighted energy-discharge price bands. All specifications include battery fixed effects, region-by-date fixed effects, region-by-period fixed effects, and own-SoC event-study controls. Standard errors, in parentheses, are two-way clustered by battery and trading date. $^{***}p<0.01$, $^{**}p<0.05$, $^{*}p<0.10$.
\end{minipage}
\end{table}

We next relate this common-provider bid behaviour to provider concentration.
The synchronisation results show that autobidding providers induce co-movement
in bid construction, especially when same-provider batteries face the same
local market conditions. The model suggests that this channel should be more
important when provider capacity is concentrated in the relevant market.

We begin with a regional concentration measure. For each region, we compute the
installed provider HHI using battery capacity shares among known autobidder
providers. We then form a region-day
panel with two outcomes. Local synchronisation is the same-region
SPSR correlation minus the same-region DPSR correlation.
Market-action variance measures how much the region's bids move over
the day: for each 5-minute interval we compute the offer-weighted
average upper-band bid price across the region's batteries, and take
the variance of this regional average across the intervals of the
trading day.

\begin{table}[!htbp]\centering
\caption{Provider concentration, bid synchronisation, and market-action variance}
\label{tab:concentration-sync-variance}
\small
\setlength{\tabcolsep}{5.5pt}
\begin{tabular}{lcccc}
\toprule
 & (1) & (2) & (3) & (4) \\
 & \multicolumn{2}{c}{Synchronisation} & \multicolumn{2}{c}{Market-action variance} \\
\cmidrule(lr){2-3}\cmidrule(lr){4-5}
 & \begin{tabular}{c}Local\\SPSR $-$ DPSR\end{tabular} & \begin{tabular}{c}SPSR\\correlation\end{tabular} & \begin{tabular}{c}on local\\sync\end{tabular} & \begin{tabular}{c}on provider\\HHI\end{tabular} \\
\midrule
Provider HHI (10 p.p.) & +0.161$^{***}$ & +0.169$^{***}$ & & +1.084$^{***}$ \\
 & (0.021) & (0.018) & & (0.378) \\
\addlinespace[2pt]
Local sync & & & +0.183 & \\
 & & & (0.882) & \\
\midrule
Unit & Region day & Region day & Region day & Region day \\
Region FE & No & No & No & No \\
Date FE & Yes & Yes & Yes & Yes \\
$N$ & 738 & 738 & 738 & 738 \\
\bottomrule
\end{tabular}
\begin{minipage}{0.94\linewidth}
\vspace{0.15cm}
\footnotesize \textit{Notes:} The outcome is the upper-band energy-discharge price. Provider HHI is the installed capacity-share HHI across known autobidder providers, scaled so one unit equals a ten percentage point HHI increase. Local sync is SPSR minus DPSR mean pairwise correlation within region. Market-action variance is the within-day variance of the offer-weighted regional average bid price, reported in millions of squared AUD/MWh. Installed HHI is fixed within region, so specifications include date fixed effects but not region fixed effects. Columns (1)--(2) are weighted by SPSR pair-days; columns (3)--(4) by market intervals. Standard errors are clustered by trading date. $^{***}p<0.01$, $^{**}p<0.05$, $^{*}p<0.10$.
\end{minipage}
\end{table}

Table~\ref{tab:concentration-sync-variance} reports the reduced-form link for
upper-band energy-discharge prices. A ten percentage point increase in
installed provider HHI is associated with a 0.161 increase in the local
synchronisation premium and a 0.169 increase in SPSR synchronisation. The same
increase in HHI is also associated with a larger within-day variance
of the regional average upper-band bid price. Appendix
Table~\ref{tab:concentration-sync-variance-markets-appendix} repeats the same
upper-band-price exercise across market families. The positive relationship
between provider HHI and synchronisation is concentrated in energy discharge, while
the relationship between provider HHI and bid-price variance is broader and is also present in FCAS. However, these estimates are descriptive: installed provider HHI is
effectively cross-sectional, so the specifications include
trading-date fixed effects but not region fixed effects. The result
can therefore not be considered a causal estimate of the effect of provider
concentration, but the signs and magnitudes are in line with the
mechanism in Section~\ref{subsec:mechanisms}. Regions with more
concentrated providers show both more same-provider synchronisation
and more within-day variance of the regional average bid price.

\subsection{Dynamic Value of Stored Energy}\label{subsec:bellman-value}

The synchronisation evidence in
Section~\ref{subsec:aem-synchronization} shows that common providers
produce co-movement in bid construction, and
Table~\ref{tab:concentration-sync-variance} relates that co-movement
to provider concentration. The conduct test in
Section~\ref{subsec:conduct-test} asks whether the observed bids can
be rationalised by owner-level dynamic incentives, or whether they
internalise same-provider, different-owner profit effects. To run
that test we first need an estimate of the dynamic opportunity cost
of moving battery energy between intervals.
For each battery $i$ we estimate a single-battery value function
$V_{i,\tau}(S,x)$ that maps own state of charge $S$ and a market state $x$
into the continuation value of stored energy. The value function solves
the Bellman recursion
\begin{equation}
V_{i,\tau}(S,x)
=
\max_{S'\in\Gamma_i(S)}
\Bigl\{
\pi_i(S,S',x,\tau)
+
\delta\,\E\!\left[V_{i,\tau+1}(S',x')\,\big|\,x\right]
\Bigr\},
\label{eq:main-bellman}
\end{equation}
where $\pi_i$ collects current energy and FCAS profits net of throughput
cost, $\delta$ is the discount factor, $\Gamma_i(S)$ enforces the SoC bounds
and ramp limits, and the expectation is taken under a transition kernel for
the market state $x$ that we estimate non-parametrically from the trailing
data window. The market state aggregates four conditioning variables: the
time-of-day block, the residual-price bin (price relative to the regional
mean profile, discretised into $N_p=11$ empirical-quantile bins), the
regional residual-demand state, and an aggregate-SoC signal that captures
whether the regional storage fleet enters the interval close to empty,
mid-range, or close to full. Conditioning on these variables lets the
continuation value reflect both intraday arbitrage (price within the day)
and dynamic scarcity (how much regional storage is available).

To keep the value function estimation tractable, we assume a finite horizon 
of a single trading day from $\tau=1$ (06:00) to $T$ (22:00). We fix the terminal value at
22:00 as the expected revenue from any remaining stored energy if it were
discharged at 22:00 at the average overnight price of a 90-day trailing-window. 
Because overnight prices are stable and low relative to the daytime peak, 
and vary little over the 90-day period, this terminal value is small and is nearly the same across states, so the
shadow values we use during the evening peak depend little
on this estimated terminal value. The assumption of a finite horizon is important, as it 
allows us to avoid any infinite-horizon solve. $\pi_i$ then co-optimises energy and
FCAS using a hierarchical estimate of the FCAS headroom value (revenue
per MWh of physical headroom, varying with block, residual demand, and
residual price), so that deviations that reallocate capacity between the
energy and FCAS sides of the stack are priced consistently.

We then calculate the implied shadow value of stored energy $\widehat m_{it}$,
\begin{equation}
\widehat m_{it}
:=
\frac{\max\bigl\{\partial V_{i,\tau}/\partial S,\,0\bigr\}}{\eta_i^{\mathrm{dis}}}+c_i,
\label{eq:main-mhat}
\end{equation}
which yields the per-MWh dynamic opportunity cost of discharge at interval $t$, with
$\eta_i^{\mathrm{dis}}$ the discharge efficiency and $c_i$ the throughput cost. The
non-negativity floor reflects the option value of waiting.

The Bellman is re-estimated each week on a 90-day trailing window so the
price profile, transition kernel, and FCAS headroom values can absorb
seasonal variation in market conditions. Within each window we construct
the regional mean price profile, the residual-price bin edges, the
non-parametric transition matrix (with a hierarchical fallback chain
across conditioning groups when within-group support is thin), the FCAS
headroom-value coefficients, and the overnight terminal value. The DP is
then solved by backward induction on a $41$-point SoC grid with $201$
candidate next-period SoC controls per interval, separately by signal and
residual-demand state. Asset parameters (energy capacity, discharge and
charge power limits, charge and discharge efficiencies, self-discharge,
and a parasitic-drain decomposition) are calibrated from the same SoC
reconstruction model as in Section~\ref{subsec:empirical-panel}. 
Appendix~\ref{app:bellman-estimation} reports the full
state-space specification, the FCAS headroom-value construction, the
rolling-pipeline coverage, the asset/SoC calibration parameters, the
shadow-value profile across time blocks and SoC bins, and transition,
FCAS, and deviation-valuation diagnostics.

We then use our estimated Bellman values to determine profits under 
counterfactual market outcomes. 
To turn a bid-stack edit into a counterfactual market outcome we use the
open-source NEMDE replication \texttt{nempy}
\citep{Prakash2022nempy}, which reproduces AEMO's regional dispatch
optimisation from the same MMS-DB and NEMDE-XML inputs that AEMO uses at
clearing time. For each 5-minute interval, \texttt{nempy} solves the
constrained least-cost dispatch problem
\begin{equation}
\min_{u,v,f,F}\ \
\sum_{i,b} p^{E,\mathrm{gen}}_{ib}\, u_{ib}
-\sum_{i,b} p^{E,\mathrm{load}}_{ib}\, v_{ib}
+\sum_{i,j,b} p^{F,j}_{ib}\, f^{\,j}_{ib}
+\Phi(\text{slacks})
\label{eq:nempy-objective}
\end{equation}
subject to per-region energy balance with inter-regional flow and losses,
inter-regional transmission limits, raise- and lower-FCAS service
requirements at the regional and global levels, unit-level ramp rates and
availability, the FCAS trapezium constraints that couple energy operating
points to enabled FCAS quantities (Appendix~\ref{app:fcas-trapezium}),
fast-start unit constraints, and AEMO's set of generic network and
security constraints. Decision variables are energy generation bid
acceptances $u_{ib}$ (units $i$, price bands $b$), energy load bid
acceptances $v_{ib}$, FCAS service enablements $f^{\,j}_{ib}$ across the
ten regulation and contingency services, and aggregate interconnector
flows $F$. The slack penalties $\Phi(\cdot)$ replicate AEMO's
constraint-violation pricing so that infeasibilities are resolved at
AEMO's preference ordering rather than producing infeasible solves.
Regional reference prices recover as Lagrange multipliers on the
regional energy-balance constraints. 

The \texttt{nempy} library has been validated
against historical AEMO dispatch outcomes in \citet{Prakash2022nempy}, which 
shows high accuracy in replicating market clearing results. We report an end-to-end
replication check against AEMO's published regional reference prices on
our own panel in Appendix~\ref{app:reclearing-validation}. 
We find a median absolute discrepancy of \$$0.05$/MWh. More than $97\%$ of intervals
are reproduced within \$$20$ of AEMO's published RRP.\footnote{Replication degrades when RRP approaches the price cap. This is expected, as this represents extreme scarcity intervals during which AEMO's full constraint set is difficult to reproduce. These intervals however are not concentrated in the evening-peak windows the conduct test focuses on.}

As a final validation, we construct a randomly selected panel of broad counterfactual bid-stack
deviations. Each deviation modifies the focal battery's submitted stack
while holding all other batteries' stacks fixed, and the cleared interval
is recomputed using \eqref{eq:nempy-objective} including the FCAS
co-optimisation. We consider three types of deviations, illustrated in
Figure~\ref{fig:deviation-operators}:
within-market downward and upward reallocations across the bid bands of
a single market (energy or FCAS), and cross-market reallocations between
energy and FCAS. Based on observed usage patterns we simplify by 
considering deviations between three sections of the ten bid bands. Top bands $8$--$10$ are used to park energy, middle
bands $4$--$7$ are used near the clearing price, and low bands $1$--$3$ are used
to ensure dispatch in most intervals. The deviations always shift between 
two of these three sections. Shift sizes are fixed shares of
the relevant offered MW and are capped by the available MW in the source
band. We then randomly select $800$ candidate rows per deviation type. Appendix
Table~\ref{tab:deviation-menu} lists the full set of deviations with
their shift sizes.

Each recleared deviation produces effects in three channels. Let
$\Delta\Pi^E_{itd}$ and $\Delta\Pi^F_{itd}$ denote the changes in current
energy and FCAS profit, and let
\begin{equation}
\Delta V_{itd}
:=
\E\!\left[V_{i,t+1}\bigl(S^{cf}_{i,t+1},x'\bigr)\,\big|\,x_t\right]
-
\E\!\left[V_{i,t+1}\bigl(S^{base}_{i,t+1},x'\bigr)\,\big|\,x_t\right]
\label{eq:main-deltav}
\end{equation}
denote the change in continuation value induced by the deviation's effect
on next-period state of charge, where $S^{cf}_{i,t+1}$ and
$S^{base}_{i,t+1}$ are the recleared and observed next-period SoC and the
expectation is taken under the same matrix that appears in
\eqref{eq:main-bellman}. The focal owner's deviation gain on battery
$i$ at interval $t$ is the sum of the three channels,
\begin{equation}
\Delta\Pi^{\mathrm{own}}_{itd}
:=
\Delta\Pi^E_{itd}
+\Delta\Pi^F_{itd}
+\Delta V_{itd}.
\label{eq:main-gain-decomp}
\end{equation}
A battery-interval is incentive compatible if no sampled deviation
yields a strictly positive own-profit gain. We report the share of
battery-intervals with at least one deviation yielding a positive 
profit deviation gain above a given threshold ranging from
\$$1$ to \$$1{,}000$ as a robustness check on Bellman
measurement error.\footnote{Appendix~\ref{app:bellman-estimation} also reports
the diagnostic under an alternative current-bin convention for
$\Delta V_{itd}$.}

\begin{figure}[t]
\centering
\caption{Schematic bid-stack deviations}
\label{fig:deviation-operators}
\begin{subfigure}[t]{0.26\linewidth}
\centering
\begin{tikzpicture}[scale=0.70,transform shape,font=\scriptsize,>=Stealth,line cap=round,line join=round]
  \draw[->,black!50] (-0.15,-0.15) -- (3.25,-0.15) node[right] {MW};
  \draw[->,black!50] (-0.15,-0.15) -- (-0.15,3.05) node[above] {Price};
  \foreach \y/\lab in {0/{$1$--$3$},1/{$4$--$7$},2/{$8$--$10$}} {
    \filldraw[fill=white,draw=black!45] (0,\y) rectangle (2.75,\y+0.65);
    \node[anchor=west] at (0.12,\y+0.325) {\lab};
  }
  \filldraw[fill=black!12,draw=black!55] (0,2) rectangle (2.75,2.65);
  \node[anchor=west] at (0.12,2.325) {$8$--$10$};
  \node[anchor=east] at (2.63,2.325) {source};
  \filldraw[fill=black!55,draw=black!55] (0,0) rectangle (2.75,0.65);
  \node[anchor=west,text=white] at (0.12,0.325) {$1$--$3$};
  \node[anchor=east,text=white] at (2.63,0.325) {target};
  \draw[->,thick,black!80] (2.95,2.35) to[out=5,in=5,looseness=1.15] (2.95,0.35);
\end{tikzpicture}
\caption{Move down}
\end{subfigure}
\hfill
\begin{subfigure}[t]{0.26\linewidth}
\centering
\begin{tikzpicture}[scale=0.70,transform shape,font=\scriptsize,>=Stealth,line cap=round,line join=round]
  \draw[->,black!50] (-0.15,-0.15) -- (3.25,-0.15) node[right] {MW};
  \draw[->,black!50] (-0.15,-0.15) -- (-0.15,3.05) node[above] {Price};
  \foreach \y/\lab in {0/{$1$--$3$},1/{$4$--$7$},2/{$8$--$10$}} {
    \filldraw[fill=white,draw=black!45] (0,\y) rectangle (2.75,\y+0.65);
    \node[anchor=west] at (0.12,\y+0.325) {\lab};
  }
  \filldraw[fill=black!12,draw=black!55] (0,1) rectangle (2.75,1.65);
  \node[anchor=west] at (0.12,1.325) {$4$--$7$};
  \node[anchor=east] at (2.63,1.325) {source};
  \filldraw[fill=black!55,draw=black!55] (0,2) rectangle (2.75,2.65);
  \node[anchor=west,text=white] at (0.12,2.325) {$8$--$10$};
  \node[anchor=east,text=white] at (2.63,2.325) {target};
  \draw[->,thick,black!80] (2.95,1.35) to[out=-5,in=-5,looseness=1.15] (2.95,2.35);
\end{tikzpicture}
\caption{Move up}
\end{subfigure}
\hfill
\begin{subfigure}[t]{0.43\linewidth}
\centering
\begin{tikzpicture}[scale=0.70,transform shape,font=\scriptsize,>=Stealth,line cap=round,line join=round]
  \begin{scope}
    \draw[->,black!50] (-0.15,-0.15) -- (3.25,-0.15) node[right] {MW};
    \draw[->,black!50] (-0.15,-0.15) -- (-0.15,3.05) node[above] {Price};
    \node[anchor=south] at (1.25,3.02) {Market A};
    \foreach \y/\lab in {0/{$1$--$3$},1/{$4$--$7$},2/{$8$--$10$}} {
      \filldraw[fill=white,draw=black!45] (0,\y) rectangle (2.75,\y+0.65);
      \node[anchor=west] at (0.12,\y+0.325) {\lab};
    }
    \filldraw[fill=black!12,draw=black!55] (0,2) rectangle (2.75,2.65);
    \node[anchor=west] at (0.12,2.325) {$8$--$10$};
    \node[anchor=east] at (2.63,2.325) {source};
  \end{scope}
  \begin{scope}[xshift=4.25cm]
    \draw[->,black!50] (-0.15,-0.15) -- (3.25,-0.15) node[right] {MW};
    \draw[->,black!50] (-0.15,-0.15) -- (-0.15,3.05) node[above] {Price};
    \node[anchor=south] at (1.25,3.02) {Market B};
    \foreach \y/\lab in {0/{$1$--$3$},1/{$4$--$7$},2/{$8$--$10$}} {
      \filldraw[fill=white,draw=black!45] (0,\y) rectangle (2.75,\y+0.65);
      \node[anchor=west] at (0.12,\y+0.325) {\lab};
    }
    \filldraw[fill=black!55,draw=black!55] (0,0) rectangle (2.75,0.65);
    \node[anchor=west,text=white] at (0.12,0.325) {$1$--$3$};
    \node[anchor=east,text=white] at (2.63,0.325) {target};
  \end{scope}
  \draw[->,thick,black!80] (2.92,2.325) to[out=25,in=155,looseness=1.08] (4.08,0.325);
\end{tikzpicture}
\caption{Move across markets}
\end{subfigure}
\vspace{2pt}
{\footnotesize \textit{Notes:} Each panel uses three price-band groups:
low bands \(1\)--\(3\), middle bands \(4\)--\(7\), and high bands
\(8\)--\(10\). Light grey marks the source band from which quantity is
removed; dark grey marks the target band into which the same quantity is
added. Panel (a) illustrates within-market downward reallocations, Panel (b)
illustrates within-market upward reallocations such as clearing-to-above
withholding, and Panel (c) illustrates cross-market reallocations between
energy and FCAS stacks.}
\end{figure}
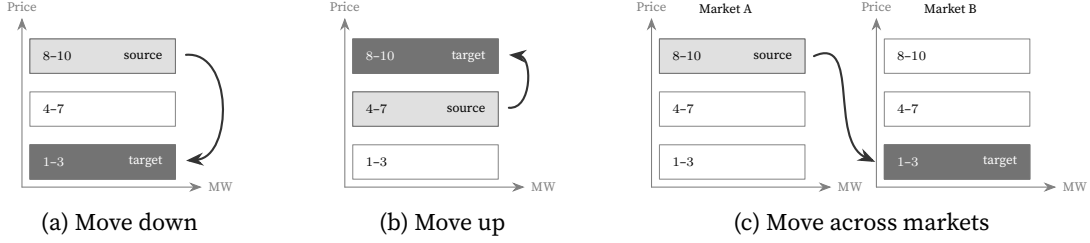

Table~\ref{tab:ic-violations-broad} reports the pre/post-ERI share of
focal battery-intervals with at least one sampled deviation that increases 
profits relative to the observed bid stack by more than the listed threshold. 
We find that the shares of intervals with positive deviation gains post-ERI
are below pre-ERI shares in every column of every deviation
type. For example, for within energy deviations, we observe 21.7\% of our sample with a \$1 difference in total profit (after accounting for the change in the 
continuation value) pre-ERI, which drops to 8.8\% post-ERI. This is consistent with the reduced-form evidence that ERI led to better and 
more precise targeting of the underlying scarcity state and thus efficiency 
gains. We consistently observe that remaining profitable deviations are 
fewer but also, on average, economically smaller. The generally small differences 
that we find support the reliability of the Bellman estimation. If our Bellman 
values were imprecise, we should observe large deviation gains and many intervals 
with substantial differences between profit from observed bids and from the 
counterfactual deviation.

\begin{table}[h]
\centering
\caption{Incentive-compatibility violations, pre/post-ERI, by deviation
panel and dollar tolerance}
\label{tab:ic-violations-broad}
\small
\begin{tabular}{@{} l l r r r r r @{}}
\toprule
 & & & \multicolumn{4}{c}{Share above} \\
\cmidrule(lr){4-7}
Deviation family & Period & Intervals & \$1 & \$10 & \$100 & \$1{,}000 \\
\midrule
Within energy & Pre-ERI & 702 & 21.7\% & 17.4\% & 13.0\% & 10.8\% \\
 & Post-ERI & 1,114 & 8.8\% & 3.7\% & 0.9\% & 0.3\% \\
Within FCAS & Pre-ERI & 388 & 19.1\% & 1.5\% & 0.3\% & 0.0\% \\
 & Post-ERI & 853 & 6.0\% & 0.2\% & 0.0\% & 0.0\% \\
Cross-market & Pre-ERI & 1,035 & 23.8\% & 14.0\% & 6.0\% & 3.5\% \\
 & Post-ERI & 1,506 & 10.2\% & 4.1\% & 1.3\% & 0.5\% \\
Cross-FCAS & Pre-ERI & 372 & 10.8\% & 0.5\% & 0.0\% & 0.0\% \\
 & Post-ERI & 920 & 2.1\% & 0.0\% & 0.0\% & 0.0\% \\
\bottomrule
\end{tabular}
\normalsize

\vspace{4pt}
{\footnotesize \textit{Notes:} Pre/post-ERI share of battery-intervals in which at least one feasible sampled deviation yields a strictly positive own-profit gain exceeding the listed dollar tolerance. The own-profit gain combines the recleared change in current energy plus FCAS profit and the change in the Bellman continuation value, with continuation values evaluated under the kernel-expectation convention.}
\end{table}

We now employ the state-of-charge calibration, Bellman value estimation, and 
validated reclearing methodology to perform our conduct test.

\subsection{Conduct Test by Counterfactual Reclearing}\label{subsec:conduct-test}

We now turn to the conduct test itself. The starting point is the
standard view of market power in uniform-price electricity auctions.
In supply-function models, firms exercise market power by making their
supply schedules less elastic or by offering less effective capacity near
the clearing margin \citep{KlempererMeyer1989,Vives2011,Holmberg2008}.
In repeated uniform-price auctions, collusive bidding similarly sustains
high prices by weakening the incentive to expand supply, and deviations
from a collusive arrangement are typically undercutting or output-expansion
deviations that win more dispatch while rivals continue to bid less
aggressively \citep{Fabra2003}. Storage adds an intertemporal layer:
strategic storage owners internalise the price impact of when they buy
and sell electricity, so market power can appear as smoothed or withheld
discharge rather than purely as a static markup
\citep{AndresCerezoFabra2023}. This logic also arises with 
independent learning algorithms that can generate seemingly collusive battery decisions in a
dynamic storage environment \citep{AbadaLambin2023}.
We therefore consider deviations that move part of the MW
parked just above the regional clearing price down into the clearing band,
thereby increasing the focal battery's dispatch and own profit while
lowering the price and profits earned by other same-provider batteries in
the region.  

We formalise the test as follows. For a focal battery $i$ in region
$r$ at interval $t$ and a candidate deviation $d$, the reclearing
engine returns, for every battery $j$ in $r$, a change
$\Delta\pi_{jtd}$ in current energy plus FCAS profit and a change
$\Delta V_{jtd}$ in the continuation value of stored energy where each
$\Delta V_{jtd}$ is obtained by applying $j$'s Bellman value function
from~\eqref{eq:main-bellman} to the next-period state of charge
implied by $j$'s recleared dispatch under $d$. For the focal $j=i$,
$\Delta\pi_{itd}+\Delta V_{itd}$ coincides with the focal own-profit
gain $\Delta\Pi^{\mathrm{own}}_{itd}=\Delta\Pi^E_{itd}+\Delta\Pi^F_{itd}+\Delta V_{itd}$
defined in~\eqref{eq:main-gain-decomp}.

We then compute the focal owner's total deviation gain on its
regional portfolio: $i$'s own gain plus the gains on any other
batteries the same firm operates in $r$,
\begin{equation}
\Delta\Pi_{o_i,t}(d)
:=
\sum_{\substack{j:\,o_j=o_i\\ r_j=r_i}}
\bigl[\Delta\pi_{jtd}+\Delta V_{jtd}\bigr].
\label{eq:conduct-owner-gain}
\end{equation}
We also calculate the profit change for same-provider,
different-owner batteries that share $i$'s regional clearing, or
\begin{equation}
\Delta\Pi_{k_i,-o_i,t}(d)
:=
\sum_{\substack{j:\,k_j=k_i,\,o_j\neq o_i\\ r_j=r_i}}
\bigl[\Delta\pi_{jtd}+\Delta V_{jtd}\bigr].
\label{eq:conduct-spillover-gain}
\end{equation}
By construction, the deviations we consider can deliver positive
$\Delta\Pi_{o_i,t}(d)$ -- the focal owner gains by capturing
additional dispatch at a price still above marginal cost -- and
negative $\Delta\Pi_{k_i,-o_i,t}(d)$ -- same-provider peers see the
recleared price fall and lose revenue on their inframarginal MW, and
their stored energy depreciates in value at the lower realised
price.

Under \textit{owner-level} conduct ($\lambda=0$), the observed bid is
optimal at the focal owner's information set $\mathcal I_{ot}$ if no
recleared deviation leaves a positive expected own-portfolio gain,
\begin{equation}
\E_{it}\!\left[\Delta\Pi_{o_i,t}(d)\,\middle|\,Z_{it}\right]\le 0
\quad\text{for every $d$ and every observable $Z_{it}\subseteq\mathcal I_{ot}$.}
\label{eq:conduct-hypothesis-own}
\end{equation}
Under \textit{same-provider} conduct with weight
$\lambda\in(0,1]$, the bidder also values the recleared profit
effect on same-provider peers,
\begin{equation}
\E_{it}\!\left[
\Delta\Pi_{o_i,t}(d)
+\lambda\,\Delta\Pi_{k_i,-o_i,t}(d)
\,\middle|\,Z_{it}\right]\le 0
\quad\text{for every $d$ and every $Z_{it}$.}
\label{eq:conduct-hypothesis-jpm}
\end{equation}
Because $\Delta\Pi_{k_i,-o_i,t}(d)$ is typically negative on the
deviations the test samples, a positive $\lambda$ shrinks
the left-hand side of~\eqref{eq:conduct-hypothesis-jpm}: the bidder
is rationalised in leaving a deviation on the table
whenever the own-portfolio gain is more than offset by the recleared
loss on same-provider peers, weighted by $\lambda$. Equivalently, a
positive $\lambda$ identifies bidding behaviour that looks like
sub-optimal parking of high-band capacity unless the
bidder is also pricing in the price-suppressing effect on same-provider
peers.

To test these inequalities, we fix a candidate $\lambda$ and calculate
which observed deviations the corresponding hypothesis cannot
rationalise. A deviation $(i,t,d)$ contradicts hypothesis $\lambda$
whenever its $\lambda$-weighted gain
$\Delta\Pi_{o_i,t}(d)+\lambda\,\Delta\Pi_{k_i,-o_i,t}(d)$ is strictly
positive -- under that hypothesis the bidder would have left a
positive expected payoff on the table by not deviating. Following
\citet{PakesPorterHoIshii2015}, we accumulate the contradictions
as follows
\begin{equation}
Q_n(\lambda)
:=
\frac{1}{n}\sum_{(i,t,d)}
w_{itd}\,
\max\!\left\{
\Delta\Pi_{o_i,t}(d)+\lambda\,\Delta\Pi_{k_i,-o_i,t}(d),\,0
\right\}^{2},
\label{eq:conduct-objective}
\end{equation}
where the sum is over the test panel described shortly and $w_{itd}$
are stratification weights to ensure balance across
coalitions and ERI periods. Deviations the hypothesis already
rationalises contribute zero by construction, so $Q_n(\lambda)$ is also 
the squared dollar size of the contradictions left under hypothesis
$\lambda$. The conduct estimator is the value of $\lambda$ on $[0,1]$
that minimises this contradiction count,
\begin{equation}
\widehat\lambda
:=
\arg\min_{\lambda\in[0,1]}Q_n(\lambda),
\label{eq:conduct-estimator}
\end{equation}
and we report a $95\%$ confidence set obtained by cluster
subsampling at the $(i,t)$ level \citep{RomanoShaikh2008}, which
respects the within-interval dependence between deviations drawn
from the same focal battery. A value $\widehat\lambda=0$ means that
the owner-level hypothesis~\eqref{eq:conduct-hypothesis-own}
already minimises the unrationalised contradictions and no positive
$\lambda$ improves on it. A value $\widehat\lambda>0$ means that
adding the same-provider spillover with weight $\widehat\lambda$
leaves fewer contradictions than $\lambda=0$, and the magnitude of
$\widehat\lambda$ measures how much of the spillover the observed
bid behaves as if it internalises. In finite samples the minimiser
can sit at a boundary even when the objective is nearly flat. We
therefore treat $\widehat\lambda=1$ as evidence of same-provider
conduct only when the reduction in $Q_n(\widehat\lambda)$ relative
to $Q_n(0)$ is economically meaningful and the subsampling
confidence set excludes zero.

We estimate the conduct parameter as a function of the focal
provider's near-margin capacity share in the clearing interval. Let
$s^{\mathrm{prov}}_{k_i,r_i,t}$ denote provider $k$'s share of
near-margin battery MW in region $r$ at interval $t$, defined as
the MW within $\pm 1$ price band of the clearing price operated by
batteries served by $k_i$, divided by total near-margin battery MW
in $(r_i,t)$. This is a standard single-firm share but applied to
the near-margin capacity of batteries only. Hence, this is not 
equivalent to the total market share which would need to account for 
competing capacity from gas, hydro, and other dispatchable plants. The
per-interval near-margin share tracks the installed-capacity
provider share reported in Section~\ref{subsec:empirical-panel}
on average: across the regions and providers in our sample, the
median per-interval near-margin share lies within a few percentage
points of the installed share, with per-interval dispersion that
reflects which batteries are dispatchable at the margin in that
interval (Appendix~\ref{app:installed-vs-nearmargin}). The
near-margin window is the same concentration object measured at a
finer (per-interval) resolution, not a separate measure. For each
window centre $s_0$ on a grid in $[0.20, 0.85]$ we fit a local
kernel-weighted version of~\eqref{eq:conduct-estimator},
\begin{equation}
\widehat\lambda(s_0)
:=
\arg\min_{\lambda\in[0,1]}
\sum_{(i,t,d)} K_h\!\left(s^{\mathrm{prov}}_{k_i,r_i,t}-s_0\right)
\,w_{itd}\,
\max\!\left\{
\Delta\Pi_{o_i,t}(d)+\lambda\,\Delta\Pi_{k_i,-o_i,t}(d),\,0
\right\}^{2},
\label{eq:conduct-local-window}
\end{equation}
where $K_h$ is a triangular kernel with bandwidth $h=0.10$.
Confidence sets are constructed by cluster subsampling at the
$(i,t)$ level as for the unweighted moment objective, with the
same kernel weights applied inside each subsample draw. The
estimator $\widehat\lambda(s_0)$ traces out a continuous profile
of the conduct parameter across the identifying range of
$s^{\mathrm{prov}}$ without committing to a discrete bin
specification. Appendix~\ref{app:conduct-table} reports the
corresponding fit on $10$-percentage-point
$s^{\mathrm{prov}}$ bins as a complementary discrete view of the
same estimand.

The test requires that the focal provider's near-margin set
contain at least two distinct owners; otherwise the same-provider,
different-owner partner set is empty and
$\Delta\Pi_{k_i,-o_i,t}(d)\equiv 0$. We therefore restrict the
sample to battery-intervals with
$\mathrm{HHI}^{\mathrm{own}}_{k_i,r_i,t}<1$, where
$\mathrm{HHI}^{\mathrm{own}}_{k,r,t}=\sum_{o}(s^{\mathrm{own}}_{o,k,r,t})^{2}$
is the Herfindahl index of owner shares inside the provider's
near-margin set.
We further restrict the deviation sample along three dimensions.
First, we only keep focal batteries whose provider has at least one
same-region battery operated by a different firm, so that the
two-owner condition above can be met; in practice this confines the
test to batteries served by the two largest commercial providers,
$12$ Tesla Autobidder batteries spanning NSW, QLD, SA, and VIC and
$7$ Fluence Mosaic batteries spanning NSW, QLD, and VIC. Second, we
keep evening-peak intervals between 16:00 and 22:00 -- the window in
which Section~\ref{subsec:aem-synchronization} locates the
ERI-driven upper-band responsiveness and in which regional prices
are highest and most often set by batteries. Third, we drop
intervals in which the focal battery offers no movable MW just
above the clearing band, since we cannot construct a deviation 
for these intervals. For each interval we then construct nine deviations: three deviation
sizes ($10\%$, $25\%$, and $50\%$ of the relevant offered MW) for each
of three markets (energy, raise-FCAS, lower-FCAS). We retain only deviations
that are feasible in the realised stack, meaning that the source band has
movable MW and the bid-stack edit passes the reclearing-engine constraints.
The final conduct-test panel consists of $49{,}481$ feasible
deviations -- $46{,}677$ in energy and $2{,}804$ in FCAS -- over
$15{,}882$ focal battery-intervals.

\begin{figure}[h!]
\centering
\caption{Conduct-test sample balance}
\label{fig:conduct-sample-balance}
\includegraphics[width=\textwidth]{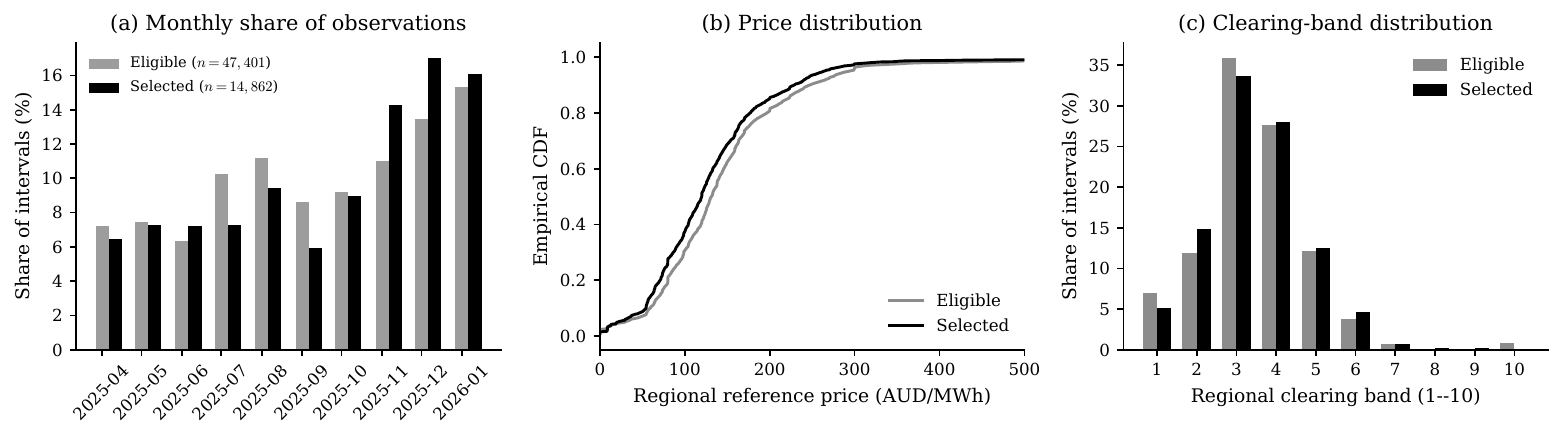}

\vspace{2pt}
{\footnotesize \textit{Notes:} Comparison of the selected conduct-test
panel with the eligible-but-not-selected evening-peak intervals.
Panel (a): share of selected and eligible intervals by
trading month. Panel (b): empirical CDF of the regional
reference price, capped at \$$500$/MWh for readability.
Panel (c): share of intervals by regional clearing band.}
\end{figure}

Figure~\ref{fig:conduct-sample-balance} compares the distribution
of intervals over time, of regional reference prices, and of
regional clearing-band placements between the deviation panel and
the universe of eligible evening-peak intervals. The two
distributions track each other closely on each dimension, so the
panel is not selectively over-sampling outcome-relevant states
within the test's identifying universe.
Appendix~\ref{app:conduct-sample-balance} reports the corresponding
numerical balance table together with the full evening-peak set of
all Tesla- and Fluence-managed battery-intervals for completeness.

Figure~\ref{fig:conduct-local-window} shows the main results of 
the conduct test using the local kernel-weighted estimator~\eqref{eq:conduct-local-window} as a
function of the focal-provider share on the
deviation panel. The top row shows the point estimate
$\widehat\lambda(s_0)$ at each window centre $s_0$ and the 
$95\%$ cluster-subsampling confidence set. The bottom row plots
the corresponding $Q$-reduction
$1-Q_n(\widehat\lambda)/Q_n(0)$, measuring the share of the
unrationalised deviation evidence at $\lambda=0$ that the conduct
hypothesis at $\widehat\lambda$ accounts for. Appendix
Figure~\ref{fig:conduct-Q-curves} shows the underlying
$Q_n(\lambda)/Q_n(0)$ curves on the discrete share-band view, together with
the FCAS panels, which we omit here. The same estimator applied to FCAS deviations 
finds no evidence of joint profit maximisation. However, this is because FCAS accounts
for about $0.1\%$ of these batteries' revenue in the evening peak, and enabled 
capacity clears at a median of well under one cent per MW per interval, 
so an equivalent deviation in the FCAS stack moves peer profits by almost nothing -- 
the total same-provider spillover across all $2{,}804$ FCAS deviations is under $\$5$, 
compared to $\$559{,}000$ in energy discharge. Thus, the conduct signal is present where 
the spillover is large and absent where it is near-zero, which is the same dose-response 
pattern that produces the concentration threshold.

\begin{figure}[h!]
\centering
\caption{Conduct parameter $\widehat\lambda$ with $95\%$
confidence set}
\label{fig:conduct-local-window}
\includegraphics[width=\textwidth]{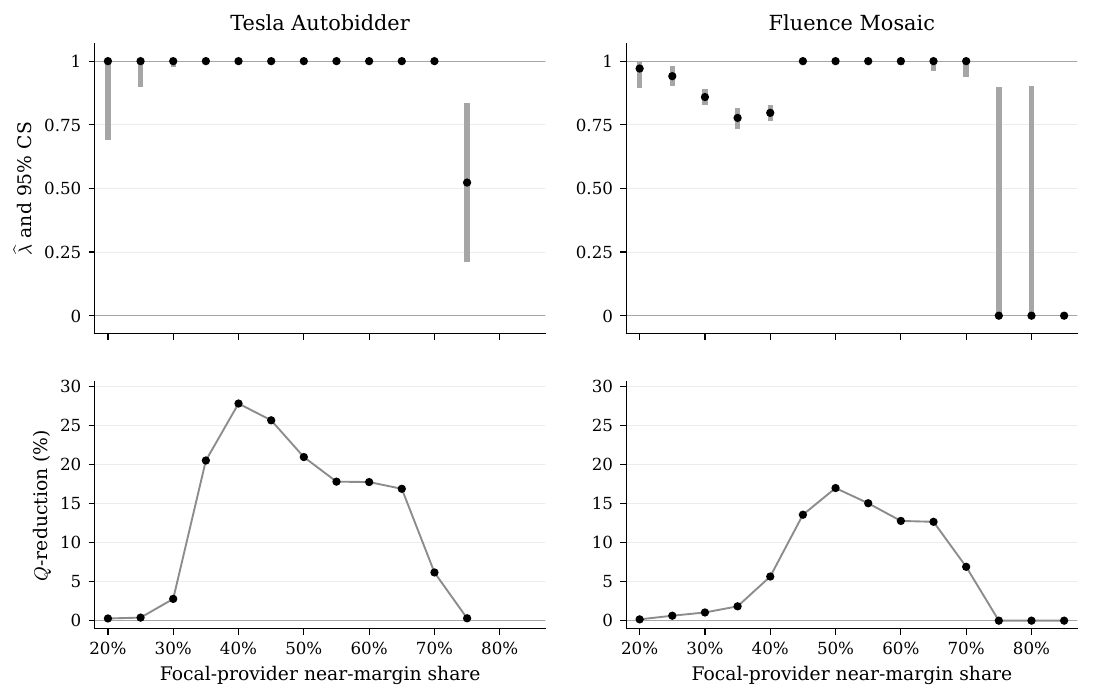}

\vspace{2pt}
{\footnotesize \textit{Notes:} Local-window kernel-weighted fit of
the two-tier conduct estimator on the deviation
sample, separately for Tesla Autobidder (left) and Fluence Mosaic
(right). Top row: point estimate $\widehat\lambda(s_0)$ at each
window centre $s_0$ with a vertical $95\%$ cluster-subsampling
confidence set; dotted reference lines mark 
$\lambda=0$ and $\lambda=1$. Bottom row: $Q$-reduction
$1-Q_n(\widehat\lambda)/Q_n(0)$ at the same window centres. Kernel
is triangular with bandwidth $h=0.10$. Windows with fewer than
$50$ multi-owner $(i,t)$ clusters are omitted.}
\end{figure}

For both Tesla Autobidder and Fluence Mosaic, $\widehat\lambda$ is
above $0.75$ and mostly at the corner $\widehat\lambda=1$ for provider
shares between $20\%$ and $70\%$. The estimates are informative where accounting
for the profit change in same-provider, different-owner batteries
meaningfully reduces $Q_n$, which holds for provider shares of
roughly $30$--$70\%$ for both providers. The $95\%$ subsampling
confidence sets exclude the owner-only null across this range.
Outside this range the estimates are uninformative, for different reasons at the two ends. Below $30\%$ the identifying sample is large enough that an
improvement in fit of the size found above $30\%$ would be detected if
same-provider internalisation were operating, yet the $Q$-reduction stays near
zero and then rises discontinuously at $30\%$. This points to a genuine
concentration threshold. Internalisation becomes detectable only once a
provider controls a sufficiently large share of near-margin capacity. 
Above $70\%$ the multi-owner identifying sample becomes too thin to
estimate reliably.

The results therefore consistently show that joint-profit
maximisation across same-provider batteries is a statistically
significant improvement in fit over owner-level profit maximisation
alone. Two features of the estimates are worth considering. First, 
the batteries in question are non-negligible price setters. Because 
the reclearing engine returns the counterfactual clearing price for 
every deviation, we can measure each battery's price impact directly 
rather than inferring it from an estimated residual demand curve. A 
single battery's most profitable owner-only deviation moves the regional 
price by $\$1.28$/MWh at the median, $\$2.70$/MWh on average and 
$\$6.24$/MWh in the upper decile, reaching roughly $\$59$/MWh at the maximum 
(Appendix~\ref{app:conduct-cost}). Structural concentration screens 
exist to proxy for exactly this effect that we can directly calculate 
in our setting. In fact, the conduct test
rejects the owner-only null at focal-provider near-margin shares
between $30\%$ and $70\%$. Because the per-interval near-margin
share tracks the installed share on average, this identifying range
corresponds to provider installed-capacity shares of approximately
$20$--$60\%$, which is inside the range the classic
wholesale-market-power literature treats as material unilateral
conduct~\citep{BorensteinBushnellWolak2002,Mansur2007}
Second, the profit internalisation arises across ownership. In 
the identifying intervals the focal provider's near-margin capacity is 
held by $2.2$ separately owned firms on average, and $83\%$ of the 
same-provider spillover that the test finds internalised accrues 
to batteries owned by a different firm than the deviating one -- 
precisely the pairs that ownership-based concentration analysis 
treats as independent competitors. The two partitions also differ in 
level: in QLD1 and VIC1, which together supply $73\%$ of the 
identifying intervals, provider concentration among batteries is 
more than twice owner concentration. The identifying range of 
$30$--$70\%$ near-margin share corresponds to roughly $20$--$60\%$ 
of installed battery capacity, of the same order as the single-firm 
shares that trigger scrutiny under standard screens \citep{FERCOrder697},
though we stress that our shares are computed within the battery segment 
and the comparison is a yardstick rather than a market-power finding

We now investigate possible alternative explanations for the
results of the conduct test, and rule them out. All results of the robustness exercises 
are documented in Appendix~\ref{app:conduct-robustness}.
First, the conduct test treats a positive recleared owner-portfolio gain
as a contradiction of owner-level optimisation. A natural
alternative explanation is that the estimated owner payoff omits private
contract positions, tolling arrangements, degradation shadow costs,
or other unobserved owner-level costs that make the owner-only
deviations privately unattractive. We bound this omitted cost for
each identifying band by solving for the smallest constant private
cost $c$ (in AUD per MWh of focal discharge) such that adding
$c\cdot\mathrm{exposure}_{itd}$ to the focal owner's deviation gain
reduces the moment objective as much as the conduct hypothesis at
$\widehat\lambda$.

The implied charges are well above the Bellman opportunity-cost
benchmark of about \$$11$/MWh, which is the largest movement in
$\widehat m_{it}$ across our numerical sensitivity checks. For
Tesla Autobidder, two of the four identifying bands ($40$--$50\%$ and
$60$--$70\%$) even require an infinite cost -- the
irreducible part of the moment objective from positive own-gain
rows with no corresponding dispatch exposure already exceeds
$Q_n(\widehat\lambda)$, so no finite per-MWh cost can rationalise
the data -- and the other two require \$$16$ and \$$59$/MWh, respectively. For
Fluence, the three corner intervals at $\widehat\lambda=1$ require
\$$11$, \$$15$, and \$$28$/MWh; the $30$--$40\%$ interior interval, the
weakest in the identifying set, requires \$$8$/MWh. This is the only
interval that sits below the Bellman benchmark. In every other
interval the implied cost is at least a factor above the
benchmark or infinite.

Second, the deviation panel we construct and the reclearing calculation 
determine deviations and their value at the realised market outcome, instead 
of the expected market outcome prior to the 5-minute interval. An apparent owner-level
violation could therefore in principle reflect hindsight -- a deviation that looks
profitable ex post but not under the information the bidder held when the
operative bid was submitted. To rule out this possibility, we date each interval's operative
quantity bid to its submission timestamp in the data and split the test by the lead time
to dispatch. If the results are driven by hindsight, then they should be strongest for
bids submitted well ahead of dispatch and weakest for bids submitted shortly 
before the dispatch interval. We find the opposite. The conduct test is strongest 
for rebids submitted within the final fifteen minutes prior to dispatch -- when five-minute
predispatch forecasts can be expected to be accurate and the no-profitable-deviation logic
applies. 

Third, a positive $\widehat\lambda$ could reflect a purely mechanical price 
externality rather than same-provider internalisation. Any deviation that lowers 
the clearing price reduces the inframarginal revenue of every battery near 
the margin, so same-provider profits may simply fall 
because same-provider batteries happen to be near the margin. Because the 
reclearing engine returns the profit change each deviation imposes on every 
regional battery, we can test this possibility directly. For each focal deviation we build 
placebo coalitions from the different-provider, different-owner batteries near 
the margin in the same region and interval, matched to the actual same-provider 
coalition on the number of batteries and on installed capacity, so that they 
carry the same mechanical price exposure by construction. Across $500$ capacity-matched 
placebo coalitions -- different-provider batteries carrying the same near-margin 
price exposure -- not one improves the fit of the conduct objective as much as 
the actual same-provider coalition ($p=0.002$ for both providers), and the 
same-provider improvement exceeds the $95$th percentile of the placebo 
distribution. Moreover, if we do not match on capacity and construct a simpler 
count-matched placebo, we continue to obtain the same result ($p\le0.010$).
This comparison also bounds owner-level objectives more complex than 
the per-MWh constant cost bound above. Some owners may hold 
generation portfolios and contract positions whose value rises with 
the regional price, an incentive denominated in price exposure 
rather than in discharge volume. Because any price-denominated 
objective scales with the recleared price change, a capacity-matched 
placebo coalition -- which carries the same price exposure by 
construction -- would rationalise the observed restraint as well as 
the same-provider coalition does. It does not.
The conduct test results are therefore specific to the provider coalition, 
not a generic consequence of near-margin batteries losing revenue when the price falls.

Finally, we verify that the estimator recovers internalisation where theory requires it:
we estimate the weight the focal owner places on its own other batteries,
instead of imposing it. We indeed find a substantial improvement in fit and an 
estimated weight on same-owner batteries of approximately one,
confirming that the method is able to detect joint-profit incentives when they should be
present by construction. 

Taken together, these robustness exercises  
leave same-provider profit internalisation as the most likely interpretation of the 
observed bidding behaviour and results of the conduct test. We now close 
by calculating the magnitude of the impact this conduct has. We estimate it 
to be around \$$5.5$m per year. We obtain this from the deviation
panel: for each focal battery-interval in the identifying range, we
take the owner-only best deviation $d^0_{it}$ -- the deviation that
maximises the focal owner's recleared portfolio gain -- and record
the counterfactual regional price the reclearing engine returns. The
gap between the observed and counterfactual price is the price increase that the owner's restraint sustains. We multiply it by regional
demand for the 5-minute interval, sum across identified $(r,t)$, and
re-weight by inverse inclusion probability, giving \$$4.6$m over the
sample window and \$$5.5$m annualised, of which Tesla Autobidder
accounts for \$$3.1$m and Fluence Mosaic for \$$2.4$m.

This is likely a conservative bound on the cost of the identified
conduct. The counterfactual is partial-equilibrium and unilateral:
we compute the price increase from a single owner-only best
deviation while holding all other bid stacks fixed, whereas fully
coordinated same-provider bidding would move more capacity and raise
prices further. It also counts only discrete deviations on the
$\{10\%,25\%,50\%\}$ MW grid, and only within the identifying
concentration bands. Finally, the magnitude is tied to the current
fleet. Batteries account for $46\%$ of AEMO's record $64$~GW 
December-2025 connection pipeline \citep{AEMO2026DecemberQuarterPipeline}, 
several times the $6.6$~GW of known-provider capacity in our sample, 
so if provider shares persist as this capacity is commissioned, 
batteries will set the regional price in more intervals and the 
same per-interval conduct would scale the annual cost accordingly.

\section{Conclusion}\label{sec:conclusion}

Whether and how algorithms may facilitate collusion is a central question in current
competition policy debate, but direct empirical evidence is still scarce.
\citet{AssadClarkErshovXu2024} show that the adoption of pricing algorithms
raised margins among competing retail fuel stations, and
\citet{CalderWangKim2026} find that a common pricing algorithm in rental housing
internalises competitors' profits.\footnote{A related empirical literature studies simpler repricing
software that adjusts a posted price in response to rivals
\citep{BrownMacKay2023,Musolff2025algorithmic}.} We provide, to our knowledge, the first evidence 
of the effects of common third-party \emph{bidding} algorithms, which construct an
entire offer schedule from primitives rather than repricing a single posted
price, and of concentration in the upstream layer
of algorithm providers. We argue that the mechanism is not specific
to batteries or electricity markets: where competitors delegate high-frequency decisions
to a shared algorithmic provider, that provider maps common information into
correlated conduct. We develop a simple stylised model to capture our setting  
of common third-party providers with dynamic storage and show that previously documented 
price and welfare effects from the literature arise in our setting as 
well \citep{HarringtonOrtner2025, AleksenkoMiklosThal2026}.

To conduct our analysis, we hand-collect a mapping
of autobidding providers to utility-scale batteries in the
Australian National Electricity Market and combine it
with public 5-minute bid, price, and dispatch data, as well as 
a reclearing engine that reproduces the
regional clearing outcome from public market inputs and allows us to 
construct counterfactual market outcomes. We document substantial 
provider concentration, and then examine how this concentration affects 
downstream market conduct.

First, we find that bids constructed by the same provider move
together more than bids constructed by different providers, and
that the gap widens after the July~2025 ERI disclosure reform. The
pattern is provider-specific: rebid timing has a provider-wide
component, while bid-stack co-movement is strongest when
same-provider batteries operate in the same region and face the
same local scarcity state. The result is robust to alternative
samples and specifications and is consistent with a common
provider's algorithm converting shared scarcity information into
correlated bids. We also find that the disclosure reform led to
higher bids during periods of scarcity and lower bids during
periods of abundance, showing the effect of the improved information
about the market state.

Second, we relate this co-movement to regional provider
concentration and find that both local co-movement and the
within-day variance of the regional average bid price rise with the
installed provider HHI.
These reduced-form patterns are consistent with responsive bidding
but cannot, on their own, separate that channel from coordinated
conduct. We therefore turn to a structural test of the latter.

Third, to test for conduct, we estimate a Bellman opportunity-cost value
for each battery, generate local feasible deviations from observed
bid stacks, and reclear the market under each. The resulting profit
changes define a revealed-preference moment inequality whose
conduct parameter $\lambda\in[0,1]$ measures the weight a battery's
bidder places on same-provider, different-owner profits in its
effective objective. For the dominant autobidding providers
Tesla Autobidder and Fluence Mosaic, $\widehat\lambda$ is at or
above $0.75$ across focal-provider near-margin shares of roughly
$30$--$70\%$, and the $95\%$ subsampling confidence sets exclude
the owner-only null throughout this range. Below $30\%$ the test
retains power but finds no internalisation, so the conduct emerges
only above a concentration threshold; above $70\%$ the multi-owner
sample the test requires becomes too thin. 
Our conduct test shows that these batteries move prices: a single 
battery's best owner-only deviation shifts the regional clearing 
price by $\$1.28$/MWh at the median and by more than $\$6$/MWh 
in the upper decile. And the internalisation crosses ownership
boundaries -- $83\%$ of the spillover it captures accrues to 
batteries owned by other firms, which ownership-based concentration 
analysis treats as independent competitors. FCAS deviations, by contrast, show
no conduct signal.

A set of robustness exercises shows that this is not an artefact
of measurement or of the test's construction. One, the unobserved
owner-level cost required to rationalise the observed restraint under
owner-only incentives exceeds our Bellman opportunity-cost benchmark
in every identifying band. Two, the 
conduct signal is concentrated in rebids submitted within fifteen
minutes of dispatch, when the realised outcome closely approximates
the bidder's information set, so it does not reflect hindsight. Three, it is
specific to the provider coalition: none of $500$ placebo coalitions
of different-provider batteries, matched on battery count and
installed capacity and therefore carrying the same mechanical price
exposure, improves the conduct fit as much as the actual
same-provider coalition. Four, we confirm that the
estimator correctly recovers full internalisation of profits across batteries
of the same owner. We conclude by calculating that the identified withholding 
implies an annualised consumer-side cost of
approximately \$$5.5$~million.

Our findings have implications beyond this market. The NEM is one
of several wholesale electricity markets in which utility-scale
batteries are expanding rapidly and a small number of third-party
autobidder providers operate a growing share of installed capacity.
As algorithmic bidding spreads, the relevant unit of concentration may
become the provider whose software constructs the bid, not only the
asset owner. Our results also speak to the disclosure
debate in these markets: the July~2025 reform improved the scarcity
information batteries can condition on, but the same richer signal
let a concentrated provider's bids track the common state more
tightly. Improved disclosure and stronger same-provider co-movement
are two sides of one channel, and weighing them is a policy
trade-off. 

\clearpage

\bibliographystyle{plainnat}
\bibliography{bib}

\newpage
\appendix
\counterwithin{figure}{section}
\counterwithin{table}{section}
\renewcommand\thefigure{\thesection\arabic{figure}}
\renewcommand\thetable{\thesection\arabic{table}}

\section{Formal Model Details}\label{app:formal-model}

This appendix provides the formal details behind the stylised model in
Section~\ref{subsec:mechanisms}. We define the
strategy and belief objects suppressed in the main text, show the local
responsive-bidding comparative static, and record a simple static affine
equilibrium that illustrates the concentration channel.

\subsection{Beliefs, Strategies, and Local Responsive Bidding}

Let an owner-level bidding strategy be a map
\[
\sigma_o:\mathcal I_{ot}\to\mathcal B_o(S_t,X_{rt}),
\]
and let $\mu_{ot}$ denote owner $o$'s belief over current and future scarcity,
rival information, and rival bid strategies. The compact expectation
$\E_{ot}[\cdot]$ in the main text is shorthand for expectation under
$(\mu_{ot},\sigma_{-o})$. Thus the owner-level responsive benchmark is
\[
\sigma_o(\mathcal I_{ot})
\in
\arg\max_{\mathbf b_o\in\mathcal B_o(S_t,X_{rt})}
\E_{\mu_{ot},\sigma_{-o}}
\left[
\Pi_{ot}(\mathbf b_o,\sigma_{-o}(\mathcal I_{-o,t});
S_t,X_{rt},\{\theta_{r\tau}\}_{\tau\ge t})
\mid\mathcal I_{ot}
\right].
\]

For the responsive-bidding comparative static below we require
two regularity conditions: monotonicity of the continuation-value gap in the
scarcity state, and monotonicity of the conditional distribution of
next-period scarcity in the provider's current posterior.

\begin{assumption}[Monotone storage value]\label{ass:app-monotone-storage}
For a fixed public state \(z\) and provider \(k\), let
\[
\Delta V_k(z,\theta)
:=
V_k^F(z,\theta)-V_k^E(z,\theta)
\]
denote the continuation-value gain from retaining stored energy rather than
depleting the battery. The map
\(\theta\mapsto\Delta V_k(z,\theta)\) is weakly increasing.
\end{assumption}

\begin{assumption}[Monotone scarcity posterior]\label{ass:app-monotone-posterior}
For each provider $k$ and date $t$, the conditional distribution of
$\theta_{r,t+1}$ given $m_{krt}=m$ is weakly increasing in $m$ in the
first-order stochastic dominance sense.
\end{assumption}

The provider's posterior shadow cost of discharge is
\begin{equation}
\widehat c_k(z,m)
:=
c_d
+
\delta\,
\E\!\left[
\Delta V_k(z,\theta_{r,t+1})
\,\middle|\,
z_t=z,\ m_{krt}=m
\right],
\label{eq:app-shadow-cost}
\end{equation}
where $c_d$ is the marginal throughput cost of discharge and $\delta$ is the
discount factor.

Given these two assumptions, we can then state the following Lemma, which directly delivers the co-movement of same-provider batteries.

\begin{lemma}[Posterior shadow cost monotonicity]\label{lem:app-shadow-cost}
Under Assumptions~\ref{ass:app-monotone-storage}--\ref{ass:app-monotone-posterior},
the posterior shadow cost $\widehat c_k(z,m)$ is weakly increasing in $m$.
\end{lemma}

\begin{proof}
For any $m'\ge m$, Assumption~\ref{ass:app-monotone-posterior} implies that the
conditional distribution of $\theta_{r,t+1}$ given $m_{krt}=m'$ first-order
stochastically dominates the conditional distribution given $m_{krt}=m$. Since
$\theta\mapsto\Delta V_k(z,\theta)$ is weakly increasing by
Assumption~\ref{ass:app-monotone-storage}, the conditional expectation in
\eqref{eq:app-shadow-cost} is weakly increasing in $m$. Adding $c_d$ and
multiplying by $\delta\ge 0$ preserve monotonicity.
\end{proof}

To derive the local response to scarcity, take one near-margin generation block
of battery $i$ and hold the rest of the stack fixed. Let $b_i$ denote the offer
price for this block and let
\(q_i(b_i;\theta_{rt},\mathbf b_{-i},X_{rt})\) denote dispatched MW. Around the
clearing margin, a higher offer price weakly lowers dispatch, so
\(q_{i,b}\le0\). The local owner payoff can be written as
\begin{equation}
U_i(b_i;m_{krt})
=
\E\!\left[
\pi_i(b_i;\theta_{rt},\mathbf b_{-i},X_{rt})
\mid \mathcal I_{krt}
\right]
+
\delta\,
\E\!\left[
V_{i,t+1}\!\left(
S_{it}-\Delta q_i(b_i;\theta_{rt},\mathbf b_{-i},X_{rt}),
\theta_{r,t+1}
\right)
\mid \mathcal I_{krt}
\right].
\label{eq:app-local-payoff}
\end{equation}
At an interior optimum,
\[
U_{i,b}(b_i^R;m_{krt})=0,
\qquad
U_{i,bb}(b_i^R;m_{krt})<0.
\]
The implicit-function theorem gives
\[
\frac{\partial b_i^R}{\partial m_{krt}}
=
-
\frac{U_{i,bm}}{U_{i,bb}}.
\]
Lemma~\ref{lem:app-shadow-cost} gives the storage part of the cross-partial:
a higher provider scarcity posterior raises the expected shadow cost of
discharge. Since \(-\Delta q_{i,b}\ge0\), a higher bid preserves energy by
reducing current dispatch. When this dynamic force is not overturned by the
current-profit term, \(U_{i,bm}>0\). Hence
\[
\frac{\partial b_i^R}{\partial m_{krt}}>0.
\]

The same notation also gives the covariance logic in the main text. If
\[
b_{it}^R\approx \bar b_{it}+\psi_i(S_{it},X_{rt})+\phi_{it}m_{k_i,rt},
\qquad \phi_{it}>0,
\]
then same-provider pairs load on the same provider-level object, whereas
different-provider pairs load on distinct provider-level objects. Same-provider
bid covariance exceeds cross-provider covariance whenever the provider-level
object contains variation shared within provider but not across providers.

\subsection{A Static Affine Benchmark}

This subsection develops a static equilibrium benchmark.
This closed-form benchmark uses Gaussian dynamics and signals. Under
these assumptions the regularity conditions of Appendix~A.1 hold automatically
and the affine equilibrium has a closed form.

Let $(\theta_{r\tau})_{\tau\ge t}$ be a stationary mean-zero Gaussian process
with $\Var(\theta_{r\tau})=\tau_\theta^{-1}$ and
$\Cov(\theta_{rt},\theta_{r,t+1})=\rho_\theta/\tau_\theta$ for some serial
correlation $\rho_\theta\in[0,1)$. In each interval, all providers observe the
public signal
\[
g_t=\theta_{rt}+\eta_{0t},
\]
and provider $k$ additionally observes the private signal
\[
y_{kt}=\theta_{rt}+\eta_{kt},
\]
where $\eta_{0t}\sim\mathcal N(0,\tau_g^{-1})$,
$\eta_{kt}\sim\mathcal N(0,\tau_y^{-1})$, and all noise terms are mutually
independent and independent of $(\theta_{r\tau})$. Provider $k$'s
contemporaneous posterior mean is
\begin{equation}
m_{kt}:=\E[\theta_{rt}\mid g_t,y_{kt}].
\label{eq:short-posterior}
\end{equation}

\begin{lemma}[Gaussian posterior structure]\label{lem:app-posterior}
Let $T:=\tau_\theta+\tau_g+\tau_y$. Then
\begin{equation}
m_{kt}
=
\frac{\tau_g g_t+\tau_y y_{kt}}{T},
\qquad
\mathcal V_t:=\Var(m_{kt})
=
\frac{\tau_g+\tau_y}{\tau_\theta T},
\label{eq:app-posterior-mean}
\end{equation}
and, for $j\neq k$,
\begin{equation}
\rho_t
:=
\frac{\Cov(m_{kt},m_{jt})}{\Var(m_{kt})}
=
1-\frac{\tau_\theta\tau_y}{(\tau_g+\tau_y)T}
\in(0,1),
\qquad
\E[m_{jt}\mid m_{kt}]=\rho_t m_{kt}.
\label{eq:app-V-rho}
\end{equation}
There exist mutually independent Gaussian random variables $m_{gt}$,
$\widetilde h_t$ (both common across providers), and $u_{kt}$
(provider-specific) such that
\[
m_{kt}=m_{gt}+\widetilde h_t+u_{kt}.
\]
The conditional distributions
\[
\theta_{rt}\mid m_{krt}=m\;\sim\;\mathcal N(m,\,1/T),
\qquad
\theta_{r,t+1}\mid m_{krt}=m\;\sim\;\mathcal N(\rho_\theta m,\,\sigma_+^2),
\]
have variances that do not depend on $m$, so both are weakly first-order
stochastically increasing in $m$ (strictly so for $\theta_{r,t+1}$ when
$\rho_\theta>0$). In particular,
Assumption~\ref{ass:app-monotone-posterior} holds.
\end{lemma}

\begin{proof}
Gaussian updating gives the posterior mean in
\eqref{eq:app-posterior-mean}. Since the posterior variance is $1/T$,
$\Var(m_{kt})=\Var(\theta_{rt})-1/T=(\tau_g+\tau_y)/(\tau_\theta T)$.

Let
\[
m_{gt}:=\E[\theta_{rt}\mid g_t],
\qquad
r_t:=\theta_{rt}-m_{gt},
\qquad
\widetilde h_t:=\frac{\tau_y}{T}r_t,
\qquad
u_{kt}:=\frac{\tau_y}{T}\eta_{kt}.
\]
Then $m_{kt}=m_{gt}+\widetilde h_t+u_{kt}$. The residual $r_t$ is orthogonal to
$g_t$, and hence independent of $m_{gt}$ by Gaussianity; the
provider-specific term $u_{kt}$ depends only on $\eta_{kt}$, which is
independent of all other primitives. This proves the independent decomposition.
For $j\neq k$, $\Cov(m_{kt},m_{jt})=\Var(m_{gt}+\widetilde h_t)$, which gives
\eqref{eq:app-V-rho}. Joint normality and zero means imply
$\E[m_{jt}\mid m_{kt}]=\rho_t m_{kt}$.

For the conditional distributions, write $\theta_{rt}=m_{kt}+r_{kt}$ where the
posterior residual $r_{kt}\sim\mathcal N(0,1/T)$ is independent of $m_{kt}$;
this gives the first conditional distribution as a location family in $m$. For
$\theta_{r,t+1}$, joint normality of $(\theta_{r,t+1},m_{kt})$ gives
\[
\Cov(\theta_{r,t+1},m_{kt})
=
\frac{\tau_g+\tau_y}{T}\Cov(\theta_{r,t+1},\theta_{rt})
=
\frac{(\tau_g+\tau_y)\rho_\theta}{\tau_\theta T}
=
\rho_\theta\mathcal V_t.
\]
Linear projection then yields $\E[\theta_{r,t+1}\mid m_{kt}=m]=\rho_\theta m$
with conditional variance
$\sigma_+^2=\Var(\theta_{r,t+1})-\rho_\theta^2\mathcal V_t$ independent of $m$.
Both conditional distributions are location families in $m$ with $m$-invariant
spread, hence weakly first-order stochastically increasing in $m$.
\end{proof}

Now consider one region and one interval. Providers are indexed by $k=1,\dots,K$,
and provider $k$ controls share $x_k\ge0$, with $\sum_k x_k=1$.
All batteries using provider $k$ receive the same posterior $m_{kt}$ from
Lemma~\ref{lem:app-posterior} and choose the same scalar local action
$a_{kt}$. Let
\[
A_t:=\sum_{k=1}^K x_k a_{kt}
\]
denote the share-weighted aggregate action. We use the posterior variance
$\mathcal V_t$ and the cross-provider posterior correlation $\rho_t$ from
Lemma~\ref{lem:app-posterior} throughout.

The reduced payoff of a provider-\(k\) client from choosing action \(a\) is
\begin{equation}
\Pi_{kt}(a\mid m_{kt})
=
(a-c_t)
\Bigl(
\psi_t+\xi_t m_{kt}-\beta_t a
+d_t\,\E[A_t\mid m_{kt}]
\Bigr),
\label{eq:short-payoff}
\end{equation}
where \(\beta_t>0\), \(0<d_t<\beta_t\), and \(\xi_t>0\). The coefficient
\(\xi_t\) summarises the local scarcity loading of the reduced bid problem,
including the dynamic opportunity-cost channel. The coefficient \(d_t\) captures
strategic complementarity in local bid construction.

\begin{definition}[Affine delegated equilibrium]\label{def:app-affine}
An affine delegated equilibrium is a vector
$(\alpha_{kt},\gamma_{kt})_{k=1}^{K}$ such that every
client of provider $k$ chooses
\begin{equation}
a_{kt}=\alpha_{kt}+\gamma_{kt}m_{kt}
\label{eq:app-affine-rule}
\end{equation}
and this action is optimal under \eqref{eq:short-payoff}.
\end{definition}

\begin{proposition}[Affine provider-posterior benchmark]
\label{prop:short-model}
There exists a unique affine delegated equilibrium. The intercepts coincide
across providers, $\alpha_{kt}=\alpha_t$ for all $k$, with
\begin{equation}
\alpha_t=\frac{\psi_t+\beta_t c_t}{2\beta_t-d_t}.
\label{eq:app-affine-intercept}
\end{equation}
Let
\[
\mathcal S_t:=\sum_{k=1}^K x_k\gamma_{kt},
\qquad
\Lambda_t(x)
:=
\sum_{k=1}^K
\frac{x_k}{2\beta_t-d_t(1-\rho_t)x_k}.
\]
Then
\begin{equation}
\mathcal S_t
=
\frac{\xi_t\Lambda_t(x)}{1-d_t\rho_t\Lambda_t(x)},
\qquad
\gamma_{kt}
=
\frac{\xi_t}{1-d_t\rho_t\Lambda_t(x)}
\frac{1}{2\beta_t-d_t(1-\rho_t)x_k}.
\label{eq:app-affine-solution}
\end{equation}
\end{proposition}

\begin{proof}[Proof of Proposition~\ref{prop:short-model}]
Fix \(t\) and suppress time subscripts. Under an affine conjecture
\[
a_k=\alpha_k+\gamma_km_k,
\qquad
A=\sum_{j=1}^Kx_j(\alpha_j+\gamma_jm_j).
\]
Let \(\bar\alpha:=\sum_jx_j\alpha_j\) and
\(S:=\sum_jx_j\gamma_j\). From Lemma~\ref{lem:app-posterior},
\[
\E[m_j\mid m_k]=\rho m_k
\qquad (j\neq k).
\]
Therefore
\begin{equation}
\E[A\mid m_k]
=
\bar\alpha+\bigl[\rho S+(1-\rho)x_k\gamma_k\bigr]m_k.
\label{eq:app-conditional-A}
\end{equation}

For a provider-\(k\) client, the objective in \eqref{eq:short-payoff} is
strictly concave in \(a\). The first-order condition is
\[
2\beta a
=
\psi+\beta c+\xi m_k+d\,\E[A\mid m_k].
\]
Substituting \eqref{eq:app-conditional-A} and matching constant and slope terms
gives
\begin{align}
2\beta\alpha_k&=\psi+\beta c+d\bar\alpha,
\label{eq:app-alpha-system}\\
2\beta\gamma_k&=\xi+d\bigl[\rho S+(1-\rho)x_k\gamma_k\bigr].
\label{eq:app-gamma-system}
\end{align}
The right-hand side of \eqref{eq:app-alpha-system} is common across providers,
so all intercepts coincide. With \(\bar\alpha=\alpha\),
\[
(2\beta-d)\alpha=\psi+\beta c,
\]
which proves \eqref{eq:app-affine-intercept}.

For the slopes, define
\[
D_k:=2\beta-d(1-\rho)x_k,
\qquad
\Lambda(x):=\sum_{k=1}^K\frac{x_k}{D_k}.
\]
Rearranging \eqref{eq:app-gamma-system} gives
\[
D_k\gamma_k=\xi+d\rho S,
\qquad
\gamma_k=\frac{\xi+d\rho S}{D_k}.
\]
Multiplying by \(x_k\) and summing over \(k\) yields
\[
S=(\xi+d\rho S)\Lambda(x),
\]
so
\[
S=\frac{\xi\Lambda(x)}{1-d\rho\Lambda(x)}.
\]
Substitution back into the expression for \(\gamma_k\) gives
\eqref{eq:app-affine-solution}.

It remains only to verify that the denominator is positive. Since
\(0<d<\beta\) and \(x_k\le1\),
\[
D_k\ge 2\beta-d(1-\rho)>0,
\qquad
\Lambda(x)\le \frac{1}{2\beta-d(1-\rho)}.
\]
Hence
\[
d\rho\Lambda(x)
\le
\frac{d\rho}{2\beta-d(1-\rho)}
<1,
\]
because \(2\beta-d(1-\rho)-d\rho=2\beta-d>0\). The affine system is therefore
well defined and unique.
\end{proof}

Define
\[
B_{2t}:=\sum_{k=1}^K x_k\gamma_{kt}^2,
\qquad
Q_t:=\sum_{k=1}^K x_k^2\gamma_{kt}^2,
\qquad
\Delta_t^2:=\sum_{k=1}^K x_k(a_{kt}-A_t)^2.
\]

In the affine delegated equilibrium, the aggregate action moments are
\begin{align}
\Cov(A_t,\theta_t)&=\mathcal V_t \mathcal S_t,\label{eq:app-cov}\\
\Var(A_t)&=\rho_t\mathcal V_t \mathcal S_t^2+(1-\rho_t)\mathcal V_t Q_t,\label{eq:app-var}\\
\E[\Delta_t^2]&=\mathcal V_t\bigl[B_{2t}-\rho_t \mathcal S_t^2-(1-\rho_t)Q_t\bigr].
\label{eq:app-disp}
\end{align}
These identities follow by substituting the affine rule into \(A_t\) and using
the posterior decomposition in Lemma~\ref{lem:app-posterior}.

\begin{proposition}[Concentration and responsiveness]\label{prop:app-concentration}
Fix \((\tau_g,\tau_y)\). Under any majorization-increasing perturbation of the
provider-share vector \(x\):
\begin{enumerate}[label=(\roman*)]
\item the unconditional mean action \(\E[A_t]=\alpha_t\) is invariant;
\item the common-response object \(\mathcal S_t\), and therefore \(\Cov(A_t,\theta_t)\),
is weakly increasing;
\item the common component \(\rho_t\mathcal V_t\mathcal S_t^2\) and total variance
\(\Var(A_t)\) are weakly increasing.
\end{enumerate}
Along the dominant-provider family
\[
x(x_D)=
\left(x_D,\frac{1-x_D}{K-1},\dots,\frac{1-x_D}{K-1}\right),
\qquad
x_D\in[1/K,1],
\]
the increases in parts (ii) and (iii) are strict for interior \(x_D>1/K\).
\end{proposition}

\begin{proof}[Proof of Proposition~\ref{prop:app-concentration}]
The intercept in \eqref{eq:app-affine-intercept} is independent of \(x\), and
\(\E[m_{kt}]=0\), so \(\E[A_t]=\alpha_t\).

Fix \((\tau_g,\tau_y)\), so \(\rho_t\) and \(\mathcal V_t\) are fixed. Let
\[
c_t^{\rho}:=d_t(1-\rho_t),
\qquad
f(x):=\frac{x}{2\beta_t-c_t^{\rho}x}.
\]
Since
\[
f''(x)
=
\frac{4\beta_t c_t^{\rho}}
{(2\beta_t-c_t^{\rho}x)^3}
>0,
\]
\(\Lambda_t(x)=\sum_k f(x_k)\) is Schur-convex by Karamata's inequality. The
map
\[
z\mapsto \frac{\xi_t z}{1-d_t\rho_t z}
\]
is increasing on the admissible domain, so \(\mathcal S_t\) is weakly increasing under
any majorization-increasing perturbation of \(x\). Equation~\eqref{eq:app-cov}
then implies that \(\Cov(A_t,\theta_t)\) is weakly increasing as well.

For the variance term, write
\[
Q_t
=
\frac{\xi_t^2}{(1-d_t\rho_t\Lambda_t(x))^2}
\sum_{k=1}^K
\frac{x_k^2}{(2\beta_t-c_t^{\rho}x_k)^2}.
\]
Let $g(x):=x/(2\beta_t-c_t^{\rho}x)$. Since $g$, $g'$, and $g''$ are positive
on the admissible domain, $(g^2)''=2(g'^2+g\,g'')>0$, so $g^2$ is convex. The
symmetric sum $\sum_k g(x_k)^2$ is therefore Schur-convex by Karamata's
inequality, and weakly increases under any majorization-increasing
perturbation of $x$. The multiplicative factor $(1-d_t\rho_t\Lambda_t(x))^{-2}$
is positive and increasing in $\Lambda_t(x)$, which itself weakly increases
under the same perturbation. A product of two non-negative quantities, each
weakly increasing along a given perturbation, is itself weakly increasing
along that perturbation, so $Q_t$ is weakly increasing.
Equation~\eqref{eq:app-var} then implies that total variance is weakly
increasing; the common component \(\rho_t\mathcal V_t\mathcal S_t^2\) rises because
\(\mathcal S_t\) rises.

Along the dominant-provider family, the convex primitives are strictly convex
and \(x_D\) moves strictly away from the equal-share vector, so the increases in
\(\mathcal S_t\), the common component, and total variance are strict for interior
\(x_D>1/K\).
\end{proof}

\subsection{Consumer Welfare and Provider Concentration}\label{app:welfare}

The concentration result in Proposition~\ref{prop:app-concentration} concerns the
aggregate action. To translate it into a statement about consumers we add a
price-formation map and a consumer-welfare standard. Throughout, $(\tau_g,\tau_y)$
are held fixed, so the posterior variance $\mathcal V_t$ and cross-provider
correlation $\rho_t$ from Lemma~\ref{lem:app-posterior} are fixed, and only the
provider-share vector $x$ varies.

\begin{assumption}[Price formation and consumer standard]\label{ass:welfare}
\leavevmode
\begin{enumerate}[label=(W\arabic*)]
\item \textup{(Affine near-margin pass-through.)} The regional clearing price is
affine and increasing in the aggregate near-margin action,
\[
P_t=\omega_0+\omega_1 A_t,\qquad \omega_1>0.
\]
\item \textup{(Exposure.)} Consumer load is weakly higher in scarcer states,
$D_t=D_0+\upsilon\,\theta_t$ with $\upsilon\ge0$ and $D_0$ large enough that
$D_t>0$ on the relevant range; only the moments
$\E[D_t]=D_0$ and $\Cov(D_t,\theta_t)=\upsilon/\tau_\theta$ enter below.
\item \textup{(Consumer loss.)} Period consumer loss is expected payment plus a
convex price-risk penalty,
\[
\mathcal L_t=P_tD_t+\tfrac{\kappa}{2}\bigl(P_t-\E P_t\bigr)^2,\qquad \kappa\ge0.
\]
\end{enumerate}
\end{assumption}

\begin{lemma}[Aggregate action moments]\label{lem:moments-welfare}
In the affine delegated equilibrium of Proposition~\ref{prop:short-model}, with
$\E[\theta_t]=0$ and $Q_t:=\sum_{k}x_k^2\gamma_{kt}^2$,
\[
\E[A_t]=\alpha_t,\qquad
\Cov(A_t,\theta_t)=\mathcal V_t \mathcal S_t,\qquad
\Var(A_t)=\mathcal V_t\bigl(\rho_t \mathcal S_t^2+(1-\rho_t)Q_t\bigr).
\]
\end{lemma}

\begin{proof}
By Lemma~\ref{lem:app-posterior} write $m_{kt}=w_t+u_{kt}$, where the common
component $w_t:=m_{gt}+\widetilde h_t$ satisfies $\Var(w_t)=\rho_t\mathcal V_t$ and
$\Cov(w_t,\theta_t)=\mathcal V_t$, and the $u_{kt}$ are mean-zero, mutually
independent, independent of $w_t$ and of $\theta_t$, with
$\Var(u_{kt})=(1-\rho_t)\mathcal V_t$ and $\Cov(u_{kt},\theta_t)=0$. Since
$a_{kt}=\alpha_t+\gamma_{kt}m_{kt}$ and $\sum_kx_k=1$,
\[
A_t=\alpha_t+\mathcal S_t\,w_t+\sum_{k}x_k\gamma_{kt}u_{kt},
\qquad \mathcal S_t=\sum_k x_k\gamma_{kt}.
\]
Taking expectations gives $\E[A_t]=\alpha_t$. Since $\Cov(u_{kt},\theta_t)=0$,
$\Cov(A_t,\theta_t)=\mathcal S_t\Cov(w_t,\theta_t)=\mathcal V_t \mathcal S_t$. By independence of
$w_t$ and $\{u_{kt}\}$,
$\Var(A_t)=\mathcal S_t^2\Var(w_t)+\sum_k x_k^2\gamma_{kt}^2\Var(u_{kt})
=\mathcal V_t\bigl(\rho_t \mathcal S_t^2+(1-\rho_t)Q_t\bigr)$.
These coincide with \eqref{eq:app-cov}--\eqref{eq:app-var}.
\end{proof}

\begin{proposition}[Concentration raises expected consumer loss]
\label{prop:welfare-concentration}
Under the affine delegated equilibrium and Assumption~\ref{ass:welfare}, expected
consumer loss decomposes as
\begin{equation}
\E[\mathcal L_t]
=\;L_0\;+\;
\underbrace{\omega_1\upsilon\,\mathcal V_t \mathcal S_t}_{\textnormal{(I) state amplification}}
\;+\;
\underbrace{\tfrac{\kappa}{2}\,\omega_1^2\,\mathcal V_t\bigl(\rho_t \mathcal S_t^2+(1-\rho_t)Q_t\bigr)}
_{\textnormal{(II) price variance}},
\label{eq:welfare-decomp}
\end{equation}
where $L_0:=\omega_0D_0+\omega_1D_0\alpha_t$ is independent of $x$. Terms
\textnormal{(I)} and \textnormal{(II)} equal $\Cov(P_t,D_t)$ and
$\tfrac{\kappa}{2}\Var(P_t)$, respectively. The mean price
$\E[P_t]=\omega_0+\omega_1\alpha_t$ is invariant to $x$, and along any
majorization-increasing (concentration-increasing) change in $x$ both
\textnormal{(I)} and \textnormal{(II)} weakly increase, so $\E[\mathcal L_t]$
weakly increases; the increase is strict along the dominant-provider family
$x(x_D)$ for interior $x_D\in(1/K,1)$.
\end{proposition}

\begin{proof}
First, by Proposition~\ref{prop:short-model},
$\alpha_t=(\psi_t+\beta_tc_t)/(2\beta_t-d_t)$ is independent of $x$, and
$\E[m_{kt}]=0$; hence $\E[A_t]=\alpha_t$ and $\E[P_t]=\omega_0+\omega_1\alpha_t$
are independent of $x$.

Second, using (W1)--(W2), $\E[\theta_t]=0$, and
Lemma~\ref{lem:moments-welfare},
\[
\E[P_tD_t]
=\E\bigl[(\omega_0+\omega_1A_t)(D_0+\upsilon\theta_t)\bigr]
=\omega_0D_0+\omega_1D_0\alpha_t+\omega_1\upsilon\,\Cov(A_t,\theta_t)
=L_0+\omega_1\upsilon\,\mathcal V_t\mathcal S_t,
\]
and $\omega_1\upsilon\,\mathcal V_t\mathcal S_t=\Cov(P_t,D_t)$.

Third, by (W3) and Lemma~\ref{lem:moments-welfare},
\[
\tfrac{\kappa}{2}\E\bigl[(P_t-\E P_t)^2\bigr]
=\tfrac{\kappa}{2}\Var(P_t)
=\tfrac{\kappa}{2}\omega_1^2\Var(A_t)
=\tfrac{\kappa}{2}\omega_1^2\mathcal V_t\bigl(\rho_t\mathcal S_t^2+(1-\rho_t)Q_t\bigr).
\]
Adding the two terms gives \eqref{eq:welfare-decomp}.

Finally, by Proposition~\ref{prop:app-concentration}(ii)--(iii), under a
majorization-increasing change in $x$ (with $\mathcal V_t,\rho_t$ fixed) both
$\mathcal S_t$ and $\Var(A_t)=\mathcal V_t(\rho_t\mathcal S_t^2+(1-\rho_t)Q_t)$ are weakly increasing,
strictly so along $x(x_D)$ for interior $x_D$. Since
$\omega_1>0$ and $\upsilon,\kappa,\mathcal V_t\ge0$, terms (I) and (II) are weakly
(strictly) increasing while $L_0$ is constant; hence $\E[\mathcal L_t]$ is weakly
(strictly) increasing.
\end{proof}

Proposition \ref{prop:welfare-concentration} shows that all welfare consequences of concentration are driven by 
the two second moments in \eqref{eq:welfare-decomp}.
If $\upsilon=\kappa=0$ then $\E[\mathcal L_t]=L_0$ is invariant to $x$: because the
mean action is invariant (Proposition~\ref{prop:app-concentration}(i)), provider
concentration is welfare-neutral for a risk-neutral consumer facing inelastic
demand and linear pricing.  Channel~(I) only needs demand to co-move with scarcity ($\upsilon>0$): concentration makes the clearing
price track scarcity more tightly ($\Cov(P_t,\theta_t)=\omega_1\mathcal V_t\mathcal S_t$
rises), so consumers pay more precisely when they consume more. Channel~(II)
additionally requires a convex/risk-averse consumer standard ($\kappa>0$): concentration raises
aggregate price variance, and specifically its common component
$\rho_t\mathcal V_t\mathcal S_t^2$.

Combining \eqref{eq:app-var}--\eqref{eq:app-disp} yields the identity
\[
\Var(A_t)=\mathcal V_tB_{2t}-\E[\Delta_t^2],
\qquad
\mathcal V_tB_{2t}=\sum_{k}x_k\Var(a_{kt}),
\]
so Channel~(II) equals
$\tfrac{\kappa}{2}\omega_1^2\bigl(\sum_k x_k\Var(a_{kt})-\E[\Delta_t^2]\bigr)$. That is,
the aggregate price risk borne by consumers is the share-weighted average
provider-level bid variance net of cross-provider dispersion $\E[\Delta_t^2]$.
Concentration raises the aggregate by reallocating variance from the diversifying
cross-provider dispersion into the common component. This is the counterpart, in
the auction setting, of the cross-seller price-correlation channel in
\citet{AleksenkoMiklosThal2026}. Channel~(I) is the corresponding counterpart of their
across-state price-dispersion channel.

\section{Empirical Appendix}\label{app:empirics}

\subsection{Synchronisation Detail by Market Family}\label{app:sync-details}

This appendix reports the synchronisation and concentration results for the
market families summarised in Section~\ref{subsec:aem-synchronization}.
Table~\ref{tab:sync-market-details-appendix} gives the pre- and post-ERI
pair-type correlations underlying the SPSR-minus-SPDR estimates of
Table~\ref{tab:sync-market-effects}, for the three market families outside
energy discharge: energy charge, raise FCAS, and lower FCAS. Reading the
levels alongside the differences shows where a relative increase reflects a
rise in same-region co-movement and where it reflects a flat or falling
same-provider different-region baseline; in FCAS the different-region
correlation is already high before the reform and changes little after it.
Table~\ref{tab:concentration-sync-variance-markets-appendix} repeats the
upper-band-price concentration exercise of
Table~\ref{tab:concentration-sync-variance} across the same market families.
The association between installed provider HHI and synchronisation is
concentrated in energy discharge, while the association with bid-price
variance is present in FCAS as well.

\begin{table}[!htbp]\centering
\caption{Detailed synchronisation correlations outside energy discharge}
\label{tab:sync-market-details-appendix}
\scriptsize
\setlength{\tabcolsep}{4pt}
\begin{tabular}{llcccccc}
\toprule
Market family & Outcome & SPSR pre & SPSR post & SPDR pre & SPDR post & $\Delta$ gap & DiD estimate \\
\midrule
Energy charge & Bid-band index & 0.332 & 0.378 & 0.174 & 0.135 & 0.085 & +0.085$^{***}$ (0.016) \\
 & Upper-band price & 0.352 & 0.415 & 0.091 & 0.057 & 0.097 & +0.087 (0.070) \\
 & Lower-band price & 0.180 & 0.287 & 0.088 & 0.103 & 0.091 & +0.091$^{***}$ (0.025) \\
 & Rebid timing & 0.216 & 0.354 & 0.196 & 0.252 & 0.081 & +0.081$^{***}$ (0.017) \\
\addlinespace[2pt]
Raise FCAS & Bid-band index & 0.153 & 0.270 & 0.116 & 0.081 & 0.151 & +0.151$^{***}$ (0.019) \\
 & Upper-band price & 0.156 & 0.230 & 0.074 & 0.066 & 0.081 & +0.081$^{***}$ (0.021) \\
 & Lower-band price & 0.299 & 0.298 & 0.231 & 0.123 & 0.107 & +0.107$^{***}$ (0.022) \\
 & Rebid timing & 0.325 & 0.411 & 0.318 & 0.317 & 0.087 & +0.087$^{***}$ (0.027) \\
\addlinespace[2pt]
Lower FCAS & Bid-band index & 0.265 & 0.397 & 0.212 & 0.120 & 0.224 & +0.224$^{***}$ (0.017) \\
 & Upper-band price & 0.061 & 0.268 & 0.120 & 0.032 & 0.295 & +0.295$^{***}$ (0.021) \\
 & Lower-band price & 0.328 & 0.352 & 0.276 & 0.208 & 0.092 & +0.092$^{***}$ (0.020) \\
 & Rebid timing & 0.337 & 0.410 & 0.315 & 0.316 & 0.072 & +0.072$^{***}$ (0.027) \\
\bottomrule
\end{tabular}
\begin{minipage}{0.94\linewidth}
\vspace{0.15cm}
\footnotesize \textit{Notes:} Entries are raw Pearson correlations, not Fisher-transformed correlations. Pre/post columns are simple averages of daily pair-type means. $\Delta$ gap is the post-minus-pre change in the SPSR--SPDR correlation gap. The DiD estimate is from group-by-day regressions with pair-type and trading-date fixed effects. Standard errors, in parentheses, are clustered by trading date. $^{***}p<0.01$, $^{**}p<0.05$, $^{*}p<0.10$.
\end{minipage}
\end{table}

\begin{table}[!htbp]\centering
\caption{Provider concentration, synchronisation, and variance across market families}
\label{tab:concentration-sync-variance-markets-appendix}
\small
\setlength{\tabcolsep}{3.6pt}
\begin{tabular}{lccccr}
\toprule
 & \multicolumn{2}{c}{Synchronisation} & \multicolumn{2}{c}{Market-action variance} & \\
\cmidrule(lr){2-3}\cmidrule(lr){4-5}
Market family & Local sync & SPSR sync & Local sync & Provider HHI & $N$ \\
\midrule
Energy discharge & +0.161$^{***}$ (0.021) & +0.169$^{***}$ (0.018) & +0.183 (0.882) & +1.084$^{***}$ (0.378) & 738 \\
Energy charge & -0.017 (0.057) & -0.021 (0.055) & -0.581 (0.784) & -0.947$^{**}$ (0.443) & 258 \\
Raise FCAS & -0.026$^{*}$ (0.016) & +0.005 (0.015) & +4.049$^{***}$ (0.943) & +1.666$^{***}$ (0.320) & 951 \\
Lower FCAS & -0.065$^{***}$ (0.015) & -0.040$^{***}$ (0.015) & +1.016 (0.993) & +2.203$^{***}$ (0.384) & 957 \\
\bottomrule
\end{tabular}
\begin{minipage}{0.94\linewidth}
\vspace{0.15cm}
\footnotesize \textit{Notes:} The outcome is the upper-band price in each bid-stack family. The first two columns report coefficients from regressions of local synchronisation and SPSR synchronisation on installed provider HHI. The next two columns report coefficients from regressions of within-day market-action variance on local sync and installed provider HHI. Provider HHI is scaled so one unit equals a ten percentage point HHI increase. Market-action variance is reported in millions of squared AUD/MWh. All specifications include trading-date fixed effects, omit region fixed effects because installed HHI is fixed within region, and cluster standard errors by trading date. Columns using synchronisation outcomes are weighted by SPSR pair-days; variance columns are weighted by market intervals. $^{***}p<0.01$, $^{**}p<0.05$, $^{*}p<0.10$.
\end{minipage}
\end{table}

\subsection{State-Responsive Bidding Checks}\label{app:state-responsive-checks}

This appendix gives the specification behind the common-state result in
Section~\ref{subsec:aem-synchronization} and reports the companion own-state
check.
Let \(z_{rt}\) denote public regional battery scarcity, measured as one minus
the regional SoC share and standardised within region-by-trading period using
the pre-ERI window. We split this signal into a tight-state component
\(T_{rt}=\max\{z_{rt},0\}\) and a loose-state component
\(L_{rt}=\max\{-z_{rt},0\}\). For each evening window \(h\), we estimate
\begin{equation}
Y_{irt}
=
\alpha_i+\gamma_{rd}+\mu_{rp}
+\sum_{m\neq -1}
\left(\delta^T_{mh}T_{rt}+\delta^L_{mh}L_{rt}\right)
\ind\{m(t)=m\}
+G_h(\mathrm{SoC}_{it},m(t))
+\varepsilon_{irt},
\label{eq:common-state-event-study}
\end{equation}
where \(Y_{irt}\) is a quantity-weighted bid-price outcome, \(\alpha_i\) are
battery fixed effects, \(\gamma_{rd}\) are region-by-date fixed effects,
\(\mu_{rp}\) are region-by-period fixed effects, and
\(G_h(\mathrm{SoC}_{it},m(t))\) collects own-SoC controls and their event-study
interactions. The reported object is the post-ERI average of
\(\delta^T_{mh}-\delta^L_{mh}\): the change in the loading of bids on scarce
states relative to abundant states.

The full common-state table by window and band group, including the results
summarised in Section~\ref{subsec:aem-synchronization}, is
Table~\ref{tab:common_state_window} in the main text.

The own-SoC specification replaces the public scarcity loading in
\eqref{eq:common-state-event-study} with phase-specific slopes on the
battery's own reconstructed state of charge, again interacted with event
months. This exercise validates that the same bid-stack outcomes respond to the
dynamic state that determines the opportunity cost of stored energy.

\begin{table}[!htbp]
\centering
\caption{Post-ERI own-SoC bid conditioning across markets}
\label{tab:soc_multimarket}
\scriptsize
\setlength{\tabcolsep}{4pt}
\begin{tabular}{llccc}
\toprule
Market family & Phase & All bands & Upper bands & Lower bands \\
\midrule
Energy discharge & Lead-in & 3,288$^{***}$ (849) & 4,960$^{***}$ (829) & -953$^{***}$ (126) \\
 & Peak-early & 3,947$^{***}$ (1,022) & 6,003$^{***}$ (1,051) & -264$^{***}$ (82) \\
 & Peak-late & -2,284$^{*}$ (1,182) & -1,247 (1,481) & 231 (196) \\
 & Post-peak & -2,342$^{**}$ (1,075) & -718 (1,387) & 417 (317) \\
\addlinespace
Energy charge & Lead-in & 992$^{***}$ (272) & 1,235$^{**}$ (496) & -125$^{**}$ (57) \\
 & Peak-early & -193 (221) & -867 (860) & -164$^{***}$ (54) \\
 & Peak-late & -68 (366) & -615 (1,063) & 353$^{***}$ (96) \\
 & Post-peak & -2 (339) & 854 (904) & 344$^{***}$ (79) \\
\addlinespace
Raise FCAS & Lead-in & 6,657$^{***}$ (746) & 6,692$^{***}$ (817) & -4$^{***}$ (1) \\
 & Peak-early & 5,964$^{***}$ (832) & 7,011$^{***}$ (906) & -3$^{***}$ (1) \\
 & Peak-late & -1,015 (1,220) & 1,045 (1,293) & -3$^{***}$ (1) \\
 & Post-peak & -1,525 (1,237) & 353 (1,269) & -3$^{***}$ (1) \\
\addlinespace
Lower FCAS & Lead-in & 6,186$^{***}$ (776) & 5,571$^{***}$ (674) & -3$^{***}$ (1) \\
 & Peak-early & 5,236$^{***}$ (869) & 3,677$^{***}$ (732) & -1 (1) \\
 & Peak-late & -312 (1,092) & -164 (1,415) & -2$^{**}$ (1) \\
 & Post-peak & -1,770$^{*}$ (1,046) & -1,020 (1,141) & -2$^{***}$ (1) \\
\bottomrule
\end{tabular}
\begin{minipage}{0.95\linewidth}
\vspace{0.2em}
\scriptsize \textit{Notes:} Entries are post-ERI average changes in the slope of bid prices on own SoC, in AUD/MWh for a 0-to-1 SoC change. Specifications include battery fixed effects and region-by-period-by-date fixed effects. Standard errors are clustered by trading date, matching Regression 3.
\end{minipage}
\end{table}

Table~\ref{tab:soc_multimarket} shows strong own-state responsiveness in the
same windows in which the common-state result is economically relevant. For
energy discharge, the post-ERI slope of upper bid prices on own SoC is
\(4{,}960\) AUD/MWh in the lead-in window and \(6{,}003\) AUD/MWh in the
early-peak window. Lower-band discharge prices move in the opposite direction
over the same intervals.

\subsection{Bellman Estimation and Deviation Valuation}\label{app:bellman-estimation}

This appendix describes the dynamic value-function estimates used to price
the opportunity cost of state-of-charge changes in the recleared deviation
analysis. The implementation builds on the
battery and FCAS optimisation implementation of
\citet{KarimiArpanahiPourmousaviMahdavi2024},
\citet{PrakashBruceMacGill2025}, \citet{Butters2021}, and
\citet{AEMO2021FCASNEMDE}.

\paragraph{State space}
Time within a trading day is indexed by $\tau\in\{1,\ldots,T\}$ with
$\Delta=5/60$~hours. For each region $r$ we estimate a mean price profile
$\bar p_{r,h}$ on the trailing window, where $h\equiv h(\tau)$ is the
hour-of-day bin. The realised regional price decomposes as
\begin{equation}
p_{rt}=\bar p_{r,h(t)}+\varepsilon_{rt},
\qquad
\varepsilon_{rt}\in\{e_{r1},\ldots,e_{rN_p}\},
\label{eq:bellman-price-decomp}
\end{equation}
with the residual $\varepsilon_{rt}$ discretised into $N_p$ price bins by
empirical quantiles of the same window; we write $k_t$ for the active
bin index. The exogenous transition kernel for the discretised residual
price is
\begin{equation}
k_{t+1}\sim
\widehat P_{r,b(t),z_t,\sigma_t}(k_t,\cdot),
\label{eq:bellman-kernel}
\end{equation}
where $b(t)\in\{\text{morning},\text{peak1},\text{peak2}\}$ is the
time-of-day block, $z_t$ is a three-level regional residual-demand state,
and $\sigma_t$ is a four-level aggregate battery-SoC signal
(low/mid/high SoC tertile, plus an unconditional fallback). Transitions
are estimated non-parametrically from the trailing window separately for
each $(r,b,z,\sigma)$ group, with a fallback chain that successively
pools $\sigma$ and then $z$ when group support is thin; the fallback
status of each visited group is recorded in the transition diagnostics
below.

The endogenous state is the own state of charge $S_{i,t}\in[0,\bar S_i]$.
We write dispatch in signed form $q_{it}$, with $q_{it}>0$ for discharge
and $q_{it}<0$ for charge, and decompose $q_{it}=q_{it}^{+}-q_{it}^{-}$
where $q_{it}^{+}:=\max\{q_{it},0\}$ and $q_{it}^{-}:=\max\{-q_{it},0\}$.
Using the same calibrated asset parameters as the SoC reconstruction in
Appendix~\ref{app:soc-model}, the stock-flow recursion is
\begin{equation}
S_{i,t+1}
=
\rho_i\,S_{it}-\kappa_{it}\,\Delta
-\frac{q_{it}^{+}\,\Delta}{\eta_i^{\mathrm{dis}}}
+q_{it}^{-}\,\eta_i^{\mathrm{ch}}\,\Delta,
\qquad
0\le S_{i,t+1}\le \bar S_i.
\label{eq:bellman-soc}
\end{equation}
The parasitic drain $\kappa_{it}$ is split into a constant aux-load
component $\alpha_i$ and an FCAS-attributable component $\beta_i F_{it}$
that scales with expected enabled FCAS MW at the operating point, using
the typical-utilisation decomposition documented in
\citet{KarimiArpanahiPourmousaviMahdavi2024}.

\paragraph{Flow payoff and FCAS headroom values}
Conditional on the operating point $q_{it}$ and the next-period state
$S_{i,t+1}$, the per-interval flow payoff under the
energy-and-FCAS market scope is
\begin{equation}
\pi_i(S_{it},S_{i,t+1},k_t,\tau)
=
q_{it}\,p_{rt}\,\Delta
-c_i\,|q_{it}|\,\Delta
+\sum_{m}a^{m}_{ibzk}\,H^{m}_i(q_{it},S_{i,t+1})\,\Delta,
\label{eq:bellman-flow}
\end{equation}
where $c_i$ is a throughput (wear) cost in AUD/MWh, $H^{m}_i(q,S')$ is
the physical raise- or lower-headroom that the battery can deliver for
FCAS service $m$ at $(q,S')$ subject to the trapezium and energy-duration
constraints of Appendix~\ref{app:fcas-trapezium}, and $a^{m}_{ibzk}$ is
the estimated FCAS revenue per MWh of physical headroom for service $m$
in the $(b,z,k)$ group. The $a^{m}_{ibzk}$ coefficients are estimated on
the trailing window from the NEMDE/FCAS clearing logic of
\citet{AEMO2021FCASNEMDE} and
\citet{KarimiArpanahiPourmousaviMahdavi2024}, using a hierarchical
fallback chain that backs off from the full
$\{\text{battery}\times\text{block}\times z\times k\}$ group to coarser
pools when within-group observations are thin. Optional shrinkage towards
the parent group and upper-tail winsorisation are available but are
disabled in the production configuration reported here.

\paragraph{Bellman recursion and shadow value}
For each battery $i$, signal state $\sigma$, residual-demand state $z$,
and trading day, the value function $V_{i\tau}^{\sigma z}$ is obtained
by backward induction on
\begin{equation}
V_{i\tau}^{\sigma z}(S,k)
=
\max_{S'\in \Gamma_i(S)}
\left\{
\pi_i(S,S',k,\tau)
+
\delta
\sum_{\ell}
\widehat P_{r,b(\tau),\sigma,z}(k,\ell)\,
V_{i,\tau+1}^{\sigma z}(S',\ell)
\right\},
\label{eq:bellman-recursion}
\end{equation}
with $\delta=1$ within the trading day (intraday discounting is
negligible at 5-minute resolution) and a terminal value
$V_{i,T+1}^{\sigma z}(\cdot,\cdot)$ taken from an overnight-value profile
estimated on the trailing window. The feasibility correspondence
$\Gamma_i(S)$ enforces the SoC bounds and the asset's discharge and
charge power limits.

The shadow value of stored energy is the derivative of $V$ with respect
to own SoC,
\begin{equation}
\nu_{it}
:=
\frac{\partial V_{i\tau}^{\sigma_t z_t}(S,k_t)}{\partial S}
\bigg|_{S=S_{it}},
\qquad
\widehat m_{it}
:=
\frac{\max\{\nu_{it},0\}}{\eta_i^{\mathrm{dis}}}+c_i.
\label{eq:bellman-shadow-value}
\end{equation}
The non-negativity floor reflects the fact that storage value is never
negative for a battery with the option to wait; the discharge-efficiency
adjustment converts the SoC shadow value into a per-MWh marginal cost
of discharge at the meter; and adding $c_i$ folds in the throughput
cost. The quantity $\widehat m_{it}$ is the per-MWh dynamic
opportunity cost used in the deviation reclearing of
Section~\ref{subsec:bellman-value} and in the conduct test of
Section~\ref{subsec:conduct-test}.

\paragraph{Rolling estimation}
The Bellman is re-estimated each week on a 90-day trailing window. For
anchor date $a$, the estimation window is $[a-90,\,a-1]$ and the target
week of evaluated intervals is $[a,\,a+6]$. Within an anchor, we
sequentially construct the regional mean price profile, residual-price
bin edges, transition kernel with group fallbacks, FCAS headroom-value
coefficients, and an overnight terminal value; we then solve the
backward induction \eqref{eq:bellman-recursion} and evaluate the value
function only at the target-week intervals required by the deviation
panel. The production configuration uses $N_p=11$ residual-price bins,
$N_S=41$ own-SoC grid points, and $201$ candidate next-period SoC
controls per interval; sensitivity tests at $N_S=21$ and $501$ controls
shift $\widehat m_{it}$ by less than $11$ AUD/MWh, well below its
typical magnitude. The estimated value functions currently cover
\mbox{2025-04-01 13:05} through \mbox{2026-02-02 21:55}, from $44$
successful weekly anchors and three failed anchors at
\mbox{2026-02-03}, \mbox{2026-02-10}, and \mbox{2026-02-17}.
Table~\ref{tab:bellman-coverage} reports the coverage by month and
asset.

\begin{table}[h]
\centering
\caption{Bellman rolling-pipeline coverage}
\label{tab:bellman-coverage}
\begin{tabular}{@{} l r @{}}
\toprule
Item & Value \\
\midrule
Weekly anchors attempted & 47 \\
Weekly anchors successful & 44 \\
Weekly anchors failed & 3 \\
Date range of evaluated intervals & 2025-04-01 13:05 -- 2026-02-02 21:55 \\
Regions covered & 4 \\
Batteries (DUIDs) covered & 51 \\
Battery-intervals evaluated & 1,379,979 \\
Battery-intervals pre-ERI / post-ERI & 335,552 / 1,044,427 \\
Share of intervals with non-missing $\widehat m_{it}$ & 100.0\% \\
Share of intervals with $\lambda_{it}=0$ (option-value floor) & 13.8\% \\
\bottomrule
\end{tabular}

\vspace{4pt}
{\footnotesize \textit{Notes:} Coverage of the rolling weekly Bellman estimation. Failed anchors: 2026-02-03, 2026-02-10, 2026-02-17. $\lambda_{it}=0$ indicates intervals at which the implied shadow value is non-positive and the option-value floor binds; $\widehat m_{it}$ is the corresponding throughput cost $c_i$.}
\end{table}

Table~\ref{tab:bellman-soc-calibration} reports the asset-level
calibration parameters used by the DP, alongside SoC RMSE and bias on
the calibration window. These parameters are inherited from the SoC
reconstruction in Appendix~\ref{app:soc-model}; the table is included
here so that the Bellman appendix is self-contained.

\footnotesize
\begin{longtable}{@{} l l r r r r r r r r @{}}
\caption{Bellman asset and SoC calibration parameters by battery}
\label{tab:bellman-soc-calibration} \\
\toprule
DUID & Region & $\bar S_i$ & $\bar P_i^{\mathrm{dis}}$ & $\bar P_i^{\mathrm{ch}}$ & $\eta^c_i$ & $\eta^d_i$ & $\kappa_i$ & RMSE & Bias \\
 & & (MWh) & (MW) & (MW) & & & (MW) & (\% $\bar S_i$) & (\% $\bar S_i$) \\
\midrule
\endfirsthead
\toprule
DUID & Region & $\bar S_i$ & $\bar P_i^{\mathrm{dis}}$ & $\bar P_i^{\mathrm{ch}}$ & $\eta^c_i$ & $\eta^d_i$ & $\kappa_i$ & RMSE & Bias \\
 & & (MWh) & (MW) & (MW) & & & (MW) & (\% $\bar S_i$) & (\% $\bar S_i$) \\
\midrule
\endhead
\midrule
\multicolumn{10}{r@{}}{\footnotesize\itshape Continued on next page} \\
\endfoot
\bottomrule
\endlastfoot
\texttt{BHB1} & NSW1 & 50.0 & 49.7 & 45.2 & 0.999 & 0.999 & 0.257 & 0.1 & -0.0 \\
\texttt{CAPBES1} & NSW1 & 199.9 & 100.0 & 91.6 & 0.999 & 0.999 & 0.318 & 0.1 & -0.0 \\
\texttt{DPNTB1} & NSW1 & 46.5 & 25.0 & 25.0 & 0.923 & 0.923 & 0.000 & 0.1 & -0.0 \\
\texttt{LDBESS1} & NSW1 & 1086.2 & 75.1 & 74.7 & 0.952 & 0.952 & 5.000 & 0.0 & -0.0 \\
\texttt{LIMBESS1} & NSW1 & 534.6 & 50.0 & 50.8 & 0.999 & 0.999 & 1.904 & 0.1 & -0.0 \\
\texttt{ORABESS1} & NSW1 & 1660.0 & 43.2 & 48.0 & 0.889 & 0.889 & 0.313 & 0.0 & 0.0 \\
\texttt{QBYNB1} & NSW1 & 15.2 & 10.0 & 9.3 & 0.994 & 0.994 & 0.047 & 0.1 & 0.0 \\
\texttt{QPSFB2} & NSW1 & 40.0 & 3.0 & 19.0 & 0.900 & 0.900 & 0.000 & 0.0 & -0.0 \\
\texttt{RESS1} & NSW1 & 117.1 & 60.0 & 60.0 & 0.917 & 0.917 & 0.165 & 0.0 & 0.0 \\
\texttt{RIVNB2} & NSW1 & 126.4 & 65.0 & 65.0 & 0.922 & 0.922 & 0.185 & 0.1 & -0.0 \\
\texttt{SMTHBES1} & NSW1 & 130.0 & 65.0 & 65.0 & 0.999 & 0.999 & 0.589 & 0.1 & -0.0 \\
\texttt{WALGRV1} & NSW1 & 78.0 & 49.2 & 46.2 & 0.961 & 0.961 & 0.095 & 0.1 & -0.0 \\
\texttt{WTAHB1} & NSW1 & 1154.0 & 638.0 & 299.7 & 0.927 & 0.927 & 1.732 & 0.1 & -0.0 \\
\texttt{BBATTERY1} & QLD1 & 95.6 & 50.0 & 50.0 & 0.985 & 0.985 & 0.257 & 0.0 & -0.0 \\
\texttt{BRNDBES1} & QLD1 & 410.0 & 205.0 & 205.0 & 0.980 & 0.980 & 0.888 & 0.1 & -0.0 \\
\texttt{CHBESS1} & QLD1 & 197.2 & 100.0 & 99.9 & 0.954 & 0.954 & 0.452 & 0.1 & -0.0 \\
\texttt{GREENB1} & QLD1 & 400.0 & 199.5 & 199.3 & 0.979 & 0.979 & 1.816 & 0.0 & -0.0 \\
\texttt{KEPBG1} & QLD1 & 4.0 & 2.0 & 0.0 & 0.949 & 0.949 & 0.000 &  &  \\
\texttt{SNB01} & QLD1 & 544.3 & 260.0 & 248.3 & 0.901 & 0.901 & 0.373 & 0.1 & 0.0 \\
\texttt{SWANBBF1} & QLD1 & 410.5 & 248.8 & 236.1 & 0.896 & 0.896 & 0.536 & 0.1 & -0.0 \\
\texttt{TARBESS1} & QLD1 & 600.0 & 300.0 & 299.2 & 0.999 & 0.999 & 2.970 & 0.1 & -0.0 \\
\texttt{ULPBESS1} & QLD1 & 298.0 & 155.0 & 153.2 & 0.999 & 0.999 & 3.976 & 0.1 & -0.0 \\
\texttt{WANDB1} & QLD1 & 150.0 & 90.1 & 74.9 & 0.977 & 0.977 & 0.976 & 0.1 & -0.0 \\
\texttt{WDBESS1} & QLD1 & 400.0 & 254.9 & 253.6 & 0.999 & 0.999 & 0.968 & 0.2 & -0.1 \\
\texttt{WDBESS2} & QLD1 & 510.0 & 255.0 & 255.0 & 0.999 & 0.999 & 2.289 & 0.1 & -0.0 \\
\texttt{ADPBA1} & SA1 & 6.0 & 6.1 & 6.0 & 0.999 & 0.999 & 0.002 & 0.5 & -0.4 \\
\texttt{BLYTHB1} & SA1 & 400.0 & 199.9 & 199.1 & 0.999 & 0.999 & 2.293 & 0.1 & -0.0 \\
\texttt{BOWWBA1} & SA1 & 2.0 & 2.0 & 2.0 & 0.999 & 0.999 & 0.000 & 0.7 & -0.6 \\
\texttt{BUNGAMB1} & SA1 & 300.0 & 30.0 & 25.4 & 0.999 & 0.999 & 4.354 & 0.0 & -0.0 \\
\texttt{CBWWBA1} & SA1 & 2.0 & 2.0 & 2.0 & 0.999 & 0.999 & 0.000 & 0.6 & -0.5 \\
\texttt{CGBESS01} & SA1 & 120.0 & 20.1 & 20.2 & 0.999 & 0.999 & 2.297 & 0.0 & -0.0 \\
\texttt{DALNTH1} & SA1 & 11.7 & 10.0 & 6.0 & 0.927 & 0.927 & 0.149 & 0.1 & 0.0 \\
\texttt{HPR1} & SA1 & 169.5 & 80.0 & 80.6 & 0.980 & 0.980 & 0.865 & 0.0 & -0.0 \\
\texttt{HVWWBA1} & SA1 & 4.0 & 4.0 & 4.0 & 0.999 & 0.999 & 0.000 & 0.6 & -0.5 \\
\texttt{LBB1} & SA1 & 43.4 & 24.1 & 24.3 & 0.951 & 0.951 & 0.000 & 0.0 & -0.0 \\
\texttt{LGAPBS1} & SA1 & 10.0 & 9.0 & 13.0 & 0.898 & 0.898 & 0.086 & 0.0 & 0.0 \\
\texttt{MANNUMB1} & SA1 & 200.0 & 99.8 & 99.8 & 0.958 & 0.958 & 0.343 & 0.1 & -0.0 \\
\texttt{TB2B1} & SA1 & 41.0 & 40.5 & 38.6 & 0.999 & 0.999 & 0.339 & 0.1 & -0.0 \\
\texttt{TEMPB1} & SA1 & 285.0 & 111.0 & 111.0 & 0.999 & 0.999 & 2.019 & 0.1 & -0.0 \\
\texttt{TIB1} & SA1 & 250.0 & 152.7 & 184.0 & 0.993 & 0.993 & 2.338 & 0.1 & -0.0 \\
\texttt{BALB1} & VIC1 & 28.4 & 30.0 & 30.0 & 0.959 & 0.959 & 0.071 & 0.1 & -0.0 \\
\texttt{BULBES1} & VIC1 & 31.9 & 19.9 & 20.0 & 0.952 & 0.952 & 0.022 & 0.1 & -0.0 \\
\texttt{GANNB1} & VIC1 & 37.2 & 25.3 & 25.3 & 0.928 & 0.928 & 0.000 & 0.1 & -0.0 \\
\texttt{HBESS1} & VIC1 & 150.0 & 149.3 & 141.0 & 0.999 & 0.999 & 1.175 & 0.1 & -0.0 \\
\texttt{KESSB1} & VIC1 & 370.0 & 185.0 & 185.0 & 0.758 & 0.758 & 0.000 & 0.4 & -0.1 \\
\texttt{LVES1} & VIC1 & 200.0 & 100.0 & 100.0 & 0.999 & 0.999 & 1.084 & 0.1 & -0.0 \\
\texttt{MREHA1} & VIC1 & 400.0 & 200.0 & 200.0 & 0.974 & 0.974 & 0.152 & 0.1 & -0.0 \\
\texttt{MREHA2} & VIC1 & 400.0 & 199.7 & 199.8 & 0.964 & 0.964 & 0.457 & 0.1 & -0.0 \\
\texttt{MREHA3} & VIC1 & 800.0 & 200.0 & 200.0 & 0.992 & 0.992 & 1.074 & 0.0 & -0.0 \\
\texttt{PIBESS1} & VIC1 & 11.5 & 5.0 & 5.0 & 0.952 & 0.952 & 0.009 & 0.1 & -0.0 \\
\texttt{RANGEB1} & VIC1 & 400.0 & 200.0 & 200.0 & 0.943 & 0.943 & 0.000 & 0.1 & -0.0 \\
\texttt{VBB1} & VIC1 & 450.0 & 221.0 & 233.3 & 0.933 & 0.933 & 0.000 & 0.2 & 0.1 \\
\end{longtable}
\normalsize
{\footnotesize \textit{Notes:} Per-battery parameters used in the Bellman recursion. $\bar S_i$ is the energy capacity used in the DP, $\bar P_i^{\mathrm{dis}}$ and $\bar P_i^{\mathrm{ch}}$ the discharge and charge power limits, $\eta^c_i$ and $\eta^d_i$ the charge and discharge efficiencies, and $\kappa_i$ the calibrated parasitic drain in MW (which the DP decomposes into a constant aux-load component and an FCAS-attributable component scaling with expected enabled FCAS MW). RMSE and bias are evaluated on the un-anchored post-ERI calibration window and expressed as a percentage of $\bar S_i$.}

Table~\ref{tab:bellman-value-profile} reports the distribution of
$\widehat m_{it}$ across battery-intervals, summarised by time block,
own-SoC bin, and aggregate market-SoC bin. Consistent with the dynamic
arbitrage logic in Section~\ref{sec:battery-bidding}, the median rises
in the evening-peak block and falls when both own and market SoC are
high; the lowest values cluster in intervals with high own SoC and
high market SoC, where storage value is small and the
option-value floor frequently binds.

\begin{table}[h]
\centering
\caption{$\widehat m_{it}$ summary by block and SoC bin}
\label{tab:bellman-value-profile}
\small
\begin{tabular}{@{} l l r r r @{}}
\toprule
 & & \multicolumn{3}{c}{Market SoC bin} \\
\cmidrule(lr){3-5}
Time block & Own SoC & Low & Mid & High \\
\midrule
Peak 1 (09--15) & Low & 95 [48, 148] & 74 [33, 113] & 58 [19, 86] \\
Peak 1 (09--15) & Mid & 58 [19, 94] & 35 [1, 79] & 9 [1, 46] \\
Peak 1 (09--15) & High & 5 [1, 37] & 1 [1, 38] & 1 [1, 18] \\
Peak 2 (15--22) & Low & 151 [117, 212] & 163 [115, 269] & 97 [74, 140] \\
Peak 2 (15--22) & Mid & 131 [101, 174] & 132 [89, 193] & 74 [37, 120] \\
Peak 2 (15--22) & High & 95 [53, 135] & 93 [37, 149] & 38 [1, 88] \\
\bottomrule
\end{tabular}
\normalsize

\vspace{4pt}
{\footnotesize \textit{Notes:} Median $\widehat m_{it}$ in AUD/MWh with 25th and 75th percentiles in brackets, computed across all battery-intervals in the rolling evaluation panel. Own SoC is the battery's lagged stored-energy share; market SoC is the regional fleet's stored-energy share at the same interval (sum of asset-level SoC over the region, divided by the regional capacity). Bins are low/mid/high tertiles. Time block follows the Bellman conditioning convention (morning $=$ 06:00--09:00, peak 1 $=$ 09:00--15:00, peak 2 $=$ 15:00--22:00; the rolling evaluation panel only retains peak 1 and peak 2 intervals).}
\end{table}

\paragraph{Transition diagnostics}
Table~\ref{tab:bellman-transition-diagnostics} reports kernel quality
pooled across anchors. Row-sum errors on the unconditional kernel are
machine-zero; the table also reports the minimum group counts in the
unconditional, signal-conditioned, residual-demand-conditioned, and
fully-conditioned versions of the kernel, which together govern the
fallback chain used in solving \eqref{eq:bellman-recursion}.

\begin{table}[h]
\centering
\caption{Transition kernel diagnostics}
\label{tab:bellman-transition-diagnostics}
\small
\begin{tabular}{@{} l r r r r r @{}}
\toprule
 & & & \multicolumn{3}{c}{Share with row support $n_{r,b,\sigma,z,k}\ge$} \\
\cmidrule(lr){4-6}
Conditioning of $\widehat P_{r,b,\sigma,z}(k,\cdot)$ & Intervals & Share & 30 & 100 & 1{,}000 \\
\midrule
Signal $+$ residual demand ($\sigma+z$) & 485,138 & 35.2\% & 93.8\% & 54.6\% & 0.2\% \\
Residual demand only ($z$, $\sigma$ pooled) & 894,841 & 64.8\% & 98.2\% & 92.6\% & 8.6\% \\
\bottomrule
\end{tabular}
\normalsize

\vspace{4pt}
{\footnotesize \textit{Notes:} Transition-kernel diagnostics evaluated at the $(r, b, \sigma, z, k)$ states actually visited by the rolling evaluation panel (total 1,379,979 battery-intervals). The conditioning level used at solve time is $\sigma+z$ when the aggregate-SoC signal is named (low, mid, high) and $z$-only otherwise. ``Row support'' is the number of within-window 5-min observations the kernel row was estimated from. The maximum row-sum error across visited rows is 2.2e-16 (95th percentile 1.1e-16); the share of evaluated intervals with non-missing $\widehat m_{it}$ is 100.0\%.}
\end{table}

\paragraph{FCAS diagnostics}
Table~\ref{tab:bellman-fcas-diagnostics} reports the distribution of
the estimated FCAS headroom-value coefficients $a^{m}_{ibzk}$ across
services, together with the share of groups whose coefficients are set
by a pooled-fallback estimator rather than the fully conditioned group.
These objects are first-order for the Bellman because they govern how
much shadow value $\widehat m_{it}$ is attributable to energy arbitrage
versus FCAS option value.

\begin{table}[h]
\centering
\caption{FCAS headroom-value diagnostics}
\label{tab:bellman-fcas-diagnostics}
\small
\begin{tabular}{@{} l r r r r r @{}}
\toprule
 & & \multicolumn{3}{c}{Headroom value (AUD/MWh of headroom)} & Pooled \\
\cmidrule(lr){3-5}
Service & Cells & Median & p75 & p95 & fallback \\
\midrule
Raise (all) & 180,675 & 0.20 & 1.02 & 4.35 & 11.2\% \\
Raise reg & 180,675 & 0.09 & 0.58 & 2.35 & 11.2\% \\
Raise cont & 180,675 & 0.01 & 0.21 & 2.34 & 11.2\% \\
Lower (all) & 180,675 & 0.11 & 0.63 & 4.47 & 11.2\% \\
Lower reg & 180,675 & 0.01 & 0.21 & 1.11 & 11.2\% \\
Lower cont & 180,675 & 0.03 & 0.23 & 3.82 & 11.2\% \\
\bottomrule
\end{tabular}
\normalsize

\vspace{4pt}
{\footnotesize \textit{Notes:} Distribution of estimated FCAS headroom-value coefficients $a^{m}_{ibzk}$ across $(d, b, z, k)$ cells, pooled across the 44 successful weekly anchors. The coefficient enters the flow payoff as $a^{m}_{ibzk}\,H^{m}_i\,\Delta$ and is measured in AUD per MWh of physical headroom (MW of headroom $\times$ interval length in hours). ``Cells'' is the count of (battery, block, residual-demand state, residual-price bin) cells with a finite estimate. ``Pooled fallback'' is the share of cells whose value is set by a coarser hierarchical-fallback estimator (region-by-block-by-bin, or region-by-block) rather than the fully conditioned cell, before any optional shrinkage or upper-tail winsorization. ``Raise (all)'' and ``Lower (all)'' aggregate over regulation and contingency services within each direction.}
\end{table}

\paragraph{Deviation menu}
Table~\ref{tab:deviation-menu} reports the full set of operators used in
the broad deviation panel of Section~\ref{subsec:bellman-value}. Each
operator is a clearing-relative reallocation that takes the bid
band-group structure described in the main text as given. Shift sizes
are fixed shares of the relevant offered MW and are capped by available
movable MW in the source band.

\begin{table}[h]
\centering
\caption{Deviation menu for the incentive-compatibility diagnostic}
\label{tab:deviation-menu}
\small
\begin{tabular}{@{} l p{0.47\linewidth} r r @{}}
\toprule
Family & Deviation & Selected obs. & Evaluable obs. \\
\midrule
Within energy & Discharge energy $8$--$10\to4$--$7$ (10\%) & 800 & 800 \\
Within energy & Discharge energy $8$--$10\to1$--$3$ (3.5\%) & 800 & 800 \\
Within energy & Charge energy $4$--$7\to1$--$3$ (10\%) & 800 & 800 \\
Within FCAS & Raise FCAS $8$--$10\to1$--$3$ (2\%) & 800 & 647 \\
Within FCAS & Lower FCAS $8$--$10\to1$--$3$ (2.5\%) & 800 & 638 \\
Cross-market & Raise FCAS $8$--$10\to$ discharge energy $4$--$7$ (10\%) & 800 & 422 \\
Cross-market & Raise FCAS $8$--$10\to$ discharge energy $1$--$3$ (3.5\%) & 800 & 418 \\
Cross-market & Lower FCAS $8$--$10\to$ discharge energy $1$--$3$ (3.5\%) & 800 & 405 \\
Cross-market & Discharge energy $8$--$10\to$ raise FCAS $1$--$3$ (3.5\%) & 800 & 647 \\
Cross-market & Discharge energy $8$--$10\to$ lower FCAS $1$--$3$ (3.5\%) & 800 & 632 \\
Cross-market & Raise FCAS $8$--$10\to$ charge energy $1$--$3$ (10\%) & 800 & 324 \\
Cross-market & Lower FCAS $8$--$10\to$ charge energy $1$--$3$ (10\%) & 800 & 344 \\
Cross-FCAS & Raise FCAS $8$--$10\to$ lower FCAS $1$--$3$ (2\%) & 800 & 681 \\
Cross-FCAS & Lower FCAS $8$--$10\to$ raise FCAS $1$--$3$ (2.5\%) & 800 & 611 \\
\midrule
Total & & 11,200 & 8,169 \\
\bottomrule
\end{tabular}
\normalsize

\vspace{4pt}
{\footnotesize \textit{Notes:} One observation is a focal battery by five-minute interval by deviation type. Selected observations are the fixed-budget counterfactual-reclearing target bank. Evaluable observations require a successful exact bid-stack edit and reclear, and a valid Bellman match for the focal battery. Shift percentages are applied to the relevant offered MW and capped by available MW in the source band.}
\end{table}

\paragraph{Continuation values for deviation reclearing}
For a sampled stack deviation $d$ at battery-interval $(i,t)$, exact
reclearing produces a counterfactual energy operating point and FCAS
enablement, hence a counterfactual next-period state of charge
$S^{cf}_{i,t+1}$. We define the deviation in continuation value as
\begin{equation}
\Delta V_{itd}
:=
\E_{\ell\mid k_t}\!\left[V_{i,t+1}(S^{cf}_{i,t+1},\ell)\right]
-
\E_{\ell\mid k_t}\!\left[V_{i,t+1}(S^{base}_{i,t+1},\ell)\right],
\label{eq:bellman-deltav}
\end{equation}
with the expectation taken under the same kernel
$\widehat P_{r,b(t),\sigma_t,z_t}(k_t,\cdot)$ used in solving
\eqref{eq:bellman-recursion}. The deviation's effect on owner-level
total payoff is the recleared change in current energy and FCAS profit
plus $\Delta V_{itd}$, and this is the object summarised by the
incentive-compatibility diagnostic in
Section~\ref{subsec:bellman-value}.

As a robustness alternative, $\Delta V_{itd}$ can be evaluated at the
contemporaneous residual-price bin $\ell=k_t$ rather than averaged under
the kernel. This is the natural counterpart to the information set
available at the deviation interval, but it departs from the kernel used
to solve \eqref{eq:bellman-recursion}.
Table~\ref{tab:bellman-deviation-diagnostics} reports both conventions
side by side together with the contribution of $\Delta V_{itd}$ to the
total recleared deviation gain.

\begin{table}[h]
\centering
\caption{Deviation-valuation diagnostics}
\label{tab:bellman-deviation-diagnostics}
\small
\begin{tabular}{@{} l r r r r r r @{}}
\toprule
 & & $|\Delta V|/$ & \multicolumn{2}{c}{$|\Delta V_{\text{exp}}-\Delta V_{\text{bin}}|$} & \multicolumn{2}{c}{IC-sign flip share} \\
\cmidrule(lr){4-5}\cmidrule(lr){6-7}
Deviation family & Deviations & $|G^O|$ (med.) & Median & p90 & $\tau=1$ & $\tau=100$ \\
 & & & (AUD) & (AUD) & & \\
\midrule
Within energy & 2,400 & 70.7\% & 0.00 & 5.41 & 1.42\% & 0.08\% \\
Within FCAS & 1,285 & 0.0\% & 0.00 & 0.00 & 0.00\% & 0.08\% \\
Cross-market & 3,192 & 0.0\% & 0.00 & 1.40 & 1.32\% & 0.13\% \\
Cross-FCAS & 1,292 & 0.0\% & 0.00 & 0.00 & 0.00\% & 0.00\% \\
\bottomrule
\end{tabular}
\normalsize

\vspace{4pt}
{\footnotesize \textit{Notes:} Diagnostics for the continuation-value contribution to the recleared deviation gain $G^O_{itd}$. Sample is focal-asset rows from the four broad deviation panels with valid Bellman conditioning. Column three is the median ratio $|\Delta V_{itd}|/|G^O_{itd}|$, capped at one. Columns four and five compare $\Delta V_{itd}$ evaluated under the kernel-expectation convention used in the main results (`exp') against the contemporaneous-bin convention (`bin'). The last two columns report the share of deviations whose IC classification changes across conventions at tolerance $\tau$ AUD.}
\end{table}

The two conventions are economically indistinguishable on the realised
deviation panel: the median absolute discrepancy in $\Delta V_{itd}$ is
zero AUD across all four panels and the $90$th-percentile discrepancy
does not exceed $\$6$. The share of deviations whose IC classification
flips between conventions is below $1.5\%$ at the $1$~AUD tolerance and
near zero at the $100$~AUD tolerance. The IC pre/post pattern reported in
Table~\ref{tab:ic-violations-broad} is therefore not driven by the choice
of continuation convention.

\paragraph{Incentive-compatibility violations by provider concentration}
Table~\ref{tab:ic-violations-concentration} splits the pre/post-ERI IC
diagnostic of Table~\ref{tab:ic-violations-broad} by the focal
battery's provider-group parked share in the same interval (Tesla,
Fluence, Major-IOU, and Other), bucketed into $20$-percentage-point
bins. The summary statistic in each group is the share of focal
battery-intervals on which the maximum owner-level deviation gain
$M^O_{it}$ across the sampled broad deviation menu exceeds a dollar
threshold $\tau\in\{1,100,1000\}$. The pre/post drop in IC violation
shares holds inside every concentration bucket and at every
threshold, so the post-ERI improvement in incentive compatibility is
not driven by the changing concentration mix of the deviation panel
over time.

\begin{table}[h]
\centering
\caption{Incentive-compatibility violations by focal-provider parked
share, pre/post-ERI}
\label{tab:ic-violations-concentration}
\small
\begin{tabular}{@{} l l r r r r r @{}}
\toprule
Provider-group parked share & Period & Intervals & Mean share & \multicolumn{3}{c}{Share with $M^O_{it}>\tau$ at $\tau =$} \\
\cmidrule(lr){5-7}
 & & & & 1 AUD & 100 AUD & 1{,}000 AUD \\
\midrule
0--20\% & Pre-ERI & 631 & 9.1\% & 14.9\% & 2.2\% & 1.6\% \\
0--20\% & Post-ERI & 1,218 & 9.1\% & 8.0\% & 0.2\% & 0.0\% \\
20--40\% & Pre-ERI & 628 & 30.4\% & 22.6\% & 5.1\% & 1.8\% \\
20--40\% & Post-ERI & 1,373 & 28.7\% & 5.5\% & 0.7\% & 0.4\% \\
40--60\% & Pre-ERI & 603 & 48.9\% & 18.9\% & 3.8\% & 2.8\% \\
40--60\% & Post-ERI & 801 & 49.6\% & 9.1\% & 0.9\% & 0.1\% \\
60--80\% & Pre-ERI & 153 & 69.3\% & 15.7\% & 3.9\% & 3.9\% \\
60--80\% & Post-ERI & 375 & 69.5\% & 8.8\% & 1.1\% & 0.0\% \\
80--100\% & Pre-ERI & 167 & 86.3\% & 30.5\% & 28.1\% & 27.5\% \\
80--100\% & Post-ERI & 98 & 85.1\% & 4.1\% & 0.0\% & 0.0\% \\
\bottomrule
\end{tabular}
\normalsize

\vspace{4pt}
{\footnotesize \textit{Notes:} Provider-group parked share is the focal battery's provider-group share of regional parked discharge capacity in the sampled interval. Provider groups are Tesla, Fluence, Major-IOU, and Other. The table is restricted to rows for which this all-interval concentration measure is observed (8,169 of 8,169 evaluable deviation rows). One interval observation is a focal battery by five-minute interval. $M^O_{it}$ is the maximum owner-level gain across the sampled broad deviation menu for that battery-interval.}
\end{table}

\subsection{Software Provider Identification}\label{app:provider-id}

The empirical analysis requires a mapping from individual battery dispatch
units (DUIDs) to autobidder software providers. Because AEMO registration data
do not record which software platform operates a given battery, we assemble
this mapping from public sources. Table~\ref{tab:provider-mapping} reports the
complete result.

Each assignment carries one of two confidence levels. \emph{Confirmed}
($\bullet$) assignments rely on direct public evidence naming the software
platform for a specific battery: vendor investor-relations filings (e.g.,
Fluence 10-K and 10-Q filings listing contracted deployments), manufacturer
product pages identifying the deployed platform (e.g., W\"{a}rtsil\"{a} GEMS,
Energy Vault VaultOS), or asset-owner public statements naming the provider
(e.g., Iberdrola confirming Tesla Autobidder for its Australian portfolio).
\emph{Probable} ($\circ$) assignments rely on strong indirect evidence,
typically the combination of hardware identification (e.g., Tesla Megapack
units) with the manufacturer's known software-bundling pattern (e.g., Tesla's
standard Autobidder revenue-sharing arrangement), or a market participant's
known intermediary role for a portfolio of tolled assets (e.g.,
EnergyAustralia's intermediary function for several BESS assets). Batteries
with no public evidence linking them to any provider are marked as unknown and
excluded from all provider-level specifications.

Of the 60 registered battery DUIDs across the four NEM regions, 48 can be
linked to a software provider: 30 with confirmed evidence and 18 with probable
evidence. After consolidating the Kennedy gen-load pair into a single physical
asset, harmonising provider labels, and restricting to the set of usable
grid-scale batteries, the mapping yields 42 known-provider batteries served by
11 distinct providers across 21 provider-region combinations.

\begin{longtable}{@{} p{1.6cm} p{3.6cm} p{3.5cm} p{3.5cm} c @{}}
\caption{Battery-to-provider mapping}
\label{tab:provider-mapping} \\
\toprule
DUID & Station & Participant & Provider & C \\
\midrule
\endfirsthead
\toprule
DUID & Station & Participant & Provider & C \\
\midrule
\endhead
\midrule
\multicolumn{5}{r@{}}{\footnotesize\itshape Continued on next page} \\
\endfoot
\bottomrule
\endlastfoot
\midrule
\multicolumn{5}{@{}l}{\textbf{NSW1}} \\
\midrule
\texttt{BHB1} & Broken Hill Battery & AGL & AGL in-house & \textbullet \\
\texttt{LDBESS1} & Liddell BESS & AGL & AGL in-house & \textbullet \\
\texttt{DPNTB1} & Darlington Point ESS & EA/Edify & EA intermediary & $\circ$ \\
\texttt{RIVNB2} & Riverina ESS 2 & EA/Edify & EA intermediary & $\circ$ \\
\texttt{NESBESS1} & New England BESS & ACEN/NESF & Energy Vault VaultOS & \textbullet \\
\texttt{NESBESS2} & New England BESS & ACEN/NESF & Energy Vault VaultOS & \textbullet \\
\texttt{CAPBES1} & Capital Battery & GPG/Naturgy & Fluence Mosaic & \textbullet \\
\texttt{MLB01} & Mortlake BESS & Origin Energy & Fluence Mosaic & \textbullet \\
\texttt{ORABESS1} & Orana BESS & Akaysha/BlackRock & Fluence Mosaic & \textbullet \\
\texttt{QBYNB1} & Queanbeyan BESS & GPG/Naturgy & Fluence Mosaic & \textbullet \\
\texttt{WTAHB1} & Waratah Super Battery & Akaysha/BlackRock & Fluence Mosaic & \textbullet \\
\texttt{RESS1} & Riverina ESS 1 & Shell Energy & Shell Energy interm. & $\circ$ \\
\texttt{SMTHBES1} & Smithfield BESS & Iberdrola & Tesla Autobidder & $\circ$ \\
\texttt{WALGRV1} & Wallgrove BESS & Iberdrola & Tesla Autobidder & \textbullet \\
\texttt{ERB01} & Eraring BESS & Origin Energy & W\"{a}rtsil\"{a} GEMS & \textbullet \\
\texttt{LIMBESS1} & Limondale Battery & RWE & --- & --- \\
\texttt{QPSFB2} & Quorn Park Solar Hybrid & Potentia/Enel & --- & --- \\
\midrule
\multicolumn{5}{@{}l}{\textbf{QLD1}} \\
\midrule
\texttt{CHBESS1} & Chinchilla BESS & CS Energy & CS Energy in-house & $\circ$ \\
\texttt{GREENB1} & Greenbank BESS & CS Energy & CS Energy in-house & $\circ$ \\
\texttt{WANDB1} & Wandoan BESS & Vena Energy/AGL & Doosan GridTech & \textbullet \\
\texttt{BRNDBES1} & Brendale BESS & Akaysha & Fluence Mosaic & \textbullet \\
\texttt{ULPBESS1} & Ulinda Park BESS & Akaysha & Fluence Mosaic & \textbullet \\
\texttt{BBATTERY1} & Bouldercombe Battery & Genex Power & Tesla Autobidder & \textbullet \\
\texttt{SWANBBF1} & Swanbank BESS & CleanCo & Tesla Autobidder & $\circ$ \\
\texttt{TARBESS1} & Tarong BESS & Stanwell & Tesla Autobidder & $\circ$ \\
\texttt{WDBESS1} & Western Downs BESS & Neoen & Tesla Autobidder & \textbullet \\
\texttt{WDBESS2} & Western Downs BESS & Neoen & Tesla Autobidder & \textbullet \\
\texttt{KEPBG1} & Kennedy Energy Park & Windlab/Eurus & --- & --- \\
\texttt{KEPBL1} & Kennedy Energy Park & Windlab/Eurus & --- & --- \\
\texttt{SNB01} & Supernode BESS & Quinbrook & --- & --- \\
\texttt{SNB02} & Supernode BESS & Quinbrook & --- & --- \\
\midrule
\multicolumn{5}{@{}l}{\textbf{SA1}} \\
\midrule
\texttt{DALNTH1} & Dalrymple North BESS & AGL & AGL in-house & \textbullet \\
\texttt{TIB1} & Torrens Island BESS & AGL & AGL in-house & \textbullet \\
\texttt{TB2B1} & Tailem Bend 2 & Vena Energy & Doosan GridTech & \textbullet \\
\texttt{LGAPBS1} & Lincoln Gap WF Battery & Ratch/Nexif & Fluence Mosaic & $\circ$ \\
\texttt{MANNUMB1} & Mannum BESS & Epic Energy & Habitat Energy & \textbullet \\
\texttt{BLYTHB1} & Blyth BESS & Neoen & Neoen/NHOA & $\circ$ \\
\texttt{HPR1} & Hornsdale Power Reserve & Neoen & Tesla Autobidder & \textbullet \\
\texttt{LBB1} & Lake Bonney BESS & Iberdrola & Tesla Autobidder & \textbullet \\
\texttt{BUNGAMB1} & Bungama BESS & Amp Energy & W\"{a}rtsil\"{a} GEMS & \textbullet \\
\texttt{ADPBA1} & Adelaide Desal.\ Plant & SA Water & --- & --- \\
\texttt{BOWWBA1} & Bolivar Waste Water & SA Water & --- & --- \\
\texttt{CBWWBA1} & Christies Beach WWTP & SA Water & --- & --- \\
\texttt{CGBESS01} & Clements Gap BESS & Pacific Blue & --- & --- \\
\texttt{HVWWBA1} & Happy Valley WTP & SA Water & --- & --- \\
\texttt{TEMPB1} & Templers BESS & ZEN Energy/ZEBRE & --- & --- \\
\midrule
\multicolumn{5}{@{}l}{\textbf{VIC1}} \\
\midrule
\texttt{BALB1} & Ballarat BESS & EA/AusNet & EA intermediary & $\circ$ \\
\texttt{GANNB1} & Gannawarra ESS & EA/Edify & EA intermediary & $\circ$ \\
\texttt{HBESS1} & Hazelwood BESS & ENGIE/Eku Energy & Fluence Mosaic & \textbullet \\
\texttt{LVES1} & Latrobe Valley BESS & Tilt Renewables & Fluence Mosaic & \textbullet \\
\texttt{RANGEB1} & Rangebank BESS & Shell Energy/Eku & Fluence Mosaic & \textbullet \\
\texttt{TRGBESS1} & Terang BESS & FRV & Fluence Mosaic & \textbullet \\
\texttt{MOORABS1} & Moorabool Wind Farm & Goldwind & Goldwind in-house & $\circ$ \\
\texttt{KESSB1} & Koorangie ESS & Shell Energy & Shell Energy interm. & $\circ$ \\
\texttt{BULBES1} & Bulgana Green Power Hub & Neoen/HMC & Tesla Autobidder & $\circ$ \\
\texttt{MREHA1} & Melbourne REH A1 & Equis Energy & Tesla Autobidder & $\circ$ \\
\texttt{MREHA2} & Melbourne REH A2 & Equis Energy & Tesla Autobidder & $\circ$ \\
\texttt{MREHA3} & Melbourne REH A3 & SEC Energy & Tesla Autobidder & $\circ$ \\
\texttt{VBB1} & Victorian Big Battery & Neoen/HMC & Tesla Autobidder & \textbullet \\
\texttt{PIBESS1} & Phillip Island BESS & Mondo/AusNet & Ubi Platform & \textbullet \\
\end{longtable}

{\footnotesize \textit{Notes:} Column ``C'' denotes confidence:
$\bullet$~=~confirmed (direct public evidence naming the software platform),
$\circ$~=~probable (strong indirect evidence, typically hardware identification
combined with the manufacturer's known bundling pattern),  -- ~=~unknown
(excluded from provider-level specifications). ``EA'' abbreviates
EnergyAustralia. ``Participant'' is the asset owner or offtaker; ``Provider''
is the identified software platform operator. DUIDs follow AEMO's Dispatch
Unit Identifier convention. Multiple DUIDs sharing a station name (e.g.,
\texttt{WDBESS1}/\texttt{WDBESS2}) are separate dispatch units of the same
physical battery.}

\subsection{State-of-Charge Reconstruction}\label{app:soc-model}

We build a battery-level state-of-charge (SoC) series
at 5-minute frequency. Before the ERI reform on July~1,~2025, battery state
is not public, so SoC must be reconstructed from observable SCADA dispatch
data. After ERI, AEMO publishes end-of-interval energy storage fields with a
one-day lag, which we use both to calibrate the model and to anchor its
post-reform output. This subsection describes the stock-flow model, the
calibration procedure, and the validation results.

\paragraph{Model}
For each battery $i$ with nameplate energy capacity $E^{\max}_i$, we propagate
a 5-minute stock-flow recursion driven by SCADA charge and discharge
metering. Let $P^{\mathrm{chg}}_{it}$ and $P^{\mathrm{dis}}_{it}$ denote the
observed charging and discharging power (MW) at interval $t$, and let
$\Delta = 5/60$~h. The modelled state of charge evolves as
\begin{equation}
S_{i,t}
=
\mathrm{clip}\!\left(
\rho_i\, S_{i,t-1}
+ \eta_i^{\mathrm{ch}}\, P^{\mathrm{chg}}_{it}\,\Delta
- \frac{P^{\mathrm{dis}}_{it}\,\Delta}{\eta_i^{\mathrm{dis}}}
- \kappa_i\,\Delta,
\;\; 0,\;\; E^{\max}_i
\right),
\label{eq:soc-model}
\end{equation}
where $\eta_i^{\mathrm{ch}} \in [0.70, 0.99]$ is the charge efficiency,
$\eta_i^{\mathrm{dis}} \in [0.70, 0.99]$ the discharge efficiency,
$\rho_i \in [0.999, 1.0]$ a per-interval self-discharge factor, and
$\kappa_i \in [0.0, 2.0]$~MW an auxiliary parasitic load. The clip operator
enforces the physical bounds $S_{i,t} \in [0, E^{\max}_i]$.

\paragraph{Calibration}
The four parameters $(\eta_i^{\mathrm{ch}}, \eta_i^{\mathrm{dis}}, \rho_i, \kappa_i)$ are calibrated
separately for each battery by minimising the root mean squared error between
$S_{i,t}$ and the public post-ERI energy-storage field published by AEMO
(July~2025 through January~2026). The first seven days of each battery's
public series are excluded to avoid initial-condition bias. Optimisation uses
L-BFGS-B with the bounds stated above. Across the 42 known-provider
batteries, the calibrated charge efficiency has a median of $0.955$ (IQR
$[0.913, 0.977]$), the discharge efficiency a median of $0.990$ (most
batteries at the upper bound), the self-discharge factor a median of $0.9997$
(implying negligible inter-interval leakage), and the auxiliary load a median
of $0.20$~MW.

\paragraph{SoC series construction}
The calibrated model produces three SoC objects for each battery-interval.
\emph{Pre-ERI} (before July~1,~2025): $S_{i,t}$ is the pure model
propagation from \eqref{eq:soc-model}, initialised at half capacity. Because
no public anchor exists, pre-reform SoC is entirely model-driven.
\emph{Post-ERI}: the model is re-anchored daily at midnight to the previous
trading day's publicly released end-of-interval energy storage, then
propagated forward through the day using \eqref{eq:soc-model}. This daily
anchoring prevents drift and ensures that the modelled series reflects the
information set available to market participants at the start of each trading
day.

\paragraph{Empirical SoC variables}
We define the common-state variable $\mathrm{CommonSoC}_{rt}$ as the leave-one-out regional SoC share: for battery $i$ in region $r$ at interval $t$,
\[
\mathrm{CommonSoC}_{rt}^{(-i)}
:=
\frac{\sum_{j \in r,\, j \neq i} S_{j,t-288}}
     {\sum_{j \in r,\, j \neq i} E^{\max}_j},
\]
where the lag of 288~intervals corresponds to one trading day. The own-state
control $\mathrm{SoC}_{it} := S_{i,t-288} / E^{\max}_i$ is the
corresponding battery-level lagged share. Both variables are predetermined
with respect to current-interval bids.

\begin{table}[h]
\centering
\caption{SoC model validation: known-provider batteries}
\label{tab:soc-validation}
\small
\begin{tabular}{@{} l ccccc @{}}
\toprule
 & P10 & P25 & Median & P75 & P90 \\
\midrule
\multicolumn{6}{@{}l}{\textit{Panel A: Un-anchored full window (Jul 2025 -- Jan 2026)}} \\[3pt]
Pearson $\rho$ & 0.643 & 0.924 & 0.956 & 0.983 & 0.988 \\
$R^2$ & 0.393 & 0.846 & 0.910 & 0.964 & 0.974 \\
RMSE (\% of $E^{\max}$) & 3.1 & 5.3 & 7.2 & 10.3 & 15.0 \\
MAE (\% of $E^{\max}$) & 2.0 & 3.4 & 4.6 & 7.2 & 11.4 \\[6pt]
\multicolumn{6}{@{}l}{\textit{Panel B: Quasi-holdout (anchored Nov 1, propagated to Jan 2026)}} \\[3pt]
Pearson $\rho$ & 0.601 & 0.930 & 0.979 & 0.989 & 0.993 \\
RMSE (\% of $E^{\max}$) & 3.2 & 4.5 & 6.9 & 9.6 & 12.2 \\
\bottomrule
\end{tabular}

\vspace{4pt}
{\footnotesize \textit{Notes:} Distribution is across 42 known-provider
batteries. Panel~A runs the calibrated stock-flow model from the first
post-ERI observation without daily re-anchoring; this exercise mimics the
pre-reform backcast. Panel~B anchors the model once at November~1,~2025 and
propagates forward without re-anchoring, testing multi-month drift.
Percentiles are unweighted across batteries. Of the 42~batteries, 33 have
un-anchored Pearson~$\rho \geq 0.90$ and 32 have $R^2 \geq 0.80$. The
lower-tail batteries (P10) are predominantly recently commissioned assets with
short calibration windows or near-dormant units with minimal cycling.}
\end{table}

\paragraph{Validation}
We assess model accuracy using two exercises on the post-ERI window where
public SoC is observed. The first (``un-anchored'') runs the calibrated model
from the first ERI observation \emph{without} daily re-anchoring, mimicking
the pre-reform backcast where no public anchor exists. The second
(``quasi-holdout'') anchors the model at a single date (November~1,~2025) and
propagates forward un-anchored through January~2026, testing multi-month
accuracy with parameters trained on the full July--January window.
Table~\ref{tab:soc-validation} reports the results for the 42
known-provider batteries used in the empirical analysis.

\subsection{Reclearing-Engine Validation}\label{app:reclearing-validation}

The deviation analysis relies on \texttt{nempy} \citep{Prakash2022nempy}
to reproduce AEMO's regional dispatch outcome from the same MMS-DB and
NEMDE-XML inputs that AEMO uses at clearing time. As a direct check that
this reproduction is accurate enough for our purposes, the deviation scan
also computes a baseline clearing for every focal interval -- the engine
is run on the observed bid stack with no counterfactual edit. Comparing
this baseline price to AEMO's published regional reference price (RRP)
for the same interval is therefore an end-to-end test of the input
pipeline plus the optimiser.

Table~\ref{tab:reclearing-validation} reports the resulting price
discrepancy on the $3{,}956$ unique (interval, region) pairs in the
deviation panel. Pooled across all intervals the median absolute price
discrepancy is \$$0.05$ per MWh, and $72.5\%$ of intervals are
reproduced within \$$1$ of the published RRP, $89.5\%$ within \$$5$, and
$97.2\%$ within \$$20$. Replication is essentially exact across all four
NEM regions and across both the pre- and post-ERI windows. Accuracy
degrades, as expected, in price-cap and other extreme-price intervals
where AEMO's full constraint set is hardest to reproduce from public
inputs: among intervals with $\mathrm{RRP}>\$1{,}000$ the share within \$$20$ falls to $73\%$. These intervals are rare in the deviation panel,
are not concentrated in the evening-peak windows the conduct test
focuses on, and produce engine prices that differ from RRP by amounts
that are still small relative to the price level itself.

\begin{table}[h]
\centering
\caption{Reclearing-engine validation against AEMO's published RRP}
\label{tab:reclearing-validation}
\small
\begin{tabular}{@{} l r r r r r @{}}
\toprule
 & & & \multicolumn{3}{c}{Share with $|\Delta p|\le$} \\
\cmidrule(lr){4-6}
Cut & Intervals & Median $|\Delta p|$ & \$1 & \$5 & \$20 \\
\midrule
\multicolumn{6}{@{}l}{\textit{Panel A: pooled}} \\
All intervals & 3,956 & 0.05 & 72.5\% & 89.5\% & 97.2\% \\
\midrule
\multicolumn{6}{@{}l}{\textit{Panel B: by period}} \\
Pre-ERI & 1,335 & 0.02 & 83.1\% & 94.5\% & 97.8\% \\
Post-ERI & 2,621 & 0.11 & 67.1\% & 87.0\% & 96.9\% \\
\midrule
\multicolumn{6}{@{}l}{\textit{Panel C: by region}} \\
NSW1 & 769 & 0.05 & 71.1\% & 86.3\% & 94.1\% \\
QLD1 & 946 & 0.01 & 80.5\% & 91.3\% & 97.1\% \\
SA1 & 1,295 & 0.24 & 66.3\% & 89.7\% & 98.1\% \\
VIC1 & 946 & 0.03 & 74.2\% & 90.1\% & 98.4\% \\
\midrule
\multicolumn{6}{@{}l}{\textit{Panel D: by observed RRP magnitude (AUD/MWh)}} \\
$\le 0$ & 852 & 0.01 & 85.9\% & 95.5\% & 99.5\% \\
$0$--$50$ & 510 & 0.06 & 76.5\% & 87.6\% & 96.3\% \\
$50$--$200$ & 2,091 & 0.06 & 69.4\% & 89.7\% & 98.1\% \\
$200$--$1{,}000$ & 437 & 0.44 & 56.8\% & 81.9\% & 92.9\% \\
$>1{,}000$ & 66 & 0.00 & 72.7\% & 72.7\% & 72.7\% \\
\bottomrule
\end{tabular}
\normalsize

\vspace{4pt}
{\footnotesize \textit{Notes:} Replication accuracy of the nempy reclearing engine. For each unique (interval, region) cell in the deviation panel we feed nempy the same MMS-DB and NEMDE-XML inputs that AEMO used at clearing time and run it without any counterfactual bid edit; we then compare the engine's reproduced regional reference price to AEMO's published RRP for the same interval. $\Delta p$ is the engine price minus the published RRP, in AUD/MWh. Panel D bins the observed RRP to show where replication is tight (sub-cap intervals) and where it degrades (price-cap intervals, where AEMO's full constraint set is harder to reproduce from the public inputs).}
\end{table}

\subsection{Conduct-Test Sample Balance}\label{app:conduct-sample-balance}

Table~\ref{tab:conduct-sample-balance} reports the numerical
balance check accompanying
Figure~\ref{fig:conduct-sample-balance} in the main text. For
each variable the table reports means and medians across three
groups of Tesla- and Fluence-managed battery-intervals in
evening-peak hours (16:00--22:00). The \emph{Full} columns
cover all such intervals in the panel window. The
\emph{Eligible} columns restrict to intervals satisfying the
four structural pre-conditions under which the conduct test
has power: (i) the focal provider's near-margin set contains
more than one owner ($\mathrm{HHI}^{\mathrm{own}}_{\mathrm{within\,prov}}<1$);
(ii) the focal provider has positive near-margin share
($s^{\mathrm{prov}}>0$); (iii) the focal battery offers
positive MW within one band of the regional clearing price
(push-down feasibility); and (iv) the focal battery is itself
at or near the margin in the interval
(its marginal-position score exceeds $0.5$), which is the structural condition
that makes a push-down deviation actually change dispatch.
The \emph{Selected} columns restrict to the focal
battery-intervals that appear in the conduct-test panel. The
final column is the Kolmogorov--Smirnov $p$-value comparing
the Eligible and Selected distributions. The systematic
Full-to-Eligible step in clearing band, marginal-dispatch share,
and price reflects the structural restriction the test
requires; the Eligible-to-Selected step is small on
outcome-relevant variables, consistent with the panel being
a stratified draw on the structural pre-conditions rather
than on outcomes.

\begin{table}[h]
\centering
\caption{Conduct-test sample balance: full evening-peak panel,
eligibility set, and selected conduct panel}
\label{tab:conduct-sample-balance}
\small
\setlength{\tabcolsep}{4pt}
\begin{tabular}{lccccccc}
\toprule
 & \multicolumn{2}{c}{Full} & \multicolumn{2}{c}{Eligible} & \multicolumn{2}{c}{Selected} & \\
\cmidrule(lr){2-3}\cmidrule(lr){4-5}\cmidrule(lr){6-7}
Variable & Mean & Median & Mean & Median & Mean & Median & K--S $p$ \\
\midrule
Regional reference price (AUD/MWh) & 120 & 104 & 155 & 131 & 135 & 120 & $<\,0.001$ \\
Regional battery dispatched (MW) & 158 & 70.0 & 325 & 273 & 359 & 320 & $<\,0.001$ \\
Regional battery offered (MW) & 1,542 & 1,615 & 1,501 & 1,615 & 1,583 & 1,615 & $<\,0.001$ \\
Regional spare capacity (MW) & 1,384 & 1,325 & 1,176 & 1,099 & 1,223 & 1,178 & $<\,0.001$ \\
Active batteries in regional stack & 10.6 & 11.0 & 10.6 & 11.0 & 10.9 & 12.0 & $<\,0.001$ \\
Focal offered MW & 151 & 150 & 168 & 200 & 180 & 200 & $<\,0.001$ \\
Focal MW within $\pm 1$ band of clearing & 110 & 100 & 89.1 & 65.0 & 138 & 150 & $<\,0.001$ \\
Clearing band (1--10) & 6.09 & 6.00 & 3.47 & 3.00 & 3.44 & 3.00 & $<\,0.001$ \\
Focal-provider near-margin share $s^{\mathrm{prov}}$ & 0.403 & 0.434 & 0.441 & 0.482 & 0.435 & 0.451 & $<\,0.001$ \\
Owner HHI within focal provider & 0.594 & 0.500 & 0.513 & 0.483 & 0.457 & 0.385 & $<\,0.001$ \\
Focal dispatch share at margin & 0.162 & 0.000 & 0.896 & 1.00 & 0.978 & 1.00 & $<\,0.001$ \\
Focal near-offer share & 0.094 & 0.083 & 0.084 & 0.057 & 0.127 & 0.123 & $<\,0.001$ \\
\midrule
$N$ & \multicolumn{2}{c}{433,224} & \multicolumn{2}{c}{47,401} & \multicolumn{2}{c}{15,882} & \\
\bottomrule
\end{tabular}
\end{table}

\subsection{Conduct Estimator: Numerical Results}\label{app:conduct-table}

Table~\ref{tab:conduct-lambda-concentration} reports the point
estimates, cluster-subsampling confidence sets, and band sample
sizes underlying the $Q$-curves shown in
Figure~\ref{fig:conduct-Q-curves}. Rows are indexed by
focal-provider coalition, market, ERI period, and
$10$-percentage-point $s^{\mathrm{prov}}$ band; bands with fewer
than $30$ deviations are omitted, as in the figure. The $95\%$
confidence set is constructed by cluster subsampling at the $(i,t)$
level (Section~\ref{subsec:conduct-test}); $Q$-reduction is
$1-Q_n(\widehat\lambda)/Q_n(0)$, equal to the depth of the
corresponding curve in Figure~\ref{fig:conduct-Q-curves} expressed
as a percentage.

\begin{table}[h]
\centering
\caption{Conduct estimator by coalition, market, ERI period, and
focal-provider near-margin share band}
\label{tab:conduct-lambda-concentration}
\small
\begin{tabular}{@{} l l r r r r c r @{}}
\toprule
Period & $s^{\mathrm{prov}}$ band & Deviations & ($i$,$t$) cells & Mean $s^{\mathrm{prov}}$ & $\widehat\lambda$ & 95\% CS & $Q$-reduction \\
\midrule
\multicolumn{8}{@{}l}{\textbf{Panel A: Tesla Autobidder -- Energy deviations}} \\
\midrule
Pooled & 0--10\% & 148 & 54 & 0.08 & 1.00 & [0.42,\,1.00] & 1.0\% \\
 & 10--20\% & 441 & 160 & 0.15 & 1.00 & [0.90,\,1.00] & 0.1\% \\
 & 20--30\% & 873 & 273 & 0.24 & 1.00 & [0.87,\,1.00] & 0.5\% \\
 & 30--40\% & 1,840 & 594 & 0.36 & 1.00 & [0.95,\,1.00] & 25.2\% \\
 & 40--50\% & 8,065 & 2,721 & 0.46 & 1.00 & [0.99,\,1.00] & 27.8\% \\
 & 50--60\% & 9,111 & 3,064 & 0.55 & 1.00 & [0.99,\,1.00] & 13.8\% \\
 & 60--70\% & 3,197 & 1,070 & 0.63 & 1.00 & [0.95,\,1.00] & 20.1\% \\
 & 70--80\% & 120 & 38 & 0.73 & 0.31 & [0.00,\,0.97] & 9.2\% \\
\midrule
\multicolumn{8}{@{}l}{\textbf{Panel B: Tesla Autobidder -- FCAS deviations}} \\
\midrule
Pooled & 20--30\% & 194 & 70 & 0.24 & 0.00 & [0.00,\,0.55] & $<0.1\%$ \\
 & 30--40\% & 331 & 113 & 0.36 & 1.00 & [0.69,\,1.00] & $<0.1\%$ \\
 & 40--50\% & 151 & 48 & 0.46 & 0.00 & --- & $<0.1\%$ \\
 & 50--60\% & 72 & 24 & 0.55 & 0.00 & --- & $<0.1\%$ \\
 & 60--70\% & 36 & 11 & 0.65 & 0.00 & --- & $<0.1\%$ \\
\midrule
\multicolumn{8}{@{}l}{\textbf{Panel C: Fluence Mosaic -- Energy deviations}} \\
\midrule
Pooled & 20--30\% & 2,861 & 948 & 0.28 & 1.00 & [0.93,\,1.00] & 0.5\% \\
 & 30--40\% & 10,059 & 3,328 & 0.34 & 0.74 & [0.70,\,0.78] & 6.5\% \\
 & 40--50\% & 3,244 & 1,025 & 0.46 & 1.00 & [0.99,\,1.00] & 18.0\% \\
 & 50--60\% & 2,655 & 869 & 0.55 & 1.00 & [0.98,\,1.00] & 12.0\% \\
 & 60--70\% & 683 & 223 & 0.62 & 1.00 & [0.90,\,1.00] & 10.6\% \\
 & 70--80\% & 148 & 60 & 0.77 & 0.00 & [0.00,\,0.94] & $<0.1\%$ \\
 & 80--90\% & 133 & 60 & 0.81 & 0.00 & --- & $<0.1\%$ \\
\midrule
\multicolumn{8}{@{}l}{\textbf{Panel D: Fluence Mosaic -- FCAS deviations}} \\
\midrule
Pooled & 20--30\% & 121 & 41 & 0.28 & 0.00 & [0.00,\,0.83] & $<0.1\%$ \\
 & 30--40\% & 388 & 127 & 0.34 & 0.04 & [0.00,\,0.52] & $<0.1\%$ \\
 & 40--50\% & 276 & 86 & 0.46 & 0.00 & --- & $<0.1\%$ \\
 & 50--60\% & 272 & 87 & 0.54 & 0.00 & --- & $<0.1\%$ \\
 & 60--70\% & 81 & 25 & 0.63 & 0.00 & --- & $<0.1\%$ \\
\bottomrule
\end{tabular}
\normalsize

\vspace{4pt}
{\footnotesize \textit{Notes:} Conduct estimator $\widehat\lambda=\arg\min_{\lambda\in[0,1]}Q_n(\lambda)$ from~\eqref{eq:conduct-estimator}, fit separately by focal-provider coalition, market, period, and 10-percentage-point bands of the focal provider's near-margin MW share $s^{\mathrm{prov}}_{k_i,r_i,t}$. The sample is the conduct-test deviation panel of Section~\ref{subsec:conduct-test}, restricted to cells with at least two distinct owners inside the focal provider's near-margin set; cells with fewer than 30 rows or with no cross-cluster identifying variation are omitted. ``$Q$-reduction'' is $1-Q_n(\widehat\lambda)/Q_n(0)$, the share of the unrationalised deviation evidence at $\lambda=0$ that the conduct hypothesis at $\widehat\lambda$ accounts for; a $Q$-reduction below $1\%$ means the test is essentially flat in $\lambda$ and $\widehat\lambda$ is uninformative. 95\% confidence sets are obtained by cluster subsampling over $(i,t)$ cells and are reported as ``---'' when the cell is too small to subsample reliably. Coalitions with empty same-provider, different-owner partner sets in the region (Major-IOU and Other providers) are omitted because $\Delta\Pi_{k_i,-o_i,t}(d)\equiv 0$ leaves $\widehat\lambda$ structurally unidentified.}
\end{table}

\begin{figure}[h]
\centering
\caption{Conduct objective $Q_n(\lambda)/Q_n(0)$ by coalition, market, and near-margin share band}
\label{fig:conduct-Q-curves}
\includegraphics[width=\textwidth]{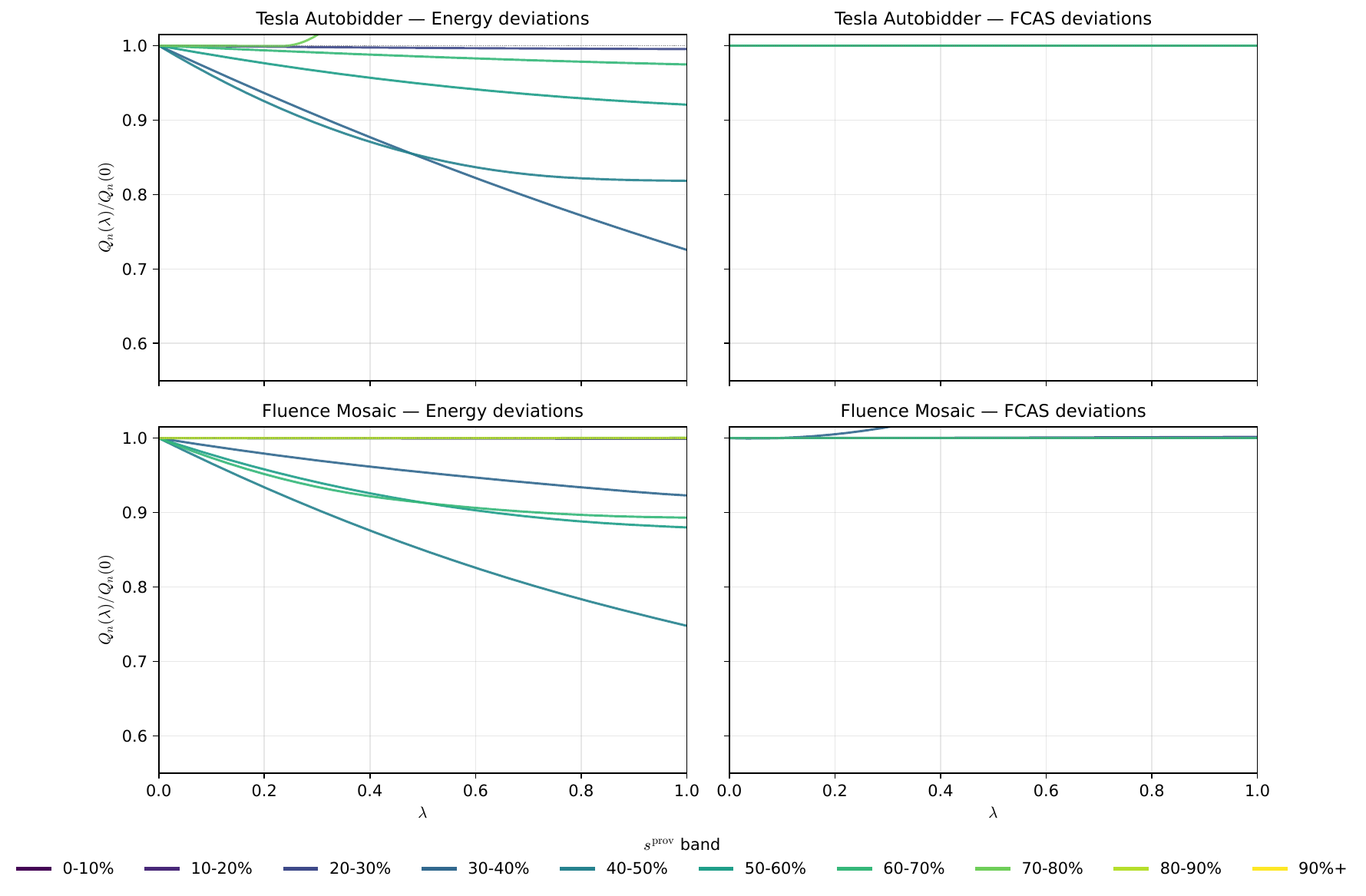}

\vspace{2pt}
{\footnotesize \textit{Notes:} Each curve plots
$Q_n(\lambda)/Q_n(0)$ for one combination of $s^{\mathrm{prov}}$
band and the pooled period within the coalition-market panel, on
the conduct-test deviation sample of
Section~\ref{subsec:conduct-test} restricted to multi-owner
intervals; colour indexes the $s^{\mathrm{prov}}$ band. Bins with
fewer than $30$ deviations are omitted. The FCAS panels (right
column) are uniformly flat: no FCAS curve in either Tesla or
Fluence is visually distinguishable from
$Q_n(\lambda)/Q_n(0)\equiv 1$. Boundary estimates of
$\widehat\lambda$ paired with near-zero $Q$-reduction are
uninformative and read as flat curves.}
\end{figure}

\subsection{Robustness: Diversity-Weighted Concentration Moderator}\label{app:conduct-Cit-robust}

The main-text fit uses the focal provider's near-margin MW share
$s^{\mathrm{prov}}_{k_i,r_i,t}$ as the concentration moderator. As
a robustness check, we re-fit the conduct estimator using the
diversity-weighted moderator
\[
C_{it}\;=\;s^{\mathrm{prov}}_{k_i,r_i,t}\,
\bigl(1-\mathrm{HHI}^{\mathrm{own}}_{k_i,r_i,t}\bigr),
\]
where $(1-\mathrm{HHI}^{\mathrm{own}})$ is the Gini-Simpson
diversity index of owner shares within the focal provider's
near-margin set. This moderator collapses the headline-share
moderator and the multi-owner sample restriction into a single
$[0,1]$-bounded index that is zero in single-owner-provider intervals
(where the cross-owner spillover channel is mechanically absent)
and rises with both larger provider near-margin presence and
greater owner dispersion within the provider.

Table~\ref{tab:conduct-lambda-Cit} reports the per-band fit using
$C_{it}$ binned into $10$-percentage-point bands. The identifying
regime is at $C_{it}\in[0.10,0.50]$ with the strongest evidence at
$C_{it}\approx 0.25$--$0.35$; under balanced two- and three-owner
splits these values correspond to focal-provider near-margin shares
of roughly $40$--$70\%$, the same regime identified by the main-text
$s^{\mathrm{prov}}$ binning. The FCAS panels are uniformly null.

\begin{table}[h]
\centering
\caption{Robustness: conduct estimator by diversity-weighted
moderator $C_{it}$}
\label{tab:conduct-lambda-Cit}
\small
\begin{tabular}{@{} l l r r r r c r @{}}
\toprule
Period & $C_{it}$ band & Deviations & ($i$,$t$) cells & Mean $C_{it}$ & $\widehat\lambda$ & 95\% CS & $Q$-reduction \\
\midrule
\multicolumn{8}{@{}l}{\textbf{Panel A: Tesla Autobidder -- Energy deviations}} \\
\midrule
Pooled & 0--10\% & 2,906 & 934 & 0.03 & 0.00 & [0.00,\,0.01] & $<0.1\%$ \\
 & 10--20\% & 1,091 & 357 & 0.17 & 1.00 & [0.93,\,1.00] & 2.5\% \\
 & 20--30\% & 11,223 & 3,756 & 0.26 & 1.00 & [1.00,\,1.00] & 27.9\% \\
 & 30--40\% & 8,411 & 2,838 & 0.34 & 1.00 & [1.00,\,1.00] & 14.8\% \\
 & 40--50\% & 1,610 & 541 & 0.42 & 0.62 & [0.35,\,0.89] & 0.7\% \\
\midrule
\multicolumn{8}{@{}l}{\textbf{Panel B: Tesla Autobidder -- FCAS deviations}} \\
\midrule
Pooled & 0--10\% & 490 & 167 & 0.04 & 1.00 & [0.79,\,1.00] & $<0.1\%$ \\
 & 10--20\% & 92 & 29 & 0.17 & 0.00 & --- & $<0.1\%$ \\
 & 20--30\% & 321 & 106 & 0.24 & 0.00 & --- & $<0.1\%$ \\
 & 30--40\% & 84 & 28 & 0.35 & 0.00 & --- & $<0.1\%$ \\
 & 40--50\% & 30 & 9 & 0.44 & 0.00 & --- & $<0.1\%$ \\
\midrule
\multicolumn{8}{@{}l}{\textbf{Panel C: Fluence Mosaic -- Energy deviations}} \\
\midrule
Pooled & 0--10\% & 1,810 & 583 & 0.01 & 1.00 & [0.94,\,1.00] & $<0.1\%$ \\
 & 10--20\% & 5,517 & 1,848 & 0.18 & 1.00 & [0.99,\,1.00] & 1.1\% \\
 & 20--30\% & 9,495 & 3,126 & 0.23 & 0.75 & [0.69,\,0.80] & 7.4\% \\
 & 30--40\% & 4,443 & 1,428 & 0.34 & 1.00 & [0.97,\,1.00] & 16.0\% \\
 & 40--50\% & 159 & 51 & 0.41 & 1.00 & [0.32,\,1.00] & 33.0\% \\
\midrule
\multicolumn{8}{@{}l}{\textbf{Panel D: Fluence Mosaic -- FCAS deviations}} \\
\midrule
Pooled & 0--10\% & 137 & 42 & 0.03 & 0.00 & --- & $<0.1\%$ \\
 & 10--20\% & 231 & 77 & 0.17 & 0.00 & [0.00,\,0.62] & $<0.1\%$ \\
 & 20--30\% & 469 & 152 & 0.24 & 0.50 & [0.18,\,0.81] & 0.8\% \\
 & 30--40\% & 354 & 113 & 0.34 & 0.00 & [0.00,\,0.41] & $<0.1\%$ \\
\bottomrule
\end{tabular}
\normalsize

\vspace{4pt}
{\footnotesize \textit{Notes:} Robustness fit of the conduct estimator using the diversity-weighted moderator $C_{it}=s^{\mathrm{prov}}_{k_i,r_i,t}\,(1-\mathrm{HHI}^{\mathrm{own}}_{k_i,r_i,t})$ instead of the headline focal-provider near-margin share $s^{\mathrm{prov}}_{k_i,r_i,t}$. $(1-\mathrm{HHI}^{\mathrm{own}})$ is the Gini-Simpson diversity index of owner shares within the provider's near-margin set; cells with a single-owner provider are excluded by construction because $C_{it}=0$ there. Conduct estimator $\widehat\lambda=\arg\min_{\lambda\in[0,1]}Q_n(\lambda)$ from~\eqref{eq:conduct-estimator}, fit separately by focal-provider coalition, market, period, and 10-percentage-point bands of $C_{it}$. ``$Q$-reduction'' is $1-Q_n(\widehat\lambda)/Q_n(0)$. 95\% confidence sets are obtained by cluster subsampling over $(i,t)$ cells and are reported as ``---'' when the cell is too small to subsample reliably.}
\end{table}

\subsection{Back-of-Envelope Cost of Identified Conduct}\label{app:conduct-cost}

This appendix gives the construction behind the back-of-envelope
estimate in Section~\ref{subsec:conduct-test}. We report two
estimates: a consumer-side price-difference estimate that uses the
recleared regional clearing price under an owner-only counterfactual
deviation and is the headline figure in the main text, and a 
battery-side diagnostic that considers only the dollar transfer on
deviations whose owner-portfolio gain flips sign once same-provider
spillovers are internalised. The construction follows the
markup logic of \citet{CalderWangKim2026}, adapted to a
uniform-price electricity auction in which the consumer side is the
regional clearing price times regional demand rather than a
demand-system markup.

\paragraph{Price-difference estimate (headline)}
For each focal battery $i$ at interval $t$ inside the identifying
range of Section~\ref{subsec:conduct-test}, define the owner-only
counterfactual deviation
\[
d^{0}_{it}\;=\;\arg\max_{d\in\{10\%,25\%,50\%\}}
\Delta\Pi_{o_i,t}(d)
\quad\text{subject to}\quad
\Delta\Pi_{o_i,t}(d)>0,
\]
that is, the discrete energy-push deviation that maximises the focal
owner's portfolio gain under the owner-only hypothesis. The
reclearing engine returns, for that deviation, a counterfactual
regional clearing price $p^{\mathrm{cf}}_{rt}(d^0_{it})$ together
with the observed baseline price $p^{\mathrm{obs}}_{rt}$. The
implied price difference
\[
\Delta p_{rt}\;=\;p^{\mathrm{obs}}_{rt}-p^{\mathrm{cf}}_{rt}(d^0_{it})
\]
is the increase in the regional clearing price that the focal
owner's coordinated restraint preserves. Because the NEM clears each
region at a single price in each interval, we collapse to the
region-interval level by selecting, among focal batteries with an
owner-only counterfactual in $(r,t)$, the focal whose $d^0$ yields
the largest non-negative price difference, and multiply once by
regional demand:
\[
\mathrm{Cost}_{rt}\;=\;\Delta p_{rt}\times \mathrm{DemandMW}_{rt}\times \tfrac{5}{60},
\]
in units of AUD. This avoids double-counting the same regional
price difference across focal batteries. The sample-window cost is
the sum of $\mathrm{Cost}_{rt}\cdot w_{rt}$ over identified
region-intervals, with $w_{rt}$ the inverse-inclusion-probability
weight of the representative $(i,t)$ observation, extrapolating
from the sample to the universe of evening-peak intervals over the
$305$-day window. The annualised row scales the weighted total by
$365/305\simeq 1.20$.

\begin{table}[h]
\centering
\caption{Consumer-side cost of identified same-provider conduct
(price-difference estimate)}
\label{tab:conduct-cost-wedge}
\footnotesize
\begin{tabular}{@{} l r r r r r r @{}}
\toprule
 & ($r$,$t$) & \multicolumn{3}{c}{$\Delta p$ (AUD/MWh)} & Sample & Annual \\
\cmidrule(lr){3-5}
Coalition & cells & Median & Mean & $p_{90}$ & (AUD) & (AUD/yr) \\
\midrule
Fluence Mosaic & 457 & 1.29 & 2.80 & 6.44 & \$1.98m & \$2.36m \\
Tesla Autobidder & 439 & 1.28 & 2.59 & 6.02 & \$2.58m & \$3.09m \\
Total & 896 & 1.28 & 2.70 & 6.24 & \$4.56m & \$5.45m \\
\bottomrule
\end{tabular}
\normalsize

\vspace{4pt}
{\footnotesize \textit{Notes:} Sample comprises the $19$ focal
batteries ($12$ Tesla Autobidder, $7$ Fluence Mosaic) of the conduct
panel, restricted to evening-peak intervals (16:00--22:00) in the
identifying bands (energy market, focal-provider near-margin share
$s^{\mathrm{prov}}$ in the identifying regime of
Table~\ref{tab:conduct-lambda-concentration}, $Q$-reduction $\ge 1\%$
and CS excluding zero). $(r,t)$ counts region-intervals with a
positive recleared price difference under the owner-only
counterfactual $d^{0}_{it}$. Price-difference percentiles are taken
over the same set. The sample-window cost is the
inverse-inclusion-probability-weighted total
$\sum_{(r,t)} \Delta p_{rt}\cdot \mathrm{DemandMW}_{rt}\cdot\tfrac{5}{60}\cdot w_{rt}$;
the annualised cost multiplies by $365/305\simeq 1.20$.}
\end{table}

The median price difference is small (about \$$1.28$/MWh) but the
distribution is fat-tailed (mean \$$2.70$/MWh, $90$th percentile around
\$$6.24$/MWh, with a maximum of approximately \$$59$/MWh). The
annualised total is \$$5.5$m, with Tesla accounting for \$$3.1$m and
Fluence for \$$2.4$m. The
estimate is bounded below the true cost of coordination for two
reasons: only deviations on the discrete $\{10\%,25\%,50\%\}$ MW grid
in the identifying bands contribute, and the per-$(r,t)$
price difference is the unilateral-counterfactual price difference,
which understates the price effect that fully coordinated
same-provider bidding could deliver.

\paragraph{Battery-side transfer (diagnostic)}
A narrower diagnostic isolates the dollar magnitude on the
\emph{sign-flipping} subset of the deviation panel. Define the
left-on-table set
\[
\mathcal{L}^\star_{(i,t)}\;=\;
\arg\max_{d}\Delta\Pi_{o_i,t}(d)
\quad\text{s.t.}\quad
\Delta\Pi_{o_i,t}(d)>0,\;
\Delta\Pi_{o_i,t}(d)+\Delta\Pi_{k_i,-o_i,t}(d)<0,
\]
that is, the $(i,t)$ intervals in which same-provider weighting
flips the sign of the focal owner's best deviation gain. On
$\mathcal{L}^\star$ the focal owner foregoes $\Delta\Pi_{o_i,t}$ in
portfolio profit, same-provider partners retain
$|\Delta\Pi_{k_i,-o_i,t}|$ of inframarginal revenue, and other
(different-provider) regional batteries retain
$|\Delta\Pi^{\mathrm{nei}}_t|$.
Table~\ref{tab:conduct-cost-boe} reports these magnitudes by
coalition.

\begin{table}[h]
\centering
\caption{Battery-side dollar transfer on sign-flipping deviations
(diagnostic)}
\label{tab:conduct-cost-boe}
\footnotesize
\begin{tabular}{@{} l r r r r r @{}}
\toprule
Coalition & $\mathcal{L}^\star$ cells & Focal foregone & Same-prov.\ retained & Other-batt.\ retained & Battery transfer \\
\midrule
\multicolumn{6}{@{}l}{\textit{Sample totals (305-day window)}} \\
Fluence Mosaic & 307 & \$10.2k & \$18.4k & \$11.1k & \$29.5k \\
Tesla Autobidder & 401 & \$18.3k & \$35.7k & \$16.2k & \$51.9k \\
Total & 708 & \$28.6k & \$54.2k & \$27.3k & \$81.5k \\
\midrule
\multicolumn{6}{@{}l}{\textit{Annualised ($\times1.20$)}} \\
Fluence Mosaic & & \$12.2k & \$22.1k & \$13.3k & \$35.4k \\
Tesla Autobidder & & \$21.9k & \$42.7k & \$19.4k & \$62.1k \\
Total & & \$34.2k & \$64.8k & \$32.7k & \$97.5k \\
\bottomrule
\end{tabular}
\normalsize

\vspace{4pt}
{\footnotesize \textit{Notes:} As Table~\ref{tab:conduct-cost-wedge}
but restricted to the sign-flipping subset $\mathcal{L}^\star$ ($708$ region-intervals in the sample) and tabulating
battery-side dollar magnitudes rather than the consumer-side price
effect. The diagnostic is intentionally narrower than
Table~\ref{tab:conduct-cost-wedge}: it considers only the discrete
deviations that flip sign between the two hypotheses, and excludes
both the consumer cost on non-battery inframarginal generation and
the magnitude-shrinking contributions from deviations that remain
profitable under both hypotheses. The annualised row scales sample
totals by $365/305\simeq 1.20$.}
\end{table}

The two estimates measure different objects. The price-difference
estimate in Table~\ref{tab:conduct-cost-wedge} is the consumer-side
cost of the higher clearing price across all regional inframarginal
MW in the identified region-intervals, while the second estimate in Table~\ref{tab:conduct-cost-boe} is the dollar size of
the battery-side transfer on the strict sign-flipping subset
$\mathcal{L}^\star$. The two differ by roughly two orders of
magnitude because the price-difference estimate applies to the
broader regional inframarginal generation (not just the same-provider
fleet), and because it counts all region-intervals with a positive
price difference while the second transfer diagnostic is limited to
$\mathcal{L}^\star$.

\subsection{Robustness of the Conduct Test}\label{app:conduct-robustness}

This appendix collects five robustness exercises for the conduct test of
Section~\ref{subsec:conduct-test}. The first three address alternative
interpretations of the conduct finding -- an unobserved owner-level cost, a
mechanical explanation based on the marginal distribution of same-provider
spillovers, and whether the effect is specific to the provider coalition rather
than to any similarly placed set of near-margin batteries. The final two probe
the construction of the test itself -- whether the finding survives when each
bid is dated to its actual submission time, and whether the estimator recovers
internalisation in a positive control where theory requires it.

\subsubsection{Omitted owner-level cost bounds}\label{app:omitted-owner-cost}

The conduct test treats a positive owner-portfolio gain
$\Delta\Pi_{o_i,t}(d)>0$ as a contradiction of owner-level conduct.
A natural concern is that the estimated owner payoff omits private
components of the owner's objective: contract positions, tolling
constraints, degradation or warranty shadow costs, risk limits, or
private forecasts. This appendix quantifies how large such an
omitted owner-level component would have to be to absorb the
conduct evidence.

Write
\[
G^O_{itd}\equiv \Delta\Pi_{o_i,t}(d),\qquad
G^K_{itd}\equiv \Delta\Pi_{k_i,-o_i,t}(d),
\]
for the focal owner's regional portfolio gain and the same-provider,
different-owner spillover from deviation $d$. The conduct estimator
asks whether $G^O_{itd}+\lambda G^K_{itd}\le 0$. Suppose instead
that the owner-only objective includes an omitted non-negative
deviation cost $\Omega_{itd}$. The owner-only inequality becomes
\[
G^O_{itd}-\Omega_{itd}\le 0 .
\]
With no structure on $\Omega_{itd}$, the minimum cost needed to
rationalise an owner-only violation is simply
$\max\{G^O_{itd},0\}$, which is tautological. We therefore report
two more disciplined benchmarks.

First, consider a band-level linear omitted cost $c$ in AUD/MWh of
extra focal discharge. Let
\[
x^+_{itd}
=
\max\{\Delta q^{\mathrm{disp}}_{itd},0\}\times \tfrac{5}{60}
\]
be the recleared increase in focal discharge, in MWh, under the
deviation. For each coalition-concentration band
$\mathcal C$, define
\[
Q^{O}_{n,\mathcal C}(c)
=
\frac{1}{n_{\mathcal C}}
\sum_{(i,t,d)\in\mathcal C}
w_{itd}
\max\{G^O_{itd}-c x^+_{itd},0\}^{2}.
\]
The first statistic is
\[
c^Q_{\mathcal C}
=
\inf\left\{
c\ge 0:
Q^{O}_{n,\mathcal C}(c)\le
Q_{n,\mathcal C}(\widehat\lambda)
\right\}.
\]
This is the smallest common throughput-style owner cost that would
reduce the owner-only contradiction objective to the value achieved
by the estimated same-provider conduct parameter. It is deliberately
favourable to the omitted-cost explanation, because it only has to
match the aggregate squared-violation improvement rather than
rationalise every individual deviation.

Second, for each focal battery-interval $(i,t)$, define the
interval-specific omitted-cost rate
\[
\mu^\star_{it}
=
\sup_{d:G^O_{itd}>0}
\frac{G^O_{itd}}{x^+_{itd}}.
\]
This is the smallest AUD/MWh charge on extra focal discharge that
would rationalise all sampled owner-profitable deviations for that
battery-interval under owner-only conduct. If any owner-profitable
sampled deviation has $x^+_{itd}=0$, then a discharge-denominated
omitted cost cannot rationalise that deviation and
$\mu^\star_{it}=\infty$.

The table reports two internal benchmarks. The column $c^Q/11$
divides the band-level cost by the $11$ AUD/MWh maximum movement
in $\widehat m_{it}$ across the Bellman grid/control sensitivity
checks of Appendix~\ref{app:bellman-estimation}; a value above one
means the required omitted owner cost exceeds anything that Bellman
discretisation can plausibly produce. The last two columns report
the share of owner-violating battery-intervals for which
$\mu^\star_{it}$ exceeds \$$100$/MWh and \$$1{,}000$/MWh, the same
materiality cutoffs used in the Bellman incentive-compatibility
diagnostics.

\begin{table}[h]
\centering
\caption{Omitted owner-cost bounds for conduct-identifying energy
deviations}
\label{tab:conduct-omitted-owner-wedge}
\scriptsize
\begin{tabular}{@{} l l l r r r r r r r @{}}
\toprule
Coalition & Period & $s^{\mathrm{prov}}$ & ($i,t$) & $c^{Q}$ & $c^{Q}/11$ & Med. $\mu^\star$ & p90 fin. $\mu^\star$ & $\mu^\star>100$ & $\mu^\star>1{,}000$ \\
 & & band & cells & & & & & share & share \\
\midrule
Tesla & Pooled & 30--40\% & 594 & 16 & 1.4 & 25 & 77 & 36.4\% & 32.0\% \\
Tesla & Pooled & 40--50\% & 2,721 & -- & -- & 24 & 92 & 79.7\% & 78.2\% \\
Tesla & Pooled & 50--60\% & 3,064 & 59 & 5.4 & 27 & 105 & 77.1\% & 74.5\% \\
Tesla & Pooled & 60--70\% & 1,070 & -- & -- & 20 & 79 & 57.7\% & 55.0\% \\
Fluence & Pooled & 30--40\% & 3,328 & 8 & 0.7 & 24 & 89 & 71.7\% & 69.4\% \\
Fluence & Pooled & 40--50\% & 1,025 & 28 & 2.6 & 21 & 70 & 48.8\% & 45.9\% \\
Fluence & Pooled & 50--60\% & 869 & 15 & 1.4 & 23 & 121 & 62.2\% & 54.1\% \\
Fluence & Pooled & 60--70\% & 223 & 11 & 1.0 & 14 & 67 & 62.5\% & 59.4\% \\
\bottomrule
\end{tabular}
\normalsize

\vspace{4pt}
{\footnotesize \textit{Notes:} The table reports the omitted-owner-cost exercise for conduct-identifying energy-push cells: $Q$-reduction at $\widehat\lambda$ is at least one percent and the cluster-subsampling confidence set excludes zero. $c^Q$ is the smallest cell-level omitted owner cost in AUD/MWh of extra focal discharge that would reduce the owner-only squared-violation objective to $Q_n(\widehat\lambda)$. The column $c^Q/11$ benchmarks this cost against the $11$ AUD/MWh maximum movement in $\widehat m_{it}$ across Bellman grid/control sensitivity checks reported in Appendix~\ref{app:bellman-estimation}. $\mu^\star_{it}$ is the battery-interval linear wedge in AUD/MWh of extra focal discharge needed to rationalise all sampled owner-profitable deviations in that battery-interval under owner-only conduct. Quantiles of $\mu^\star_{it}$ are over finite values. Threshold shares are among owner-violating battery-intervals and count $\mu^\star_{it}=\infty$ as above threshold; the $100$ and $1{,}000$ AUD/MWh thresholds match the materiality cutoffs used in the Bellman IC diagnostics. Uniform weights match the conduct estimator.}
\end{table}

For Tesla Autobidder, two of the four identifying bands
($40$--$50\%$ and $60$--$70\%$) require an infinite uniform charge.
On these bands the positive-own-gain mass is supported by deviations
with no corresponding focal-discharge exposure, so no per-MWh charge
on extra discharge can reduce the moment objective below
$Q_n(\widehat\lambda)$. The remaining two Tesla bands require
\$$16$/MWh ($c^Q/11=1.4$) at $30$--$40\%$ and \$$59$/MWh
($c^Q/11=5.4$) at $50$--$60\%$. For Fluence Mosaic, the three corner
bands at $\widehat\lambda=1$ require \$$11$, \$$15$, and \$$28$/MWh
($c^Q/11$ from $1.0$ to $2.6$); only the $30$--$40\%$ interior
band, the weakest in the identifying set, sits below the Bellman
sensitivity benchmark at \$$8$/MWh ($c^Q/11=0.7$).

The battery-interval benchmark $\mu^\star_{it}$ delivers the same
verdict. Across the eight identifying bands, the share of
owner-violating battery-intervals individually requiring an omitted
cost above \$$100$/MWh ranges from $37\%$ (Tesla $30$--$40\%$) to
$80\%$ (Tesla $40$--$50\%$), and the share above \$$1{,}000$/MWh
ranges from $32\%$ to $78\%$. The bands with infinite $c^Q$ have
the highest battery-interval shares: a substantial fraction of the
deviations the conduct test flags are not rationalisable by any
finite per-MWh discharge cost.

The omitted-cost alternative therefore cannot account for the
conduct evidence at any reasonable benchmark. The two Tesla bands
with infinite required charge are decisive on their own, and the
remaining bands need omitted costs an order of magnitude above the
Bellman discretisation margin. The single band where the implied
charge sits below the benchmark, Fluence $30$--$40\%$, is also the
band where $\widehat\lambda$ is interior at $0.74$ rather than at
the corner, and where the conduct signal is weakest.

The bounds in this appendix are denominated in AUD per MWh of 
extra focal discharge and therefore do not cover owner objectives 
denominated in price exposure, such as inframarginal generation or 
contract positions that gain when the regional price stays high. 
The matched-placebo exercise of Appendix~\ref{app:conduct-matched-placebo} 
addresses these directly: any objective that scales with the 
recleared price change would be proxied equally well by a 
capacity-matched different-provider coalition, 
yet the same-provider coalition fits strictly better in every draw.

\subsubsection{Permutation placebo}\label{app:conduct-permutation-placebo}

This appendix tests whether the per-deviation pairing between the
focal owner's deviation gain $G^O_{itd}$ and the same-provider,
different-owner spillover $G^K_{itd}$ has identifying content beyond
the marginal distribution of $G^K$. We draw a placebo null by,
within each provider coalition (Tesla / Fluence), randomly
reassigning $G^K_{itd}$ across focal deviations. This shuffle
preserves the marginal distribution of the spillover within
coalition but destroys its per-deviation joint distribution with
the focal's deviation gain. On each of $B=300$ draws we refit the
binned $\widehat\lambda$ by $(\mathrm{coalition}, s^{\mathrm{prov}})$
band and record the resulting point estimate and $Q$-reduction.
We report empirical $p$-values with the standard plus-one
finite-sample correction,
$p=(1+\#\{b: \mathrm{stat}_b\ge \mathrm{stat}_{\mathrm{actual}}\})/(B+1)$.

\begin{table}[h!]
\centering
\caption{Permutation-null vs actual for the conduct estimator, by
$(\mathrm{coalition}, s^{\mathrm{prov}})$ band}
\label{tab:conduct-permutation-placebo}
\begin{tabular}{llrrrrrrr}
\toprule
 & & & \multicolumn{3}{c}{$\widehat\lambda$} & \multicolumn{3}{c}{$Q$-reduction (\%)} \\
\cmidrule(lr){4-6}\cmidrule(lr){7-9}
Coalition & $s^{\mathrm{prov}}$ (\%) & $n$ & Actual & Null p95 & $p$ & Actual & Null p95 & $p$ \\
\midrule
Tesla   & $<\!30$  & 1{,}462 & 1.00 & 1.00 & 0.641 &  0.4 &  2.35 & 0.485 \\
        & 30--40   & 1{,}840 & 1.00 & 1.00 & 0.578 & 25.2 &  5.12 & 0.003 \\
        & 40--50   & 8{,}065 & 1.00 & 1.00 & 0.837 & 27.8 &  4.19 & 0.003 \\
        & 50--60   & 9{,}111 & 1.00 & 1.00 & 0.983 & 13.8 &  3.61 & 0.003 \\
        & 60--70   & 3{,}197 & 1.00 & 1.00 & 0.771 & 20.1 &  1.49 & 0.003 \\
        & $>\!70$  &     123 & 0.31 & 1.00 & 0.495 &  9.2 & 44.17 & 0.282 \\
\midrule
Fluence & $<\!30$  & 2{,}876 & 1.00 & 0.66 & 0.023 &  0.5 &  0.37 & 0.037 \\
        & 30--40   &10{,}059 & 0.74 & 0.17 & 0.003 &  6.5 &  0.63 & 0.003 \\
        & 40--50   & 3{,}244 & 1.00 & 0.26 & 0.003 & 18.0 &  0.78 & 0.003 \\
        & 50--60   & 2{,}655 & 1.00 & 0.28 & 0.003 & 12.0 &  0.92 & 0.003 \\
        & 60--70   &     683 & 1.00 & 0.46 & 0.007 & 10.6 &  3.44 & 0.003 \\
        & $>\!70$  &     281 & 0.00 & 0.65 & 1.000 &  0.0 &  5.79 & 1.000 \\
\bottomrule
\end{tabular}

\vspace{2pt}
{\footnotesize \textit{Notes:} $B=300$ permutation draws shuffling
$G^K_{itd}$ across focal deviations within coalition. ``Null p95''
is the $95$th percentile of the null statistic across draws; $p$
is the empirical upper-tail $p$-value with plus-one finite-sample
correction.}
\end{table}

For Fluence Mosaic, both $\widehat\lambda$ and the $Q$-reduction
lie outside the null distribution in every identifying band
($30$--$70\%$) at $p\le 0.007$. For Tesla Autobidder, the
$\widehat\lambda$ corner is preserved under shuffling because
Tesla's spillover marginal is sufficiently mean-negative that
random pairings still pin $\lambda$ at one. The $\widehat\lambda$
statistic alone has no power against the marginal-distribution
explanation for Tesla; the $Q$-reduction comparison does. The
actual $Q$-reduction lies at $13.8$--$27.8\%$ across the four
Tesla identifying bands while the null $95$th percentile sits at
$1.5$--$5.1\%$ and the empirical $p$-value is $0.003$ uniformly.
Outside the identifying range
($s^{\mathrm{prov}}<30\%$ and $s^{\mathrm{prov}}>70\%$) the actual
$Q$-reduction is at most $0.5\%$. The Tesla comparisons lie inside
the null distribution; the one nominal rejection, Fluence below
$30\%$, concerns a fit improvement an order of magnitude smaller
than in the identifying bands.

\begin{figure}[h!]
\centering
\caption{Permutation placebo: actual vs null}
\label{fig:conduct-permutation-placebo}
\includegraphics[width=\textwidth]{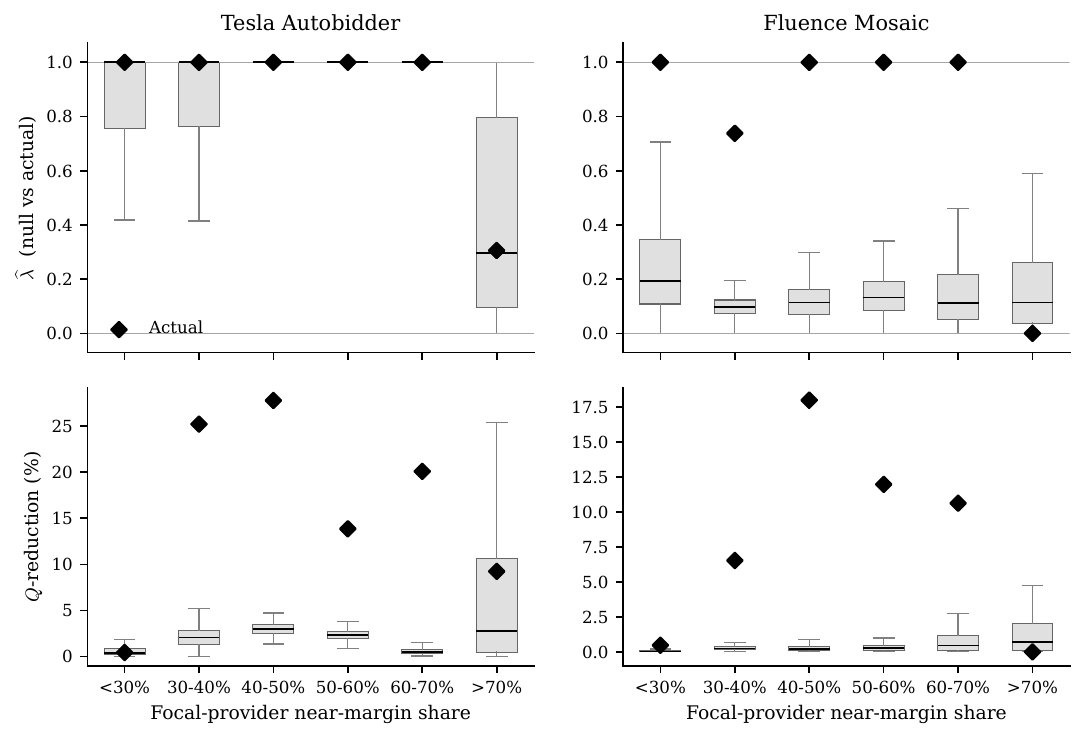}

\vspace{2pt}
{\footnotesize \textit{Notes:} Box-plots of the $B=300$ permutation
null for $\widehat\lambda$ (top row) and $Q$-reduction (bottom
row), with the actual point estimate overlaid as a solid diamond.
Within each coalition (Tesla / Fluence) the same-provider spillover
$G^K_{itd}$ is randomly reassigned across focal deviations on each
draw; whiskers extend to the full range of the null.}
\end{figure}

\subsubsection{Provider-specificity: matched placebo coalitions}\label{app:conduct-matched-placebo}

The permutation placebo shows that the per-deviation pairing matters, but not
that provider identity does -- shuffling also destroys the matching on price
impact and coalition size. This exercise tests provider identity directly.
Using the recleared profit change that each deviation imposes on every regional
battery, we form placebo coalitions from the different-provider, different-owner
batteries near the margin in the same region-interval, matched to the actual
same-provider coalition on battery count, and separately on both count and
installed capacity. We then compare the actual coalition's conduct fit
improvement to a null of $B=500$ matched placebo coalitions.
Table~\ref{tab:conduct-matched-placebo} reports the result. For both providers
the actual same-provider fit improvement exceeds that of every count- and
capacity-matched placebo draw ($p\le 0.010$). Capacity matching raises the
placebo baseline -- as it must, since the matched batteries are then comparably
exposed to the price change -- yet the same-provider coalition still rationalises
the observed restraint significantly better. The signal is therefore specific to
the provider coalition, not a generic consequence of any similarly sized and
similarly exposed set of near-margin batteries losing revenue when the price
falls.

\begin{table}[h]
\centering
\caption{Provider-specificity: matched placebo coalitions}
\label{tab:conduct-matched-placebo}
\begin{tabular}{l r r cc cc}
\toprule
 & & Same-provider & \multicolumn{2}{c}{Count-matched placebo} & \multicolumn{2}{c}{Capacity-matched placebo} \\
\cmidrule(lr){4-5}\cmidrule(lr){6-7}
Coalition & $n$ & $Q$-reduction & null mean [p95] & $p$ & null mean [p95] & $p$ \\
\midrule
Tesla & 21{,}815 & 19.6\% & 10.5 [13.2] & 0.002 & 13.2 [14.3] & 0.002 \\
Fluence & 16{,}151 & 9.7\% & 3.2 [6.7] & 0.010 & 2.1 [4.8] & 0.002 \\
\bottomrule
\end{tabular}

\vspace{4pt}
{\footnotesize \textit{Notes:} For each focal battery-interval in the identifying range ($s^{\mathrm{prov}}\in[0.30,0.70]$, multi-owner), the placebo coalition is drawn from the different-provider, different-owner batteries near the margin in the same region-interval, matched to the actual same-provider coalition on battery count (``count-matched'') or on both count and installed capacity via a Gaussian kernel (``capacity-matched''). Entries are the conduct fit improvement $R=1-Q_n(\widehat\lambda)/Q_n(0)$; the null summarises 500 placebo draws and $p$ is the fraction with $R\ge$ the actual same-provider value (plus-one correction). The same-provider coalition rationalises the observed restraint significantly better than matched different-provider coalitions with the same mechanical price exposure.}
\end{table}

\subsubsection{Lead-time split}\label{app:conduct-leadtime}

The reclearing engine values each deviation at the realised market outcome,
whereas the no-profitable-deviation restriction concerns the bidder's expected
profit given the information available when the operative bid was submitted. An
apparent owner-level violation could therefore reflect hindsight rather than
restraint. We date each interval's operative quantity bid to its submission
timestamp and split the identifying-range estimator by the lead time to
dispatch. The bulk of operative rebids are lodged close to dispatch (median
$7.6$ minutes). Table~\ref{tab:conduct-leadtime} shows that the
same-provider $Q$-reduction is concentrated in the bids submitted
within fifteen minutes of dispatch, where five-minute predispatch
forecasts are accurate and the realised outcome closely approximates
the bidder's information set, and vanishes for bids lodged more than
an hour ahead. For Tesla the share of owner-profitable deviations
falls with lead time in the same way. A hindsight artefact would produce the opposite pattern,
strongest for the longest-lead bids.

\begin{table}[h]
\centering
\caption{Conduct estimator by quantity-rebid lead time}
\label{tab:conduct-leadtime}
\begin{tabular}{llrrrr}
\toprule
Coalition & Rebid lead time & Deviations & Own-profitable (\%) & $\widehat\lambda$ & $Q$-reduction \\
\midrule
Tesla & $<15$ min & 20{,}125 & 42.8 & 1.00 & 19.8\% \\
 & $15$--$60$ min & 803 & 14.1 & 1.00 & 2.6\% \\
 & $>60$ min & 1{,}285 & 10.7 & 0.00 & 0.0\% \\
\midrule
Fluence & $<15$ min & 13{,}602 & 35.0 & 0.96 & 11.2\% \\
 & $15$--$60$ min & 2{,}649 & 28.6 & 1.00 & 4.9\% \\
 & $>60$ min & 390 & 36.4 & 1.00 & 0.0\% \\
\bottomrule
\end{tabular}

\vspace{4pt}
{\footnotesize \textit{Notes:} Conduct estimator $\widehat\lambda=\arg\min_{\lambda\in[0,1]}Q_n(\lambda)$ with a single same-provider, different-owner peer term, fit on the evening-peak energy deviation panel restricted to the identifying near-margin range $s^{\mathrm{prov}}\in[0.30,0.70]$ with at least two owners inside the focal provider's near-margin set. Rows split the sample by the number of minutes between the operative quantity rebid and dispatch. ``Own-profitable'' is the share of deviations with a positive recleared owner-portfolio gain; ``$Q$-reduction'' is $1-Q_n(\widehat\lambda)/Q_n(0)$. Under a hindsight alternative the signal would concentrate in long-lead bids; instead it concentrates in the near-dispatch bids, where predispatch foresight is accurate.}
\end{table}

\subsubsection{Same-owner positive control}\label{app:conduct-same-owner}

An owner should fully internalise the price effect of a deviation on its own
other batteries, so the conduct weight on same-owner, different-battery profits
should equal one. As a positive control we estimate that weight -- rather than
imposing it, as the headline test does -- on the battery-intervals in which the
focal owner operates at least one other battery in the region.
Table~\ref{tab:conduct-same-owner} reports $\widehat\lambda_O=1.00$ for Tesla
Autobidder with a $12$ percent improvement in fit, confirming that the
estimator detects joint-profit internalisation when it is present by
construction. The control is uninformative for Fluence Mosaic, whose owners
rarely operate more than one battery in the same region, so same-owner
spillovers are almost never present.

\begin{table}[h]
\centering
\caption{Same-owner positive control}
\label{tab:conduct-same-owner}
\begin{tabular}{lrrrr}
\toprule
Coalition & Battery-intervals & With same-owner peer & $\widehat\lambda_O$ & $Q$-reduction \\
\midrule
Tesla & 25{,}253 & 73.2\% & 1.00 & 12.4\% \\
Fluence & 21{,}424 & 11.8\% & 1.00 & 0.3\% \\
\bottomrule
\end{tabular}

\vspace{4pt}
{\footnotesize \textit{Notes:} Positive control estimating the weight $\widehat\lambda_O$ the focal owner places on the recleared profit effect on its \emph{own} other batteries in the region, $\arg\min_{\lambda_O}\sum w[\Delta\Pi^{\mathrm{focal}}+\lambda_O\Delta\Pi^{\text{same owner}}]_+^2$, on the evening-peak energy deviation panel restricted to focal battery-intervals with at least one same-owner peer in the region. Economic theory requires $\lambda_O=1$. ``With same-owner peer'' is the share of each coalition's battery-intervals that have such a peer.}
\end{table}

\subsection{Installed vs Near-Margin Provider Concentration}
\label{app:installed-vs-nearmargin}

The conduct test in Section~\ref{subsec:conduct-test} uses the
focal provider's per-interval near-margin MW share
$s^{\mathrm{prov}}_{k,r,t}$. 
Installed-capacity shares are the more familiar way to express 
a provider's footprint, and are the basis for the static index 
reported in Section~\ref{subsec:empirical-panel}. This
appendix documents that the two concentration objects are tightly
linked in our sample, so the conduct test's identifying near-margin
range can be mapped back to a range of installed-capacity shares.

Figure~\ref{fig:installed-vs-nearmargin} plots, for each region,
the per-interval near-margin share against the installed-capacity
provider share for every provider with non-zero installed
capacity. For each provider-region pair we mark the median
per-interval near-margin share and the $10$--$90\%$ band. Three
patterns are visible. First, the median per-interval near-margin
share tracks the installed share closely: for the two headline
providers, the median near-margin share lies within roughly $5$
percentage points of the installed share in every region. Second,
per-interval dispersion around the installed-share mean reflects
which batteries are dispatchable near the margin in that interval,
and varies across provider-region pairs from tight (QLD1 Tesla:
inter-decile range $\approx 17$ percentage points) to wide (NSW1
Fluence: inter-decile range $\approx 47$ percentage points).
Third, providers with installed shares above roughly $20\%$ spend
the majority of evening-peak intervals in the conduct-test
identifying near-margin band $s^{\mathrm{prov}}\in[0.30, 0.70]$:
$98\%$ of Tesla's QLD1 intervals, $80\%$ of Fluence's VIC1
intervals, $74\%$ of Tesla's VIC1 intervals, and $61\%$ of
Fluence's NSW1 intervals. The conduct test's identifying
near-margin range thus maps on average to installed-capacity shares
of roughly $20$--$60\%$.

\begin{figure}[h!]
\centering
\caption{Installed-capacity vs per-interval near-margin provider share}
\label{fig:installed-vs-nearmargin}
\includegraphics[width=\textwidth]{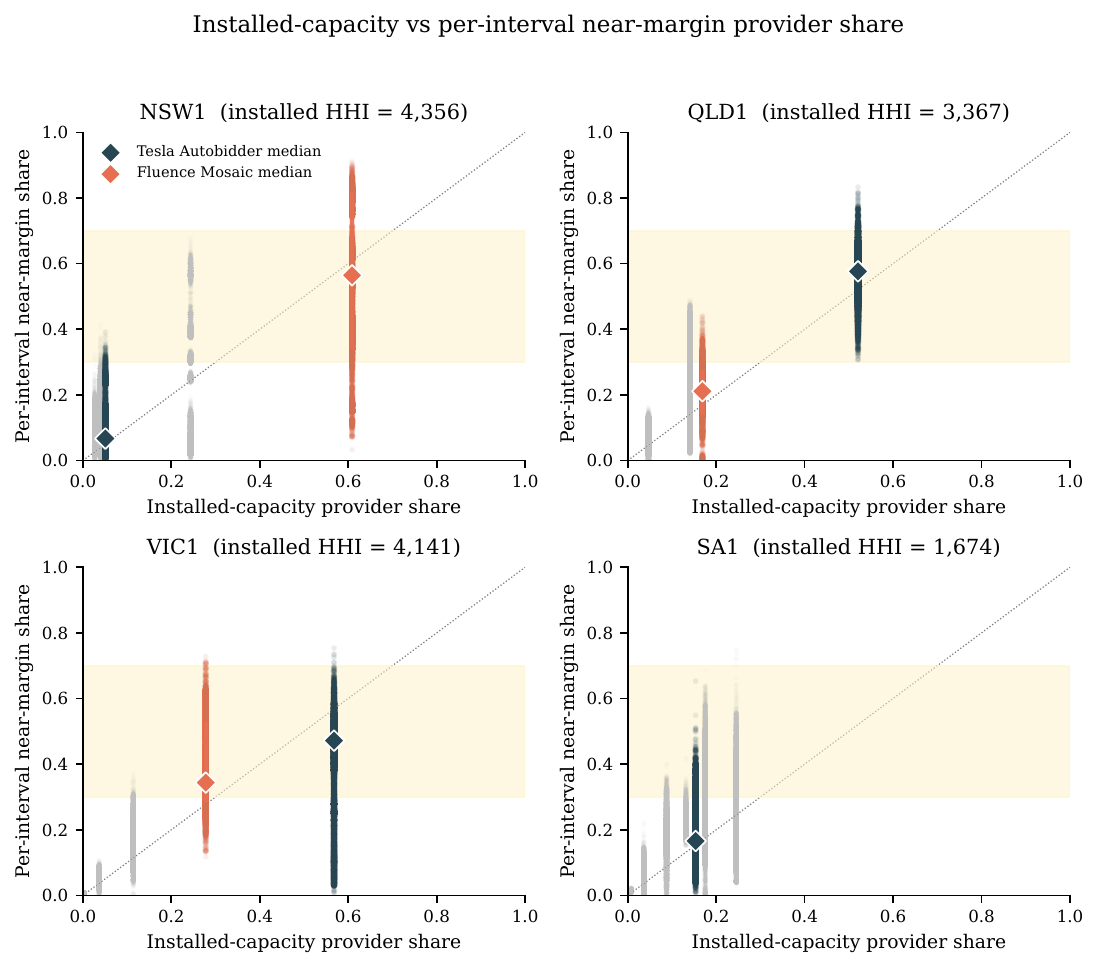}

\vspace{2pt}
{\footnotesize \textit{Notes:} One panel per NEM region.
Horizontal axis: provider's MW-weighted installed-capacity share.
Vertical axis: per-interval near-margin MW share averaged over all
evening-peak intervals in the sample. Coloured diamonds mark the
median per-interval near-margin share for Tesla Autobidder (dark
teal) and Fluence Mosaic (orange); the vertical line spans the
$10$--$90\%$ inter-decile range. Grey points are all other
providers in that region. The dashed $45^{\circ}$ line marks
equality of installed and near-margin shares. The shaded
horizontal band marks the conduct-test identifying range
$s^{\mathrm{prov}}\in[0.30, 0.70]$.}
\end{figure}

\subsection{FCAS Trapezium and Technical Constraints}\label{app:fcas-trapezium}

FCAS enablement is feasible only conditional on the battery's energy operating
point. This appendix gives the compact version of the AEMO trapezium constraint
used to interpret the bid problem and the counterfactual reclearing exercise
\citep{AEMO2021FCASNEMDE,KarimiArpanahiPourmousaviMahdavi2024}. Let
\[
x_{it}:=q_{it}-r_{it}
\]
denote the battery's net energy operating point in interval $t$, where
$x_{it}>0$ denotes discharge and $x_{it}<0$ denotes charging. For FCAS service
$j$, let $f_{ijt}$ denote enabled FCAS MW. AEMO validates $f_{ijt}$ against a
service-specific trapezium in the energy operating point. Figure~\ref{fig:fcas-trapezium}
shows both the service-level validation object and the joint capacity-allocation
logic emphasized by \citet{KarimiArpanahiPourmousaviMahdavi2024}: energy,
regulation FCAS, and contingency FCAS all draw on the same physical headroom.

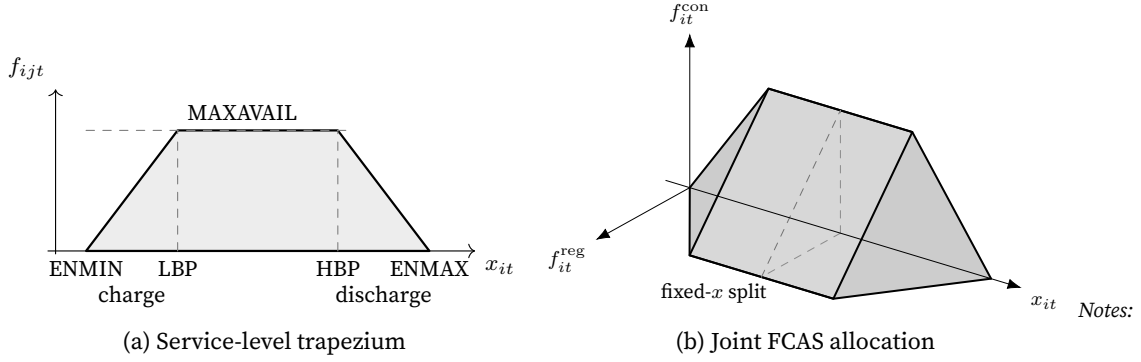
\begin{figure}[t]
\centering
\caption{Schematic FCAS trapezium constraints}
\label{fig:fcas-trapezium}
\begin{subfigure}[t]{0.47\textwidth}
\centering
\resizebox{\linewidth}{!}{%
\begin{tikzpicture}[font=\footnotesize,line cap=round,line join=round,
    x=0.96cm,y=0.74cm]
        \draw[->] (-2.75,0) -- (2.85,0) node[below right] {$x_{it}$};
        \draw[->] (-2.65,-0.05) -- (-2.65,2.75) node[above left] {$f_{ijt}$};
        \draw[thick,fill=gray!14]
            (-2.25,0) -- (-1.05,2.05) -- (1.05,2.05) -- (2.25,0) -- cycle;
        \draw[dashed,black!50] (-1.05,0) -- (-1.05,2.05);
        \draw[dashed,black!50] (1.05,0) -- (1.05,2.05);
        \draw[dashed,black!50] (-2.25,2.05) -- (1.25,2.05);
        \node[below] at (-2.25,0) {ENMIN};
        \node[below] at (-1.05,0) {LBP};
        \node[below] at (1.05,0) {HBP};
        \node[below] at (2.25,0) {ENMAX};
        \node[above] at (-0.2,2.05) {MAXAVAIL};
        \node[below] at (-1.65,-0.42) {charge};
        \node[below] at (1.65,-0.42) {discharge};
\end{tikzpicture}
}
\caption{Service-level trapezium}
\label{subfig:fcas-trapezium-service}
\end{subfigure}
\hfill
\begin{subfigure}[t]{0.47\textwidth}
\centering
\resizebox{\linewidth}{!}{%
\begin{tikzpicture}[font=\footnotesize,line cap=round,line join=round,
    x={(0.96cm,-0.29cm)},y={(-0.54cm,-0.30cm)},z={(0cm,0.84cm)}]
    \coordinate (xmin) at (-2.2,0,0);
    \coordinate (lbp) at (-1.05,0,0);
    \coordinate (hbp) at (1.05,0,0);
    \coordinate (xmax) at (2.2,0,0);
    \coordinate (lr) at (-1.05,2.05,0);
    \coordinate (hr) at (1.05,2.05,0);
    \coordinate (lc) at (-1.05,0,2.05);
    \coordinate (hc) at (1.05,0,2.05);

    \draw[black!18] (-2.2,0,0) -- (-1.05,2.05,0) -- (1.05,2.05,0) -- (2.2,0,0);
    \draw[black!18] (-2.2,0,0) -- (-1.05,0,2.05) -- (1.05,0,2.05) -- (2.2,0,0);
    \draw[black!18] (-1.05,2.05,0) -- (-1.05,0,2.05);
    \draw[black!18] (1.05,2.05,0) -- (1.05,0,2.05);

    \fill[gray!16,opacity=0.92] (xmin) -- (lr) -- (hr) -- (xmax) -- cycle;
    \fill[gray!24,opacity=0.82] (xmin) -- (lc) -- (hc) -- (xmax) -- cycle;
    \fill[gray!34,opacity=0.82] (lr) -- (hr) -- (hc) -- (lc) -- cycle;
    \fill[gray!42,opacity=0.88] (xmin) -- (lr) -- (lc) -- cycle;
    \fill[gray!42,opacity=0.88] (hr) -- (xmax) -- (hc) -- cycle;

    \draw[thick] (xmin) -- (lr) -- (hr) -- (xmax);
    \draw[thick] (xmin) -- (lc) -- (hc) -- (xmax);
    \draw[thick] (lr) -- (lc) -- (hc) -- (hr) -- cycle;
    \draw[dashed,black!55] (0,2.05,0) -- (0,0,2.05);
    \draw[dashed,black!45] (0,0,0) -- (0,2.05,0);
    \draw[dashed,black!45] (0,0,0) -- (0,0,2.05);

    \draw[-{Latex[length=2mm]}] (-2.55,0,0) -- (2.65,0,0) node[below right] {$x_{it}$};
    \draw[-{Latex[length=2mm]}] (-2.2,0,0) -- (-2.2,2.45,0) node[below left] {$f^{\mathrm{reg}}_{it}$};
    \draw[-{Latex[length=2mm]}] (-2.2,0,0) -- (-2.2,0,2.58) node[above] {$f^{\mathrm{con}}_{it}$};

    \node[below] at (-.50,2.4,0) {fixed-$x$ split};
\end{tikzpicture}
}
\caption{Joint FCAS allocation}
\label{subfig:fcas-trapezium-joint}
\end{subfigure}
\vspace{2pt}
{\footnotesize \textit{Notes:} Panel~\subref{subfig:fcas-trapezium-service}
shows the service-specific NEM validation object: enabled FCAS MW must lie
below a piecewise-linear availability function of the energy operating
point. Panel~\subref{subfig:fcas-trapezium-joint} illustrates the joint
energy-FCAS constraint for a battery. The dashed triangle is a
fixed-operating-point cross-section: regulation and contingency services
are separate markets, but their enabled quantities draw on the same
physical headroom.}
\end{figure}

Write the AEMO trapezium parameters as
\[
\underline x_{ij}:=\mathrm{ENMIN}_{ij},\quad
\ell_{ij}:=\mathrm{LBP}_{ij},\quad
u_{ij}:=\mathrm{HBP}_{ij},\quad
\overline x_{ij}:=\mathrm{ENMAX}_{ij},\quad
\overline f_{ij}:=\mathrm{MAXAVAIL}_{ij}.
\]
The market trapezium imposes
\begin{equation}
0\le f_{ijt}\le T_{ij}(x_{it}),
\label{eq:fcas-trapezium-bound}
\end{equation}
where
\begin{equation}
T_{ij}(x)=
\begin{cases}
0, & x<\underline x_{ij},\\[2pt]
\overline f_{ij}\dfrac{x-\underline x_{ij}}{\ell_{ij}-\underline x_{ij}},
& \underline x_{ij}\le x<\ell_{ij},\\[9pt]
\overline f_{ij}, & \ell_{ij}\le x\le u_{ij},\\[2pt]
\overline f_{ij}\dfrac{\overline x_{ij}-x}{\overline x_{ij}-u_{ij}},
& u_{ij}<x\le \overline x_{ij},\\[9pt]
0, & x>\overline x_{ij}.
\end{cases}
\label{eq:fcas-trapezium-function}
\end{equation}
The flat part of the trapezium is the operating range in which the unit can
provide the full enabled quantity for service $j$. The two sloped sides capture
the fact that FCAS availability falls as the unit approaches the edge of its
feasible charge or discharge range.

For batteries, the market trapezium is not the only relevant restriction. Let
$F^R_{it}$ and $F^L_{it}$ denote the raise and lower response MW that must be
physically deliverable at the operating point. A compact representation of the
headroom constraints is
\begin{equation}
F^R_{it}\le \bar P_i^{\mathrm{dis}}-x_{it},
\qquad
F^L_{it}\le x_{it}+\bar P_i^{\mathrm{ch}},
\label{eq:fcas-power-headroom}
\end{equation}
where $\bar P_i^{\mathrm{dis}}$ and $\bar P_i^{\mathrm{ch}}$ are the discharge
and charge power limits. Raise FCAS requires upward headroom: the battery must
be able to increase net injection or reduce net load. Lower FCAS requires
downward headroom: the battery must be able to reduce net injection or increase
net load.

FCAS also uses scarce energy or scarce empty room when called. For a service
with delivery duration $\Delta_j$, a conservative energy-duration restriction is
\begin{equation}
\frac{F^R_{it}\Delta_j}{\eta_i^{\mathrm{dis}}}
\le
S_{it}-\underline S_i,
\qquad
F^L_{it}\eta_i^{\mathrm{ch}}\Delta_j
\le
\bar S_i-S_{it}.
\label{eq:fcas-energy-headroom}
\end{equation}
These inequalities show why FCAS belongs in the same dynamic problem as energy
dispatch. Raise services consume stored energy if activated; lower services
consume empty room. A change in the energy stack therefore changes the FCAS
quantities that can be feasibly enabled, and a change in FCAS enablement changes
the energy headroom left for dispatch. 

\end{document}